\documentclass{amsart}

\usepackage{amsbsy}
\usepackage{amsaddr}

\usepackage{amsmath,amsthm}
\usepackage{mathrsfs}
\usepackage{amscd,amssymb, amsfonts, verbatim,subfigure, enumerate}
\usepackage[mathcal]{eucal}
\usepackage[super]{nth}

\usepackage[pdftex]{graphicx}

\usepackage{mathpazo}

\usepackage{color,slashed}

\newcommand{\wbar}{\br{w}} 

\renewcommand{\sl}{\mathfrak{sl}}

\newcommand{\note}[1]{{\color{red}{#1}}} 
\newcommand{\zbar}{\br{z}}
\newcommand{\so}{\mathfrak{so}}

\newcommand{\dpa}[1]{\frac{\partial}{\partial #1}}

\newcommand{\Obs}{\op{Obs}}

\newcommand{\eps}{\epsilon}
\newcommand{\g}{\mathfrak{g}}

\newcommand{\Ool}{\Oo_{loc}}

\newcommand{\E}{\mscr{E}}

\newcommand{\what}{\widehat}
\newcommand{\tr}{\triangle}

\newcommand{\K}{\mbb K}
\newcommand{\til}{\widetilde}
\newcommand{\mscr}{\mathscr}

\newcommand{\br}{\overline}

\newcommand{\iso}{\cong}
\newcommand{\C}{\mathbb C}

\newcommand{\norm}[1]{\left\| #1 \right\|}
\newcommand{\Oo}{\mscr O}
\newcommand{\Z}{\mathbb Z}

\newcommand{\op}{\operatorname}
\newcommand{\mbf}{\mathbf}
\newcommand{\mbb}{\mathbb}
\newcommand{\mf}{\mathfrak}
\newcommand{\mc}{\mathcal}

\newcommand{\ip}[1]{\left\langle #1 \right\rangle}
\newcommand{\abs}[1]{\left| #1 \right|}

\newcommand{\R}{\mbb R}
\renewcommand{\d}{\mathrm{d}}

\newcommand{\dbar}{\br{\partial}}

\DeclareMathOperator{\Sym}{Sym}

\DeclareMathOperator{\Tr}{Tr}

\newtheoremstyle{thm}
  {7pt}
  {7pt}
  {\itshape}
  {}
  {\bf}
  {.}
  {5pt}
  {\thmnumber{#2 }\thmname{#1}\thmnote{ (#3)}}

\newtheoremstyle{def}
  {7pt}
  {10pt}
  {\itshape}
  {}
  {\bf}
  {.}
  {5pt}
  {\thmnumber{#2} \thmname{#1}\thmnote{ (#3)}}

\newtheoremstyle{rem}
  {4pt}
  {10pt}
  {}
  {}
  {\itshape}
  {:}
  {3pt}
  {}

\newtheoremstyle{texttheorem}
  {8pt}
  {8pt}
  {\itshape}
  {}
  {\bf}
  {. \hspace{5pt}}
  {3pt}
  {}

\newtheorem*{theorem*}{Theorem}
\newtheorem*{lemma*}{Lemma}
\newtheorem*{corollary*}{Corollary}
\newtheorem*{proposition*}{Proposition}
\newtheorem*{definition*}{Definition}

\newtheorem{theorem}{Theorem}[subsection]
\newtheorem{thm-def}{Theorem/Definition}[theorem]
\newtheorem{proposition}[theorem]{Proposition}

\newtheorem{lemma}[theorem]{Lemma}

\numberwithin{equation}{subsection}

\theoremstyle{def}
\newtheorem{definition}[theorem]{Definition}

\theoremstyle{rem}

\newcommand{\cinfty}{C^{\infty}}

\newcommand{\Lap}{\bigtriangleup}
\newcommand{\vbar}{\br{v}}
\newcommand{\ubar}{\br{u}}
\newcommand{\PT}{\mbb{PT}}
\newcommand{\CP}{\mbb{CP}}
\newcommand{\half}{\tfrac{1}{2}}

\newcommand{\sixth}{\tfrac{1}{6}}
\newcommand{\sh}{\sf h}
\renewcommand{\i}{\mathrm{i}}
\renewcommand{\c}{\mathrm{c}}
\renewcommand{\t}{\mathrm{t}}

\title{Quantizing holomorphic field theories on twistor space}
\author{Kevin Costello}

\address{Perimeter Institue for Theoretical Physics}
\email{kcostello@perimeterinstitute.ca}

\begin{document}

\begin{abstract}
This paper studies a class of four-dimensional quantum field theories which arise by quantizing local holomorphic field theories on twistor space.   These theories have some remarkable properties: in particular, all correlation functions are rational functions. The two main examples are the $WZW_4$ model of Donaldson and Losev, Moore, Nekrasov and Shatashvili, and self-dual Yang-Mills theory. In each case, anomalies on twistor space must be cancelled by a Green-Schwarz mechanism, which introduces additional fields.  For $WZW_4$, this only works for $G = SO(8)$ and the additional field is gravitational. For self-dual Yang-Mills, this works for $SU(2)$, $SU(3)$, $SO(8)$ and the exceptional groups, and the additional field is an axion.  
\end{abstract}

\maketitle%

\section{Introduction}
This paper studies a class of quantum field theories in dimension $4$ which come from local holomorphic quantum field theories on twistor space. We call them \emph{twistorial} theories.  Twistorial theories enjoy many remarkable properties.  Let me highlight some of them:
\begin{enumerate} 
	\item All correlation functions
	$$\ip{\Oo_1(x_1) \dots \Oo_n(x_n)}$$
		of local operators  are entire analytic functions of the positions $x_i \in \C^4$, with singularities only on the analytically-continued light cone $\norm{x_i - x_j}^2 = 0$. 
\item  The renormalization group trajectory of a twistorial theory, in the space of analytically-continued perturbative QFTs, is periodic with period $2 \pi \i$. 	
\item Twistorial theories often have integrable properties, similar to those of two-dimensional integrable field theories (although we will not develop this aspect of twistorial theories in detail in this paper). 
\end{enumerate}
Some twistorial theories are non-renormalizable by power-counting, and yet all counter-terms are fixed uniquely by requiring that they come from local expressions on twistor space. 

There are very few twistorial theories. The properties listed above immediately fail for generic Lagrangians in dimension $4$, including Yang-Mills theory and the $\phi^4$ theory\footnote{ Indeed, property (2) precludes any logarithmic dependence on the scale, whereas both Yang-Mills theory and $\phi^4$ have a one-loop log divergence. Property (1) implies that the OPE between any two operators at $x_1,x_2$ can not contain any terms involving $\log \norm{x_1 - x_2}$, whereas any interacting scalar field theory in dimension $4$ where the interaction has fewer than $2$ derivatives will have a logarithmic OPE.} .  A Lagrangian needs to be very finely tuned for these properties to hold, and it is surprising that there are any QFTs at all with these properties. The main result of this paper is a (rigorous) construction of several $4d$ QFTs which satisfy these properties.

I will sketch the construction of a number of twistorial field theories, but focus in detail on two examples.  

\subsection{$WZW_4$ as a twistorial theory}
The first example studied in detail in this paper is the four-dimensional WZW model of Donaldson \cite{Donaldson:1985zz}  and Losev, Moore, Nekrasov and Shatashvili \cite{Losev:1995cr}. The fundamental field of this model is a map
\begin{equation} 
	\sigma : \R^4 \to G 
\end{equation}
where $G$ is a Lie group. In terms of the current $J = \sigma^{-1} \d \sigma$, the Lagrangian is  
\begin{equation}
		 \int_{\R^4}  \op{tr} J\wedge \ast J + \tfrac{1}{3}\int_{\R^4 } A \wedge \op{tr} J\wedge[J,J].
\end{equation}
Here $A$ is a background $U(1)$ gauge field on $\R^4$ whose fields strength is the K\"ahler form $\omega$. 
	
The first term in the Lagrangian is the standard $\sigma$-model Lagrangian.  The $\sigma$-model has a topological $U(1)$ symmetry in dimension $4$, because $\pi_3(G) = \Z$.  The current for this symmetry is $\op{tr}(J,[J,J])$.  The second term in the Lagrangian simply means that we have a background $U(1)$ gauge field for this topological symmetry. 

It was stated in \cite{WH} that this model arises from holomorphic Chern-Simons on twistor space; this result was proved in \cite{2011.04638,2011.05831} with several interesting generalizations presented in \cite{2011.04638}. 

Our results apply to this model with $G = SO(8)$, 
coupled to some additional ``gravitational'' fields which we refer to as ``closed string fields'' because of their origin in a certain twistor string theory. The restriction to $SO(8)$ is forced on us by a Green-Schwarz anomaly cancellation on twistor space. 

Without the addition of this closed-string field, the model is not twistorial.  In section \ref{section_twistor_construction} we show that an anomaly prevents one constructing the model as a local QFT on twistor space. In section \ref{sec:wzw4_gs} we show that property (1) above fails without the addition of the closed-string field: there are logarithmic terms in the two-loop operator product expansion which are canceled by a Green-Schwarz mechanism when we introduce the closed-string field.

	The closed string field is a scalar field $\rho$, which we view as the K\"ahler potential for a closed $(1,1)$-form $\partial \dbar \rho. $  The field $\rho$ is defined up to the addition of a constant; it might be better to think of $\rho$ as being a closed $1$-form.

The $\sigma$-model can be defined on K\"ahler manifolds of complex dimension $2$, and so has a K\"ahler stress-energy tensor $T_{\text{K\"ahler}}$ which is a $(1,1)$-form that couples to the variation of the K\"ahler form (in a fixed complex structure).  The K\"ahler potential $\rho$ couples to this by
\begin{equation} 
	 \int T_{\text{K\"ahler}} \partial \dbar \rho. 
\end{equation}
This makes it clear that we should treat $\rho$ as as gravitational, because it couples to the stress-energy tensor of the $4d$ WZW model.  

The Lagrangian for the K\"ahler potential, to the order we compute it, is 
\begin{equation}
	\begin{split}
		\half   \int (\Lap \rho)^2 +   \frac{1}{16 \pi}    
		\int_{u,\ubar} \op{tr}( J^{0,1} \partial J^{0,1}  ) \partial \rho + C \left( (\Lap \rho)  (\partial \dbar \rho)^2  - \frac{2}{3} (\Lap \rho)^3 \omega^2\right) + O(\rho^3).
	\end{split}
\end{equation}
Note the unusual kinetic term, involving the square of the Laplace operator. The constant $C$ is non-zero, and its value is determined in principle by anomaly cancellation on twistor space, but we do not determine it.   

To quadratic order in $\rho$, the equations of motion are that  the scalar curvature of the K\"ahler metric given by $\omega + \partial \dbar\rho$ vanishes. It is known \cite{derdzinski1983self} that every scalar flat K\"ahler manifold is anti-self-dual,  and so has a description in terms of an integrable twistor space\footnote{I'm very grateful to David Skinner for pointing out the result in \cite{derdzinski1983self} and for other helpful discussions.}   .   This suggests that the Lagrangian enforces vanishing of the scalar curvature to all orders. Unfortunately I was unable to prove this.

\subsection{Self-dual Yang-Mills theory }
A simpler, but more physically relevant, twistorial field theory we study is a variant of self-dual Yang-Mills theory. 

Self-dual Yang-Mills theory is a degenerate limit of ordinary Yang-Mills theory, where the fields are a gauge field $A \in \Omega^1(\R^4,\mf{g})$ and an adjoint-valued self-dual two-form $B \in \Omega^2_+(\R^4,\mf{g})$. The Lagrangian is
\begin{equation} 
	\int \op{Tr} ( B F(A)_+ ). 
\end{equation}
It is well-known \cite{Ward:1977ta,Boels:2006ir} that, at the classical level, this Lagrangian comes from a local Lagrangian on twistor space.

In section \ref{sec:sdym_twistor} we demonstrate that this fails at the quantum level, because of an anomaly on twistor space.    

In some cases, we can use a Green-Schwarz mechanism to cancel the anomaly, making self-dual Yang-Mills into a twistorial theory at the quantum level.   We can do this for any Lie algebra $\g$ for which, for $X \in \g$, 
\begin{equation} 
	\op{Tr}(X^4 ) = \lambda_{\g}^2 \op{tr}(X^2)^2 
\end{equation}
where $\op{Tr}$, $\op{tr}$ indicate trace in the adjoint and fundamental representations respectively\footnote{More generally, one can add matter in any real representation $R$, in which case the left hand side must be modified by subtracting $\op{tr}_R(X^4)$.} and $\lambda_{\g}$ is a $\g$-dependent constant.

This identity holds for $\sl_2$, $\mf{sl}_3$, $\mf{so}(8)$ and all exceptional groups, with some values of $\lambda_{\g}$ being
\begin{equation} 
	\lambda_{\sl_2}^2 = 8 \ \ \lambda_{\sl_3}^2 = 9 \ \ \lambda_{\so_8}^2 =  \frac{3} {2} 
\end{equation}
The Green-Schwarz mechanism also holds for $\g = \so(N_c)$ with matter in $N_f = N_c - 8$ copies of the fundamental representation, and for\footnote{We write $\sl(N)$ instead of $\mf{su}(N)$ as we will work with analytically continued fields; a choice of path integral contour gives us $\mf{su}(N)$ gauge theory, and in perturbation theory all results are independent of the contour.  } $\g = \sl(N_c)$ with $N_f = N_c$.  

The closed-string field we need to adjoin is $\rho$ as before, where the full Lagrangian is
\begin{equation} 
	\int \op{Tr} (B F(A)_+ ) + \half \int (\Lap \rho)^2  +  \lambda_{\g} \frac{1}{8 \pi  \sqrt{3} }\int  \d \rho CS(A) \label{eqn_sdym_axion}
\end{equation}
where $CS(A)$ is the Chern-Simons three-form normalized so $\d CS(A) = \op{tr}(F(A)^2)$.    The field $\rho$ is a kind of axion; we can take it to be periodic because of the quantization of $\op{tr}(F(A)^2)$.

\subsection{Results about $WZW_4$}
	Our main result about the $WZW_4$ model, plus the closed-string field $\rho$,   is that the counter-terms can all be fixed \emph{uniquely} in a scheme-independent way.  That is, the model is well-defined at the quantum level, despite being non-renormalizable.  

	Since this seems to dramatically contradict the usual philosophy, let me sketch the constraints that fix the counter-terms before discussing the models in more detail. 

The most direct way to fix the counter-terms is to use the fact, which we discuss in more detail in section \ref{twistor_generalities},  that the classical Lagrangian arises as the dimensional reduction of a local Lagrangian on twistor space, $\PT = \Oo(1)^2 \to \CP^1$.   At the quantum level, we can ask that the counter-terms are expressed as local functionals on twistor space.  Here, being \emph{local} on twistor space is essential: any counter-term on $\R^4$ can be lifted to non-local expressions on twistor space, but only very special counter-terms come from local expressions on twistor space.   This constraint, which is equivalent to asking that the quantum theory is dimensionally reduced from a local field theory twistor space, turns out to fix all counter-terms uniquely.    

The twistor-space theory is a type I topological string \cite{1905.09269} , and the topological string version of the Green-Schwarz mechanism fixes the group to be $SO(8)$.

It would be more satisfying to have a direct space-time way to constrain the counter-terms.    In section \ref{twistor_generalities} we will argue in detail that any local holomorphic field theory on twistor space gives rise to a theory on $\R^4$ satisfying properties (1) and (2) above. 
Both of these properties constrain the counter terms very strongly, and we conjecture that they are enough to uniquely fix all counter-terms.

As mentioned above, this argument gives a no-go theorem:
\begin{theorem} 
	Any field theory on $\R^4$ which has a one-loop logarithmic divergence does not arise from a local holomorphic field theory on twistor space. 
\end{theorem}
This applies, in particular, to  Yang-Mills theory \footnote{This result appears, at first sight, to be in tension with the interesting recent paper \cite{Popov:2021mfl}.  There is no contradiction, however.   The Lagrangian proposed in \cite{Popov:2021mfl} is based on holomorphic Chern-Simons for a  non-integrable complex structure on twistor space.  Because the $\dbar$ operator for a non-integrable complex structure does not square to zero, the Lagrangian is not invariant under the usual gauge transformations of holomorphic Chern-Simons.   If we do not impose the usual gauge transformations, the theory is not holomorphic. Just like a topological theory is one in which all operators are killed by all translations, a holomorphic theory is one in which all operators are killed by $\partial_{\zbar_i}$ in local holomorphic coordinates.  If we study holomorphic Chern-Simons theory without imposing gauge invariance, this holomorphic condition no longer holds.  } and to the $\phi^4$ theory.

We find in section \ref{sec:wzw4_gs} that the property (2) -- that the OPEs are entire analytic functions -- requires a Green-Schwarz type anomaly cancellation and fails unless we use the closed-string field $\rho$.

\subsection{Results about self-dual Yang-Mills theory}
We have already stated our main result about self-dual Yang-Mills theory: all correlation functions are rational functions, for the gauge groups mentioned above and when we couple to an axion field as in equation \eqref{eqn_sdym_axion}.  

We also show that \emph{without} the axion field there are logarithmic correlation functions.  We find that at two loops there is a logarithm appearing in the OPE of the stress-energy tensor with itself. We also find that at three loops there is a logarithm in the two-point function of the operator $B^2$:
\begin{equation} 
	\ip{B^2(0), B^2(x)} = \tfrac{4}{3}  \frac{\lambda_{\g}^2}{(8 \pi \norm{x})^8   }   \log \norm{x} + \text{ rational }  
\end{equation}

Such logarithmic two-point function can not appear in a local theory on twistor space. Therefore it is a reflection of the anomaly on twistor space.  This logarithmic OPE must be canceled when we couple to the field $\rho$. 

This gives us a technique for calculating the precise coefficient of non-rational correlation functions, which we used in the three-loop computation above. Any logarithmic correlation functions in the theory without axions must be canceled by an exchange of axion fields, as in Figure \ref{Figure_axion_exchange}.    This vastly simplifies computations, because each axion line cancels the log divergences in a loop of gluons.

For the computation of the two-point function of $B^2$, there is only one diagram involving an exchange of axion fields which can contribute a logarithm. 
\begin{figure}
	\includegraphics[scale=0.22]{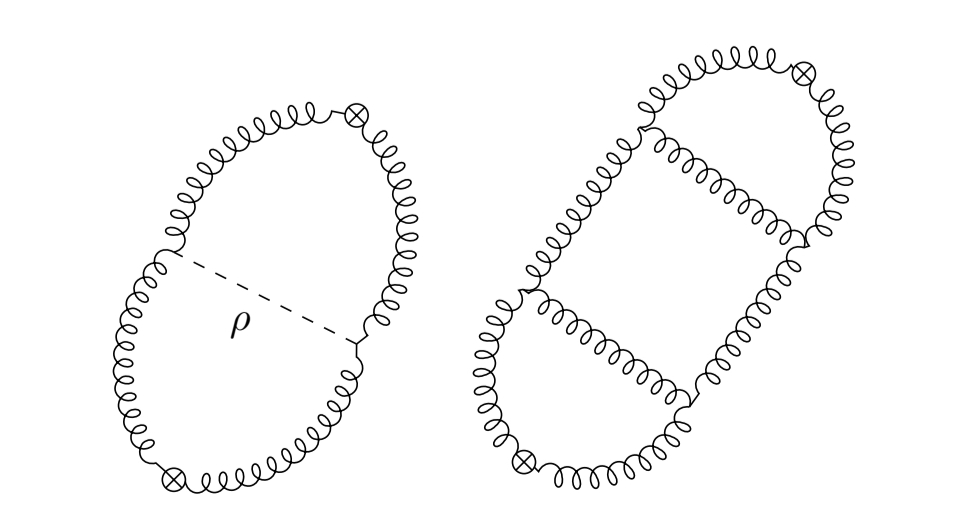}
	\caption{These diagrams contribute to the two-point function of the operator $B^2$, labeled by the vertices $\otimes$, in the theory with $\rho$ included. The exchange of the field $\rho$, as in the diagram on the left, cancels all logarithmic terms from three loop diagrams such as the one on the right. This allows us to compute the non-rational term in the two-point function at three loops using only a single diagram, that on the left. \label{Figure_axion_exchange}} 
\end{figure}

Although I didn't perform any computations of correlation functions beyond the one described above, one can expect a pattern like this to continue. For instance, we can contemplate computing  the connected $2n$ point function of the operator $B^2$ in self-dual Yang-Mills.  This involves Feynman diagrams with $2n+1$ loops.  One can expect that the term which is the most transcendental as a function of the position of the operators is canceled, when we include the axion field, by an exchange of the maximum number of axions, which is $n$.  Since each axion corresponds to a loop of gluons, this means that we could compute the maximally transcendental term by using diagrams with $n+1$ loops instead of $2n+1$.  

Full Yang-Mills theory is obtained by adjoining to the self-dual Lagrangian $B F(A)_+$ the term $-\tfrac{1}{4} g^2 B^2$. (The $B^2$ term was realized as a non-local term on twistor space in \cite{Mason:2005zm, Boels:2006ir}).  The field $B$ in first-order Yang-Mills becomes $2 g^{-2} F_+$.  The computation of the two-point function of $B^2$ in self-dual Yang-Mills immediately gives us a term in the two-point function of $F_+^2$ in full Yang-Mills: 
\begin{equation} 
	\ip{F_+^2(0) F_+^2(x)} = g^{8} \tfrac{1}{3}  \frac{\lambda_{\g}^2}{(8 \pi \norm{x})^8   }   \log \norm{x} + \text{ rational }  + O(g^{10}).  
\end{equation}
In this direction, our results also imply that if we take full Yang-Mills plus the axion field $\rho$, the first non-rational term in the two-point function of $F_+^2$ occurs at $4$ loops, or order $g^{10}$.  More generally, we have:
\begin{theorem} 
	Consider the correlation function of $n$ gauge invariant local operators $\mc{O}_1 \dots \mc{O}_n$ in full Yang-Mills theory plus the axion field $\rho$ coming from anomaly cancellation on twistor space.

	Suppose $\mc{O}_i$ contains $n_i$ copies of $F_+$. Then, the first non-rational term in the correlation function
	\begin{equation} 
		\ip{\mc{O}_1 (x_1) \dots \mc{O}_n(x_n)}  
	\end{equation}
	occurs at order $g^{2 + 2 \sum n_i}$. 
\end{theorem}

Another result we find using our method is a surprising statement about the RG flow of self-dual Yang-Mills deformed by $-\tfrac{1}{4} g^2 B^2$.   The standard RG flow in Yang-Mills theory occurs at second order in $g^2$.   This also happens when we view full Yang-Mills as a deformation of the self-dual theory. It has been proven \cite{Elliott:2017ahb} that  when we deform self-dual Yang-Mills by $- g^2 B^2$, then the RG flow modifies the deformation to $-\tfrac{1}{4}g^2 B^2 + C g^4 B^2$ for some constant $C$.

What we find here is that there is another term which appears at a lower order.  The very first term in the RG flow changes the coefficient of the toplogical term $F\wedge F$: 
\begin{equation} 
	-\tfrac{1}{4} g^2 B^2 \mapsto -\tfrac{1}{4}g^2 B^2 +  \frac{\lambda_{\g}^2}{(2 \pi)^4 3 \cdot 2^5} F^2  
\end{equation}
If the coupling constant $g$ doesn't depend on the space-time coordinates, this term in the RG flow is trivial in perturbation theory, because $F^2$ is is a total derivative. However, if $g$ is non-constant, then this analysis means that the leading term in the RG flow is given by the Chern-Simons term $\d g CS(A)$.  

\subsection*{Acknowledgments}
I am grateful to Roland Bittleston, Lionel Mason and  David Skinner for some very helpful conversations; to Greg Moore and Samson Shatashvili for explaining their work on $WZW_4$ and for correcting an important error I made in an early talk on this work; and to Yehao Zhou for collaborating in the early stages of this paper.  The author is supported by the NSERC Discovery Grant program and by the Perimeter Institute for Theoretical Physics. Research at Perimeter Institute is supported by the Government of Canada through Industry Canada and by the Province of Ontario through the Ministry of Research and
Innovation.

\section{Constructing local holomorphic field theories on twistor space}
\label{section_twistor_construction}

Let $\PT$ denote the twistor space of $\R^4 $. This is the total space of the rank two vector bundle $\Oo(1)^2$ over $\CP^1$.  As a real manifold, there is an isomorphism
$$\PT \iso \R^4 \times   \CP^1.$$
This gives the projection
$$
\pi : \PT \to \R^4.
$$
For every $ x \in \R^4$ the fibre of this map is a copy of $\CP^1$. The points in the fibre $\CP^1_x$ parameterize complex structures on $\R^4$ compatible with a given orientation.

Our primary interest in this paper is in local quantum field theories on twistor space.  As we will review below, every such theory gives rise to a local quantum field theory on $\R^4$.  A number of local holomorphic QFTs on twistor have been studied in the literature before: 
\begin{enumerate} 
	\item Witten and Berkovits \cite{Witten:2003nn, Berkovits:2004hg} show that holomorphic Chern-Simons on super twistor-space $\PT^{3 \mid 4}$ gives rise to the self-dual limit of $N=4$ Yang-Mills on $\R^4$.
	\item This was generalized by Boels, Mason and Skinner \cite{Boels:2006ir} who showed how to construct self-dual supersymmetric Yang-Mills theories with various amount of supersymmetry and matter from certain holomorphic field theories on twistor space. Here, and in \cite{Mason:2005zm} it was also shown how certain non-local terms move us away from the self-dual limit.   
	\item The Penrose-Ward correspondence  \cite{Ward:1977ta} can be viewed as the special case of this with no supersymmetry. In this case, holomorphic BF theory on twistor space gives rise to self-dual Yang-Mills theory on $\R^4$.
	\item Penrose's \cite{Penrose:1976js} non-linear graviton construction relates solutions of self-dual Einstein equations with certain geometrical structures on twistor space.  This was uplifted to a twistor Lagrangian describing self-dual gravity in \cite{Mason:2007ct} and \cite{Sharma:2021pkl}.  
\end{enumerate}
From our perspective, these classical theories suffer from two problems.  Firstly, in all of these examples, perturbation theory stops at one loop. Thus, they are not very interesting as quantum field theories on $\R^4$: they are richer than free theories, but not that much richer.   

To get something really interesting -- such as full $N=4$ Yang-Mills and not the self-dual limit -- one typically needs to deform these Lagrangians by terms which are non-local on twistor space.

The second difficulty with these models is that they often suffer from gauge anomalies, so that they tend to be ill-defined at the quantum level. We will discuss anomalies on twistor space in section \ref{subsec_anomalies}. 

Our goal in this section is to construct holomorphic theories on twistor space which are not one-loop exact, and which are anomaly free.  To construct a model where perturbation theory does not stop at one loop, we will use holomorphic Chern-Simons theory.                                                                                           
Recall that on a Calabi-Yau threefold $X$, the fields of holomorphic Chern-Simons theory with gauge Lie algebra $\g$ are
$$
\mc{A} \in \Omega^{ 0,1} ( X) \otimes \g
$$
The Lagrangian is $\int \Omega CS (\mc{A}) $ where $CS$ is the Chern-Simons three-form.

Of course, $\PT$ is not a Calabi-Yau manifold; the canonical bundle is $\Oo(-4)$. We will define the holomorphic Chern-Simons Lagrangian by choosing a volume form with poles, and then modifying the gauge field and gauge transformations at the location of the poles. If $z$ is the coordinate on the $\CP^1$ twistor fibre above the origin, we can choose a meromorphic volume form  on twistor space with second order poles when $z$ is zero or infinity. This
determines the volume form uniquely up to scale; in coordinates it is 
$$
 \Omega = z^{-2} \d z  \d v_1 \d v_2.
 $$
For the Chern-Simons Lagrangian to be well-defined, we need the gauge field to go like $z$ at the point $z = 0$, and like $1/z$ near $z = \infty$.  The zeroes in the gauge field cancel the poles in the volume form. To ensure gauge invariance, we also ask that gauge transformations also go like $z$ or $1/z$ at these points.

It is natural to think of the locus $ z= 0, \infty$  as being at the boundary of twistor space.
Asking that our gauge field, and gauge transformations, vanish at $z = 0,\infty$ can be thought of as imposing Dirichlet boundary conditions.

\subsection{Anomalies on twistor space}
\label{subsec_anomalies}
Holomorphic Chern-Simons theory, or holomorphic BF theory, has a one-loop anomaly on any Calabi-Yau manifold \cite{1505.06703,1905.09269,1910.04120}, associated to the diagram in Figure \ref{fig:anomaly}.  If $\c$ is the parameter for the gauge transformation, then the gauge variation of this diagram is 
\begin{equation} 
	  	\frac{1}{(2 \pi)^4 48 }  \op{Tr}_{\g} (\c (\partial \mc{A})^3) 
\end{equation}
with $\mc{A}$ being the holomorphic Chern-Simons gauge field.

\begin{figure}
\includegraphics[scale=0.22]{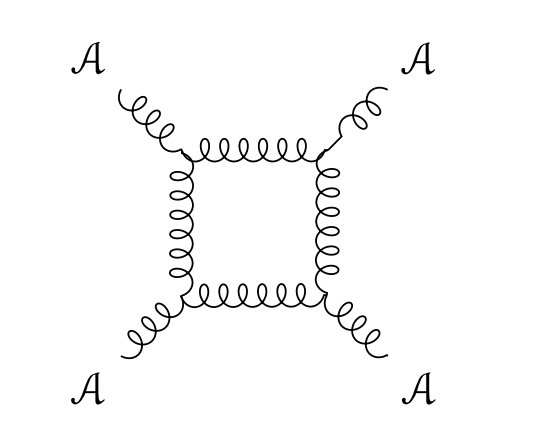}
	\caption{The gauge anomaly for holomorphic Chern-Simons. In our conventions for Feynman diagrams, external lines are labeled by arbitrary on-shell field configurations, and there are no propagators on external lines. \label{fig:anomaly}}
\end{figure}

Since anomalies are local, this means that holomorphic Chern-Simons on twistor space also has a one-loop anomaly.  The same anomaly appears in holomorphic BF theory on twistor space, meaning that any naive quantization of the Penrose-Ward correspondence suffers from a one-loop anomaly.  

Many, if not most, if the twistorial field theories mentioned above have anomalies.  One of the few examples that doesn't is the twistor representation of self-dual $N=4$ Yang-Mills, where the contribution from the bosonic and fermionic fields in the box diagram cancels. 

\subsection{The Green-Schwarz mechanism}
\label{section_typeI_intro}
Fortunately, in very special cases, it is possible to cancel the anomaly from the box diagram.  In \cite{1905.09269} Si Li and the author introduced a topological string version of the Green-Schwarz mechanism which cancels the anomaly for holomorphic Chern-Simons theory for $G = SO(8)$.   Let us explain this briefly. 

 The cancellation involves the closed-string fields of the type I $B$-model, which is a variant of the topological $B$-model whose world-sheet is un-oriented.  The closed string fields of the type I $B$-model are a subsector of those of Kodaira-Spencer theory \cite{hep-th/9309140}: on a Calabi-Yau manifold $X$, they consist of a Beltrami differential field
\begin{equation} 
	\mu \in \Omega^{0,1}(X, TX)  
\end{equation}
which is constrained to satisfy $\op{Div}\mu = 0$, where $\op{Div}$ is the holomorphic divergence operator
\begin{equation} 
	\op{Div}: \Omega^{0,1}(X,TX) \to \Omega^{0,1}(X).   
\end{equation}
(In practice, it is often better to implement this constraint homologically, by asking that $\op{Div}\mu$ is $\dbar$ exact).  The field $\mu $ describes deformations of $X$ as a Calabi-Yau manifold. 

Under the isomorphism between $\Omega^{0,\ast}(X,TX)$ and $\Omega^{2,\ast}(X)$, the holomorphic divergence operator becomes $\partial$, the holomorphic part of the de Rham operator.  As such, we will often use the symbol $\partial$ to refer to it. 

The Lagrangian for the closed-string fields of the type I $B$-model is
\begin{equation} 
	 \half\int (\dbar \mu ) \partial^{ -1}( \mu \vee \Omega ) + \sixth \int \mu \vee \mu \vee \mu \vee \Omega
\end{equation}
This is a variant of the Lagrangian introduced in \cite{hep-th/9309140}.

Note the unusual kinetic term, which only makes sense because of the constraint $\partial \mu = 0$. The gauge symmetries of this model are given by vector fields which are sections of the holomorphic tangent bundle $ TX$ and which are divergence free.

The Beltrami differential $\mu$ is coupled to the holomorphic Chern-Simons gauge field $ \mc{A}$ by replacing the $\dbar$ by the $ \dbar$ operator in the new complex structure. Explicitly, if we write 
\begin{equation} 
	\eta = \mu \vee \Omega 
\end{equation}
then the term in the Lagrangian is
\begin{equation} 
  \frac{1}{8 (2 \pi \i)^{3/2} }   \int \eta \op{tr}( \mc{A} \partial \mc{A} ) 
\end{equation} 
(The normalization is computed in the appendix \ref{sec:normalization_twistor}).

When we work on twistor space, we need to take into account the behaviour at the poles of the meromorphic volume form. Recall that there are second order poles along the divisors $ z= 0$, $ z=\infty$.  

The boundary conditions for the Beltrami differential $\mu$ are best written in terms of the $(2,1)$ form $\eta = \mu \vee \Omega$.  Our boundary condition is simply that the $(2,1)$ form  $\eta$ is regular on all of $\PT$ \footnote{In \cite{Costello:2018zrm} we studied a related model where the boundary condition asked that $\eta$ has log poles at the poles of $\Omega$.  That also works, but we find the boundary condition chosen here a little more convenient.}.   In terms of the Beltrami differential $\mu$, this condition means that it vanishes to second order near $z = 0$ and $z = \infty$ (gauge transformations behave in the same way).  Near $z = 0$, the cubic term in the action has a factor of $z^6$ from the zeroes in the fields, and of $z^{-4}$ from the second-order pole in $\Omega$ (recalling that $\Omega$ appears twice in the cubic term).   Therefore, overall,  the cubic term has a coefficient of $z^2$, and vanishes to order $2$. Similarly, it vanishes to order $2$ at $z = \infty$.

To understand the kinetic term, we note that it can be written entirely in terms of the $(2,1)$ form $\eta$, which is $\partial$ closed.  The kinetic term is 
\begin{equation} 
	\int \dbar \eta \partial^{-1} \eta 
\end{equation}
which can be written down on any complex three-fold.

Now that we understand the closed string fields, we can discuss the Green-Schwarz mechanism. The following result is proved in \cite{1905.09269}. 
\begin{theorem} 
There is a tree-level anomaly to coupling holomorphic Chern- Simons theory to type I Kodaira- Spencer theory, proportional to
\begin{equation} 
\int\op{tr} ( \c \partial \mc{A}) \op{tr} ( \partial \mc{A} \partial \mc{A})
\end{equation} 
\end{theorem}
The tree level anomaly is drawn in Figure \ref{Figure:closed_anomaly}.
\begin{figure}

\includegraphics[scale=0.21]{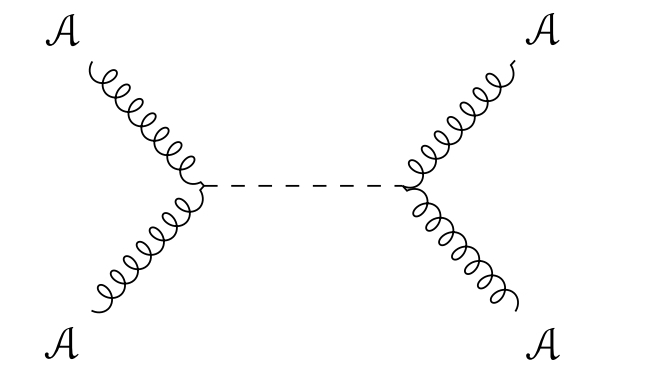}
	\caption{Two open-string fields exchanging a closed-string field gives a tree-level anomaly. \label{Figure:closed_anomaly}  }  
\end{figure}

As a consequence of this, we find that the one-loop anomaly for holomorphic Chern-Simons cancels for any Lie algebra $\g$ for which the identity
\begin{equation} 
\Tr( X^4) \propto \op{tr}( X^2)^2
\end{equation} 
holds.

This identity holds for a number of Lie algebras, including $\so(8)$, $\sl(2)$, $\sl(3)$, and all exceptional simple Lie algebras.   In \cite{1905.09269} it was shown that for $\g=\so (8)$ all anomalies at all orders in the loop  expansion vanish, and further, all higher loop counter-terms are uniquely fixed by the requirement that all anomalies vanish.   A small variant of this cohomological argument applies on twistor space:
\begin{theorem} 
	Holomorphic Chern-Simons theory for $SO(8)$, coupled to type I Kodaira-Spencer theory, admits a unique quantization on twistor space with the boundary conditions discussed above. 
\end{theorem}
We will give the proof of this, along with related results, in the appendix \ref{section_twistor_quantization}.  We have chosen our normalization of the open-closed interaction to cancel the anomaly.

We do not currently know\footnote{As in the usual Green-Schwarz mechanism, a one-loop diagram with only closed string fields on the external lines should force the gauge algebra to have dimension $28$.  This raises the possibility that the $\mf{g}_2 \oplus \mf{g}_2$ theory might also be anomaly free to all orders. } whether or not higher loop anomalies can be canceled for any other groups.  

It is worth mentioning that the type I topological string should admit an embedding in the physical string, in a way similar to the standard embedding of the topological A and B models in the type IIA and IIB string theories. For the type I topological string, the relevant set-up is type IIB with an $O7-$ plane and $4$ $ D7$' s. This system is placed in an $\Omega$ background and subject to supersymmetric localization, leaving us with an effective $6$ dimensional theory with gauge group $SO (8)$.

\subsection{Canceling the anomaly with a non-local term}
The $WZW_4$ model with target a general group $G$ is not a gauge theory, and therefore can not suffer from a gauge anomaly.  The gauge anomaly on twistor space can not be canceled by a local counter-term.  However, because there is no gauge anomaly on $\R^4$, we see that it must be possible to cancel the twistor-space anomaly with a term which is non-local from the point of view of  twistor space, but still local from the point of view of $\R^4$.

Here we will write an expression that does this.  We will coordinatize twistor space as $\R^4 \times \CP^1$, with coordinates $x,z$. We will write down an expression where we integrate over one $x$ variable but two $z$ variables -- so it is bi-local on twistor space but local on $\R^4$. The expression is
\begin{equation} 
	\int_{x \in \R^4} \int_{z,z' \in \CP^1} \mc{A}(x,z) \d \mc{A}(x,z)\frac{\d z \d z'}{(z - z')^2}  \mc{A}(x,z') \d \mc{A}(x,z')  + O(\mc{A}^5). 
\end{equation}
This expression is interpreted by viewing $\mc{A}(x,z)$ as a $1$-form on twistor space. We check in the appendix \ref{sec:nonlocal} that the gauge variation of this expression cancels the anomaly.

\subsection{The Green-Schwarz mechanism for self-dual Yang-Mills}
Holomorphic BF theory on twistor space corresponds to self-dual Yang-Mills theory on $\R^4$.    We let $\mc{B} \in \Omega^{3,1}(\PT, \g)$ and $\mc{A} \in \Omega^{0,1}(\PT,\g)$ be the fields of holomorphic BF theory. The Lagrangian is
\begin{equation} 
	\int \op{tr} (\mc{B} F^{0,2}(\mc{A})). 
\end{equation}
However, holomorphic BF theory has a one-loop anomaly, which is the same as that associated to holomorphic Chern-Simons theory,  associated to Figure \ref{fig:anomaly}.   

Take any simple Lie algebra $\g$, and let $\op{Tr}$ denote trace in the adjoint representation and $\op{tr}$ that in the fundamental representation. Let us assume that there is some $\lambda_{\g}$ such that
\begin{equation} 
	\op{Tr}(X^4) = \lambda_{\g}^2 \op{tr}(X^2)^2. 
\end{equation}
We mentioned above when this happens. 

Then, we can cancel the anomaly  coupling to the \emph{free} limit of the Kodaira-Spencer theory we used to cancel the anomaly for holomorphic Chern-Simons theory. The new field we introduce is $\eta \in \Omega^{2,1}(\PT)$, satisfying the constraint $\partial\eta = 0$ and subject to the gauge transformation $\eta \mapsto \eta + \dbar \chi$, for $\chi \in \Omega^{2,0}(\PT)$.  The Lagrangian is
\begin{equation} 
	\half \int \dbar \eta \partial^{-1} \eta 
\end{equation}
which is the kinetic term of the Lagrangian we used before.

We couple this field to the gauge field $\mc{A}$ by  
\begin{equation} 
	  \frac{\lambda_{\g} }{(2 \pi \i)^{3/2} \sqrt{48} } \int \eta \op{tr}( \mc{A} \partial \mc{A} ) 
\end{equation} 

In section \ref{sec:sdym_twistor},  we will determine the theory on $\R^4$ corresponding to this Lagrangian.

\section{Other examples of twistorial field theories}
So far, we have only discussed two twistorial field theories. Let me sketch the construction of other models, leaving detailed investigation to future work. 

\subsection{Model build from a genus $3$ curve}
One variant of the model we have introduced is obtained by using the same type I topological string, but on a different geometry. The simplest case is when we take a hyperelliptic genus $3$ curve $\Sigma \to \CP^1$.  The pull-back of $\Oo(2)$ is the canonical bundle, so that the three-fold 
\begin{equation} 
	\PT_\Sigma = \pi^\ast(\Oo(1))^2 \to \Sigma 
\end{equation}
is Calabi-Yau.  This fibres over $\R^4$ with fibres $\Sigma$.

The field theory on $\R^4$ corresponding to the type I topological string theory with this geometry includes a $\sigma$-model with target the moduli of $SO(8)$-bundles on $\Sigma$.   The fields coming from the closed-string sector include a map from $\R^4$ to the moduli space of genus $3$ hyperelliptic curves, which is a divisor in the moduli of genus $3$ curves.   As always, these theories are analytically-continued theories, and the path integral should be defined by the choice of a contour.

\subsection{Compactifying from $10$ dimensions to twistor space}

Other examples are obtained by starting with supersymmetric sectors of $10$ dimensional string theories. These supersymmetric sectors are themselves described by $10$-dimensional topological strings \cite{1606.00365, Baulieu:2010ch, Berkovits:2003pq} .  We can often use these holomorphic twists of string theory to build holomorphic field theories on twistor space. By compactifying to $\R^4$ we find a mechanism to produce non-supersymmetric field theories on $\R^4$ from string theory.

This will give a broad class of non-renormalizable and non-supersymmetric four-dimensional theories whose UV properties are nonetheless under control, because they come from holomorphic anomaly-free theories on twistor space.

The simplest such models arise by taking the holomorphic twist \cite{1606.00365} of the type I or type IIB strings, compactified to twistor space.  The geometry we need to compactify these models to twistor space is a complex $3$-fold $X$, together with a non-constant map $\pi :X \to \CP^1$, and an isomorphism of line bundles 
\begin{equation} 
	K_X \iso \pi^\ast \Oo(2). 
\end{equation}
(Note that this implies that the fibres of $\pi$ are Calabi-Yau surfaces).  We can then build a Calabi-Yau $5$-fold 
\begin{equation} 
	\PT_X = \pi^\ast \left( \Oo(1) \oplus \Oo(1)  \right) \to X. 
\end{equation}
As a real manifold, there is an isomorphism
\begin{equation} 
	\PT_X = X \times \R^4. 
\end{equation}
Complexified space-time $\C^4$ consists of the moduli space of the $3$-fold $X$ in the $5$-fold $\PT_X$.  Thus, $\PT_X$ is a higher-dimensional variation of twistor space, in which the $\CP^1$ twistor lines are replaced by the $3$-fold $X$.  

One way to construct such a variety $X$ (suggested to me by Davesh Maulik) is to start with a Calabi-Yau $3$-fold $X_0$, fibred over $\CP^1$ in $K3$ surface.  Then, we take $X$ to be the pull-back of $X_0$ along a double branched cover $\CP^1 \to \CP^1$, whose branch points consist of smooth fibres of the fibration $X_0 \to \CP^1$.  It is then elementary to check that $\pi^\ast O(2)$ is $K_X$.

The construction of field theories on $\R^4$ from field theories on $\PT$ carries through when we compactify holomorphic twists of string theories from the $5$-fold $\PT_X$ to $\R^4$.    I will leave it as an open problem to carefully work out the Lagrangian of the theory on $\R^4$ corresponding to a holomorphic twist of a ten-dimensional string theory in some geometry.  The computation seems to me to be quite challenging.

\section{Compactifying theories on twistor space}
\label{twistor_generalities}
Now that we have constructed some interesting local and holomorphic theories on twistor space, one can ask, why are they so interesting?  In this section we will explain, from an axiomatic perspective, how quantum field theories on twistor space give rise to quantum field theories on $\R^4$ with the special properties we have mentioned earlier.

\subsection{Twistor space reduction vs. Kaluza-Klein reduction}
Let us first emphasize that passing from a theory on twistor space to one on $\R^4$ is somewhat better-behaved than the familiar Kaluza-Klein reduction.  Consider, for example, KK reduction of a field theory on $S^1 \times \R^d$ to $\R^d$.  We obtain a local field theory on $\R^d$ only in the limit as the radius of the circle goes to zero, or equivalently when we work in the far IR.   Without taking this limit, we find a field theory on $\R^d$ with an infinite number of massive KK modes.  These disappear in the infra-red, so that the local field theory on $\R^d$ is an effective low-energy description of the $d+1$ dimensional theory.

When passing from a theory on twistor space to one on $\R^4$, the procedure is more robust.  The description as a local field theory on $\R^4$ is valid at all energy scales. Without passing to the IR, there are only a finite number of fields.

Let us explain how this happens, first using physical language and then formalizing the idea using the concept of factorization algebra.

It is helpful to choice a simple concrete example.  We will choose the free scalar field, which is represented on twistor space by a field
\begin{equation} 
	\mc{A} \in \Omega^{0,1}(\PT, \Oo(-2)) = \Omega^{0,1} (\PT, K^{\tfrac{1}{2}}) \label{eqn_fft_twistor} 
\end{equation}
with Lagrangian $\mc{A} \dbar \mc{A}$ and gauge symmetry $\mc{A} \mapsto \mc{A} + \dbar \chi$, where $\chi \in \Omega^{0,0}(\PT, K^{\tfrac{1}{2}})$.  

Using the isomorphism $\PT \iso \R^4 \times \CP^1$, we can write this as a field theory on $\R^4$ with an infinite number of fields.  We will see that these fields are all (except for a single scalar field) either \emph{auxiliary} fields, i.e.\ fields whose kinetic term contains only a mass term and no space-time derivatives; or else, they are \emph{unphysical} fields, because there are gauge transformations which shift their values at every point.  

To see this, let us  decompose the field $\mc{A}$ into $\mc{A}_{\zbar}$ and a doublet $\mc{A}_{\vbar_i}$. The component $\mc{A}_{\zbar}$ is a $(0,1)$ form on each twistor fibre twisted by $\Oo(-2)$, so we can view it as a $(1,1)$ form. Under the $SU(2)$ symmetry which rotates $\CP^1$, the bundle of $(1,1)$-forms is trivial so we can view $\mc{A}_{\zbar}$ as a scalar. 

The fields $\mc{A}_{\vbar_i}$ are sections of the anti-holomorphic co-normal bundle $\br{\Oo(-1)}$, tensored with $\Oo(-2)$. We can use a metric to identify $\br{\Oo(-1)}$ with $\Oo(1)$, so that $\mc{A}_{\vbar_i}$ are a doublet of sections of $\Oo(-1)$. 

We can decompose $\mc{A}_{\zbar}$ into its Fourier modes $\mc{A}_{\zbar}^{(2j)}$, which lives in the  spin $2j$ representations of $SU(2)$ for $j \ge 0$.  Similarly we can decompose $\mc{A}_{\vbar_i}$ into Fourier modes $\mc{A}_{\vbar_i}^{(2j+\tfrac{1}{2})}$, which live in the half-integer spin representations of $SU(2)$.  We can of course do the same with the generator of gauge transformations $\chi$, whose Fourier components $\chi^{(2j)}$ live in the integer spin representations of $SU(2)$ with spin $\ge 1$.  

The fields $\mc{A}_{\vbar_1}$, $\mc{A}_{\vbar_2}$ are coupled only by a $z$-derivative, and not by derivative on $\R^4$. The fields $\mc{A}_{\vbar_i}$, $\mc{A}_{\zbar}$ are coupled by a derivative on $\R^4$. Therefore the Lagrangian is schematically of the form
\begin{equation} 
	\sum_j \int \mc{A}_{\vbar_1}^{(2j+\tfrac{1}{2})} \mc{A}_{\vbar_2}^{(2j+\tfrac{1}{2})} + \sum_j \mc{A}_{\vbar_i}^{(2j+\tfrac{1}{2})} D \mc{A}_{\zbar}^{2j } + \sum_j \mc{A}_{\vbar_i}^{(2j+\tfrac{1}{2})} D \mc{A}_{\zbar}^{2j+2 }   
\end{equation}
where $D$ indicates some derivative on $\R^4$.  

The fields $\mc{A}_{\vbar_i}^{(2j+\tfrac{1}{2})}$ are an infinite tower of auxiliary fields -- that is, fields with only a mass and not a kinetic term.  They can be integrated out without destroying locality, because the propagator is a $\delta$-function.  

The fields $\mc{A}_{\zbar}^{(2j)}$ for $j > 0$ are unphysical, as the are shifted by modes of the gauge transformation $\chi^{(2j)}$:
\begin{equation} 
	\delta \mc{A}_{\zbar}^{(2j)} = \chi^{(2j)}. 
\end{equation}
It is important that there is no space-time derivative appearing in this gauge variation. We can set all the fields $\mc{A}_{\zbar}^{(2j)}$ for $j > 0$ to zero by a gauge choice, leaving us with simply the zero mode $\mc{A}_{\zbar}^{(0)}$, which is the scalar field $\phi$. The $\phi \Lap \phi$  Lagrangian on $\mc{A}_{\zbar}^{(0)}$ appears from integrating out the fields $\mc{A}_{\vbar_i}^{(1/2)}$. 

This completes the proof that the theory on $\R^4$ corresponding to a theory on twistor space is just a scalar field, without an infinite tower of Kaluza-Klein modes.

\subsection{ Moving from twistor space to $\R^4$ using factorization algebras  }  

Passing from a field theory on twistor space to one on $\R^4$ can be done in a simple and clean way using the language of factorization algebras \cite{costello2021factorization} .  Let us first briefly recall what a factorization algebra is, and how they appear in field theory.

A factorization algebra $\mc{F}$ on a manifold $M$ is a structure which assigns to every open subset $U \subset M$ a graded vector space $\mc{F}(U)$.  These graded vector spaces are related for different opens as follows. First, we have linear maps $\mc{F}(U) \to \mc{F}(U')$ if $U \subset U'$, so that $\mc{F}$ is a pre-cosheaf.  We also have a \emph{factorization product} 
\begin{equation} 
	\mc{F}(U) \otimes \mc{F}(V) \to \mc{F}(W) 
\end{equation}
whenever $U$, $V$ are disjoint and contained in $W$.  These structures satisfy natural commutativity and associativity constraints detailed in \cite{costello2021factorization}, together with a ``gluing'' relation expressing $\mc{F}(U)$ in terms of the values of $\mc{F}$ on a sufficiently fine open cover of $U$.  The gluing relation will not be important for us.

Any classical or quantum field theory on a manifold $M$ gives rise to a factorization algebra.  At the classical level, the factorization algebra $\Obs^{cl}$ of classical observables sends $U$ to the algebra of functions on the space of solutions to the equations of motion of $U$. If we are dealing with a gauge theory, we take gauge-invariant functions, where gauge invariance and the equations of motion are both defined homologically, and then we take cohomology \footnote{ We are abusing notation slightly hear: what we refer to as $\Obs^{cl}$ and $\Obs^q$ are the cohomology of the cochain complexes referred to as $\Obs^{cl}$ and $\Obs^q$ in \cite{costello2021factorization}.}.

At the quantum level, $\Obs^q(U)$ is given by the functionals of the field which only depend on its behaviour on $U$, modulo divergences with respect to the functional measure.  Removing divergences is the quantum version of working on-shell, because the classical limit of the divergence associated to a field redefinition is simply the corresponding Euler-Lagrange equation.  Again, we should also impose gauge invariance homologically.  

Of course, defining the divergence with respect to the functional measure is technically difficult and requires renormalization.  This is done in \cite{costello2011renormalization, costello2021factorization}.

The factorization algebras of classical field theories are \emph{commutative}: for every open, $\Obs^{cl}(U)$ is a commutative algebra and the structure maps are maps of commutative algebras.  The factorization algebras of quantum field theories do not have this property.  To leading order, the deformation away from commutative can be measured by a kind of Poisson bracket, built from the one-loop OPE (see \cite{costello2021factorization}).  Factorization algebras thus provide a deformation-quantization perspective on constructing quantum field theory. 

The language of factorization algebras generalizes standard algebraic constructions in field theory.  For quantum mechanics, the factorization algebra encodes the usual associative algebra of operators; for two-dimensional chiral CFTs, factorization algebras are vertex algebras \cite{costello2021factorization, 2012.12214}; for Lorentzian theories \cite{1711.06674, 1903.03396}, factorization algebras reproduce the axioms of algebraic quantum field theory; and for Euclidean theories, factorization algebras encode the OPE \cite{costello2021factorization}.   

If $g : M \to N$ is a smooth map between manifolds, and $\mc{F}$ is a factorization algebra on $M$, then we can build a factorization algebra on $N$ by the formula
\begin{equation} 
	(g_\ast \mc{F})(U) = \mc{F}(g^{-1}(U)). 
\end{equation}
This expression for the push-forward factorization algebra is the same as that for the push-forward of sheaves.

A classical or quantum field theory on twistor space gives rise to a factorization algebra on twistor space, say $\Obs^{cl}_{\PT}$ or $\Obs^{q}_{\PT}$.  The corresponding factorization algebra on $\R^4$ is simply the push-forward along the fibration $\pi : \PT \to \R^4$. 

The free field theory on twistor space with field as in \eqref{eqn_fft_twistor} corresponds, under this push-forward, to a free scalar field theory.  At the classical level, this follows from a result of Penrose \cite{Penrose:1969ae}, as we will now see. 

The factorization algebra on twistor space assigns to any open subset $V \subset \PT$ the vector space
\begin{equation} 
	\Obs^{cl}_{\PT}(V) = \Oo(   H^\ast_{\dbar}(V, \Oo(-2))[1] ).
\end{equation}
Here $\Oo$ indicates ``functions'': so $ \Oo(   H^\ast_{\dbar}(V, \Oo(-2))[1] )$ is the symmetric algebra\footnote{One must be careful to complete the tensor product appropriately when taking the symmetric algebra.  This point is discussed in detail in \cite{costello2021factorization}.}  of the dual of $H^\ast_{\dbar}(V, \Oo(-2))[1]$. The symbol $[1]$ indicates a shift of cohomological degree, so that classes of $(0,1)$-forms are placed in degree $0$.

Similarly, for an open subset $U \subset \R^4$, the factorization algebra of classical observables of the free scalar field theory is 
\begin{equation} 
	\Obs^{cl}_{\R^4}(U) = \Oo ( \op{Harm}(U) ) 
\end{equation}
where $\op{Harm}(U)$ is the vector space of harmonic functions on $U$.

Penrose showed that there is an isomorphism
\begin{equation} 
	\op{Harm}(U) \iso H^\ast_{\dbar}(\pi^{-1}(U), \Oo(-2) ) [1]. 
\end{equation}
This isomorphism implies that on the right hand side only $H^1_{\dbar}(\pi^{-1}(U), \Oo(-2))$ is non-vanishing.   Penrose's isomorphism is compatible with inclusions of open subsets (it is part of an isomorphism of sheaves).  This immediately implies that we have
\begin{equation} 
	\pi_\ast \Obs^{cl}_{\PT} = \Obs^{cl}_{\R^4}.  
\end{equation}
A small variation of this construction shows that
\begin{equation} 
	\pi_\ast \Obs^q_{\PT} = \Obs^q_{\R^4} 
\end{equation}
i.e. that the push forward of the quantum factorization algebra from twistor space to $\R^4$ on twistor space gives rise to the quantum factorization algebra for a free scalar.  

For interacting theories, the push forward of the quantum factorization algebra from twistor space to $\R^4$ can be taken as a definition of the quantum theory on $\R^4$ corresponding to that on twistor space.

The theory on twistor space we focus on most in this paper is the type I topological string. In the introduction, we stated that this becomes the $WZW_4$ model for $SO(8)$ together with a K\"ahler potential field. In the language of factorization algebras, this means that the push-forward of the factorization algebra of the type I topological string is that of $WZW_4$, plus the K\"ahler potential. We will prove this statement in sections \ref{section_4d_computation}. 

At the quantum level, this push-forward provides a definition of the quantum factorization algebra of $WZW_4$ plus the K\"ahler potential.

\subsection{Holomorphic theories on twistor space and periodic RG flow}

In this section we will prove that for any holomorphic theory on twistor space, the corresponding theory on $\R^4$ has periodic RG flow.   

Intuitively, this is clear: the scaling action of $\R_{> 0}$ complexifies to a $\C^\times$ action on twistor space, scaling the fibres of the map $\Oo(1)^2 \to \CP^1$.  For any local theory on twistor space, we can apply an element $\mu \in \C^\times$ to find a family of such theories. For a holomorphic theory, this must depend holomorphically on $\mu$, so that the renormalization group flow analytically continues to $\C^\times$.

Something very similar happens for a chiral theory in dimension $2$.  Typically, such theories are classically scale-invariant, and there are never any logarithmic divergences which spoil this. 

We find it worthwhile to carefully phrase this intuitive argument as a theorem. 

A classical theory on a complex manifold is \emph{holomorphic}  if the space of fields, in the BV formalism, is of the form
\begin{equation} 
	\Omega^{0,\ast}(\PT, E)[1]	 
\end{equation}
where $E$ is some graded holomorphic vector bundle equipped with a non-degenerate symmetric pairing of holomorphic bundles
\begin{equation} 
	\ip{-,-}:	E \otimes E \to K_{\PT}. 
\end{equation}
This map extends to a map 
\begin{equation} 
	\ip{-,-}:	\Omega^{0,\ast}(\PT, E) \otimes \Omega^{0,\ast}(\PT, E) \to \Omega^{3,\ast}(\PT). 
\end{equation}
This map, when composed with integration, is the BV odd symplectic pairing.

We require that the kinetic term in the Lagrangian is 
\begin{equation} 
	\int \ip{e, \dbar e}  
\end{equation}
for $e \in \Omega^{0,\ast}(\PT, E)[1]$. The interaction terms must only involve holomorphic derivatives and must simply wedge and integrate the $(0,\ast)$ forms appearing in the fields.

Under these hypothesis, the linearized BRST operator is the $\dbar$ operator.  

These conditions can be relaxed a little to include the case when the BV anti-bracket is degenerate, as it is for Kodaira-Spencer theory \cite{1611.00311, 2009.07116} .  We still require that the linearized BRST operator is $\dbar$, and interactions are as above.

The symmetries of such a theory include all  anti-holomorphic vector fields $V$ on $\PT$.  (The fact that $E$ is a holomorphic bundle is enough data to define their action on the space of fields). Contraction $\iota_{V}$ with an anti-holomorphic vector field $V$ is an operator of ghost number $-1$ on the fields which preserves the interaction terms.  The Cartan homotopy formula,
\begin{equation} 
	[\dbar, \iota_{V}]  = \mc{L}_{V}, 
\end{equation}
tells us that the Lie derivative $\mc{L}_V$ is BRST exact.

Now let us turn to  proving that the renormalization group flow is periodic for a holomorphic theory on twistor space.  We will first prove that no holomorphic theory on twistor space can give rise to a theory on $\R^4$ with a one-loop log divergence. This is a consequence of the following lemma.
\begin{lemma}
	Consider any classical holomorphic theory on a complex manifold $X$. Suppose we have a complex Lie algebra $\g_{\C}$ which acts holomorphically on $X$, and where the action is lifted to one on the space of fields of the theory, preserving the Lagrangian. 

	Then, the obstruction to having a holomorphic action of $\g_{\C}$ on the theory  at one loop, is the same as the obstruction to having an action of any real form of $\g_{\C}$.
\end{lemma}
\begin{proof}
	Consider its complexification
	\begin{equation} 
		\g_{\C} \otimes_{\R} \C = \g^{1,0} \oplus \g^{0,1}. 
	\end{equation}
	Classically, a holomorphic action of $\g_{\C}$ means a complex-linear  action of  $\g_{\C} \otimes_{\R} \C$ with a homotopy trivialization of the action of the anti-holomorphic piece $\g^{0,1}$.   The homotopy trivialization of the anti-holomorphic piece is implemented by asking that we have an action of the dg Lie algebra
	\begin{equation} 
		\g_{\dbar} = \g^{0,1}[1] \oplus (\g^{1,0} \oplus \g^{0,1})  
	\end{equation}
	where the differential cancels the copy of $\g^{0,1}$ in degree $-1$ with that in degree $0$.

	Let $\mc{E} = \Omega^{0,\ast}(\PT, E)$ denote the space of fields of our holomorphic theory, and $\Ool(\E)$ the cochain complex of local functionals (that is, Lagrangians up to total derivative, where the cohomological degree is ghost number).  In \cite{costello2011renormalization}, it is shown that for any Lie algebra $\mf{l}$ acting on a theory (including a possible action on space-time), obstructions to lifting this to an action on the quantum theory are in the cohomology group
	\begin{equation} 
		H^1_{red}(\mf{l}, \Ool(\E)). 
	\end{equation}
	Here the subscript indicates we are taking reduced Lie algebra cohomology. This is the cohomology of the cochain complex $\Ool(\E) \otimes \Sym^{> 0}(\g^\vee[-1])$ equipped with the sum of BV differential $\{S,-\}$ on $\Ool(\E)$ and the Chevalley-Eilenberg Lie algebra cohomology differential.

	To give a holomorphic action of $\g_{\C}$ is the same as to give an action of $\g_{\dbar}$. The anomaly to having such an action is controlled by the Lie algebra cohomology group
	\begin{equation} 
		H^1_{red}(\g_{\dbar}, \Ool(\E))		\label{eqn:obstruction_dbar} 
	\end{equation}
	Now fix any real form $\g_{\R}$ of $\g_{\C}$. There is a quasi-isomorphism
	\begin{equation} 
		\g_{\R} \otimes_{\R} \C \to \g_{\dbar}. 
	\end{equation}
	The obstruction to having an action of $\g_{\R}$ at one loop is (since our theory is defined over the field $\C$), an element in
	\begin{equation} 
		H^1_{red}(\g_{\R} \otimes_{\R} \C, \Ool(\E)). 
	\end{equation}
	This group is isomorphic to the group \eqref{eqn:obstruction_dbar} controlling obstructions to having a $\g_{\dbar}$ action.  

	Having a $\g_{\dbar}$ action implies having a $\g_{\R}$ action. Because the natural maps between the  obstruction groups is an isomorphism, this implies the converse is true.

	We note that nowhere in this proof did we assume that there are no gauge anomalies in our theory at one loop. It makes sense to consider the presence of anomalies to having a symmetry even in the presence of gauge anomalies.   
\end{proof}

Now we will use this to prove the following.
\begin{theorem} 
	Consider any classical holomorphic theory on $\PT$, invariant under the $\C^\times$ action which scales the $\Oo(1)^2$ fibres.  

	Then the corresponding theory on $\R^4$ is scale invariant at one loop. 
\end{theorem}
\begin{proof}
	To prove this, it suffices to show that there is no anomaly at one loop to the theory  having an action of $\C^\times$.  If there is, then in particular the theory at one loop is preserved by the $\R_{> 0} \subset \C^\times$ symmetry, which means that the theory on $\R^4$ is scale invariant.  

	Here we are using the interpretation (perhaps not widely used) of the one-loop RG flow as being an anomaly. It is the anomaly to the classical scaling symmetry persisting at the quantum level. 

	From the previous lemma, the obstruction to have an $\R_{> 0}$ symmetry at one loop (i.e.\ the RG flow) is the same as the obstruction to having a holomorphic $\C^\times$ action, and the same as the obstruction to having an action of $S^1$.

	There are never any perturbative anomalies to having an action of a compact group $G$. This implies that the anomaly to having a holomorphic $\C^\times$ action vanishes. 

	Since the statement that perturbative anomalies vanish for compact groups  may be in conflict with some perspectives on QFT, let me explain this point.   By an \emph{action} of the group on the theory I simply mean the most naive thing: the compact group acts as transformations of the collections of fields, preserving the renormalized amplitudes of Feynman diagrams.   

	We can always arrange for a classical action of a compact group $G$ to persist to the quantum level, simply because we can always choose both a gauge and a cut-off compatible with the $G$-action. (In the case of a holomorphic theory on $\PT$, the choice of gauge comes from a metric on $\PT$ and we can arrange for it to be $S^1$-invariant). If we do this, then all amplitudes with a UV cut-off are $G$-invariant.  Introducing counter-terms and sending the cut-off does not change this, as we have at no point broken the $G$-symmetry.  In general, in perturbation theory, anomalies only arise when we have some symmetry that does not preserve the gauge choice or the cut-off. 

This argument does not contradict the fact that there may be 't Hooft anomalies.  These are different: these are obstructions to having an inner action, in the terminology of \cite{costello2021factorization}.  Nor does this argument contradict other occurrences of anomalies for compact groups, where the anomaly arises from non-perturbative effects.     
\end{proof}

This theorem implies that most familiar theories, such as Yang-Mills theory, can not be constructed from a local theory on twistor space. 

We can also use these methods to prove that the theory on $\R^4$ built from our twistor-string theory has periodic RG flow. In this case, the string coupling constant $\lambda$ has charge $-2$ under the $\C^\times$ action on twistor space which scales the $\Oo(1)^2$ fibres.  We can see this because the holomorphic Chern-Simons action appears with a factor of $\lambda^{-1} \Omega$ where $\Omega$ is the meromorphic volume form, and $\Omega$ has charge $-2$.     

We can ask that the holomorphic $\C^\times$ action on the classical theory (giving $\lambda$ charge $-2$) persists to the quantum level.  We will give a proof of this by obstruction theory in section \ref{section_twistor_quantization}:
\begin{proposition}
	Holomorphic Chern-Simons for $SO(8)$ coupled to type I Kodaira-Spencer theory gives rise to quantum theory on $\PT$ with a holomorphic action of $\C^\times$ covering the action of $\C^\times$ on $\PT$ scaling the $\Oo(1)^2$ fibres.	
\end{proposition}

This implies one of the main results of this paper, that $WZW_4$ for gauge group $SO(8)$, coupled to the gravitation theory described above, has a quantization with periodic RG flow.  

At this point I should point out a mistake I made in the talk \cite{WH}.  There, I asserted that for $G \neq SO(8)$, there was a one-loop log divergence.  This is incorrect, as pointed out by Greg Moore and Samson Shatashvili.   In their paper \cite{Losev:1995cr} with Losev and Nekrasov, they showed that there are no one-loop divergences for any group.   

This correct statement is consistent with the general picture one gets from twistor space.  A classical theory which is local on twistor space can not have any one-loop log divergences, as we saw above. However, suppose  the theory has a one-loop gauge anomaly, that can be canceled by a non-local term on twistor space.  This non-local term at one loop can contribute to two-loop, and higher, log divergences.

The fact that non-local terms on twistor space lead to log divergences further on in the loop expansion is clear from the example of $\phi^4$, where the $\phi^4$ interaction is a non-local expression on twistor space.  Two copies of this non-local vertex lead to the familiar one-loop logarithmic divergence of $\phi^4$ theory.

This discussion tells us that we would expect $WZW_4$, for $G \neq SO(8)$, to have a two-loop or higher logarithmic divergence.  And, for $G = SO(8)$, such a divergence should be canceled by a Green-Schwartz mechanism. 

I was unable to determine whether or not two-loop log divergences arise, but I hope to revisit this question.

\subsection{Analytic continuation of observables}
Quantum field theories which arise from local theories on twistor space have another remarkable property: the correlation functions extend analytically to meromorphic functions $\C^4$, with poles when the locations of the operators are lightlike separated.  

Intuitively, this is clear.   Let us suppose we have a classical or quantum holomorphic theory on $\PT$.   The translation action of $\R^4$ lifts to a holomorphic action of $\C^4$ on $\PT$. These are the holomorphic vector fields pointing along the fibres of the fibration $\PT = \Oo(1)^2 \to \CP^1$.    Let us suppose that our theory on $\PT$ has a holomorphic action of $\C^4$.  

A local operator $O_1$  at the origin in $\R^4$ corresponds to an operator in the holomorphic theory on twistor space localized at the twistor line $\CP^1 \subset \PT$ corresponding to the zero section of the fibration $\Oo(1)^2 \to \CP^1$.   By using the action of $\C^4$ we can move the operator $O_1$ to any twistor line $\CP^1 \subset \PT$. The collection of such $\CP^1$'s is $\C^4$, so that we get an operator $O_1(x)$ corresponding to $x \in \C^4$.

If $O_i$ is a collection of local operators in the theory on $\R^4$, then we can place $O_i$ at $x_i \in \C^4$. As long as the corresponding twistor lines $\CP^1_{x_i}$ do not touch, we can define the correlation functions
\begin{equation} 
	\ip{O_1(x_1) \dots O_n(x_n)} . 
\end{equation}
(Of course, the correlation functions depend on a state at $\infty$). 

The twistor lines $\CP^1_{x_i}$, $\CP^1_{x_j}$ touch whenever $\norm{x_i - x_j}^2 = 0$ (using the \emph{holomorphic} metric on $\C^4$).    Thus, the correlation functions $\ip{O_1(x_1) \dots O_n(x_n)}$ are defined whenever the $x_i$ are such that no two are null separated. 

This shows that, for any local theory on twistor space, the correlation functions (and OPEs) extend analytically to meromorphic functions on $(\C^4)^n$, with poles on the divisors $D_{ij}$ where  $\norm{x_i - x_j}^2$. 

This observation gives another very strong constraint on which theories can arise from twistor space.  It even constrains which classical theories on $\R^4$ can arise as local theories on  twistor space.  This is because the one-loop correlation functions between local operators only depend on the classical Lagrangian.  Any one-loop counter-terms, such as non-local counter-terms introduced to cancel an anomaly, can only contribute to two-loop correlation functions.

We use this later in the calculation of Lagrangian on $\R^4$ corresponding to Kodaira-Spencer theory on twistor space.  In principle this Lagrangian can be computed directly, but we find it more efficient to constrain the terms in the Lagrangian by asking that the one-loop OPE extends analytically to $\C^4$.   

\subsection{Factorization algebra formulation of analytic continuation}
It is helpful to formulate these ideas in the language of factorization algebras.   This will allow us to prove a theorem about the analytic continuation of the factorization algebra of the quantum theory corresponding to our type I twistor string theory. 

First, let us make a preliminary definition:
\begin{definition} 
	Two open subset $U,V$ in $\C^4$ are \emph{causally disjoint} if there does not exist a null straight line $\C \subset \C^4$ which intersects both $U$ and $V$.  
\end{definition}
\begin{definition}
	A \emph{causal prefactorization algebra} on $\C^4$ is a structure that assigns to every open subset $U \subset \C^4$, a graded vector space $\mc{F}(U)$; and to every collection $U_1, \dots, U_n \subset \C^4$ of opens, which are pairwise causally disjoint and contained in an open $V$, a linear map
	\begin{equation} 
		\mc{F}(U_1) \otimes \dots \mc{F}(U_n) \to \mc{F}(V). 
	\end{equation}
	These satisfy the same axioms as in \cite{costello2021factorization}, except with the word ``disjoint'' replaced everywhere by ``causally disjoint''.
\end{definition}
In categorical language, one can define a multi-category whose objects are open subsets of $\C^4$, and where there is a morphism $(U_1,\dots,U_n) \to V$ if $U_i$ are pairwise causally disjoint and contained in $V$. A causal factorization algebra is a functor of multicategories from this multicategory to that of graded vector spaces and multi-linear maps. One can contemplate adding a local-to-globed axiom to the definition of causal prefactorization algebra, to define a causal factorization algebra, but we will not do so here.

Causal prefactorization algebras on $\R^4$ are very closely related to the nets of observables of algebraic quantum field theory \cite{1711.06674, 1903.03396}.   The connection is roughly as follows.  For a region $U$ in space, the associative algebra $\mc{A}(U)$in the AQFT picture is the vector space the factorization algebra assigns to a small open $\what{U}$ in space-time which contains $U$. 

The factorization product for time-like separated opens associative product structure. The compatibility between the product in the space direction and the time-like product was shown in  \cite{1711.06674, 1903.03396} to imply that the algebras $\mc{A}(U)$, $\mc{A}(U')$ commute with each other when we map them to $\mc{A}(U'')$, for $U, U'$ disjoint opens in $U''$.  This is the causal axiom of algebraic quantum field theory, that operators in disjoint regions of space commute with each other.

\begin{lemma}
	Every factorization algebra $\mc{F}$  on $\PT$ gives rise to a causal prefactorization algebra $\what{F}$ on $\C^4$.   When restricted to $\R^4 \subset \C^4$, this coincides with the push-forward $\pi_\ast \mc{F}$ along the fibration $\PT \to \R^4$. 
\end{lemma}
\begin{proof}
	Every open $U\subset\C^4$ gives rise to an open subset $\what{U} \subset \PT$, which is
	\begin{equation} 
		\what{U} = \bigcup_{x \in U} \CP^1_x 
	\end{equation}
	i.e.\ $\what{U}$ is the union of the twistor lines $\CP^1_x$ corresponding to $x \in U$.  

	Given a factorization algebra $\mc{F}$ on $\PT$, we get a factorization algebra $\what{\mc{F}}$ on $\C^4$ by setting
	\begin{equation} 
		\what{\mc{F}}(U) = \mc{F}(\what{U}). 
	\end{equation}
	Given opens $U,U' \subset \C^4$, the corresponding opens $\what{U}$, $\what{U}'$ are disjoint if and only if $U,U'$ are causally disjoint.  

	This means the factorization product is defined whenever $U, U'$ are causally disjoint.

	The restricting of $\what{\mc{F}}$ to $\R^4$ assigns to an open $V \subset \R^4$ the limit
	\begin{equation} 
		\what{\mc{F}}\mid_{\R^4}(V) = \op{lim}_{V \subset U} \what{\mc{F}}(U). 
	\end{equation}
	The limit is taken over smaller and smaller opens in $\C^4$ containing $V$. 

This is the same as the limit
	\begin{equation} 
		\what{\mc{F}}\mid_{\R^4}(V) = \op{lim}_{V \subset U} \mc{F}(\what{U}). 
	\end{equation}
	The limit of the opens $\what{U}$ as $U$ ranges over small neighbourhoods of $V$ is the same as $\pi^{-1}(V)$, so we get
\begin{equation} 
	\what{\mc{F}}\mid_{\R^4}(V) = \mc{F}(\pi^{-1}(V) ).  
	\end{equation}
\end{proof}
This lemma shows that any factorization algebra on $\PT$ gives rise to a factorization algebra in every signature. Indeed, it gives rise to a factorization algebra defined for the analytically continued family of metrics $e^{2 \pi \i \theta} g$ for all $\theta$.    

This fits well with our discussion about the periodic RG flow.  Performing the RG flow in an imaginary direction is the same as scaling the metric by a complex number of norm $1$.  The statement that the theory has periodic RG flow is then the statement that we can analytically continue the family of metrics to $e^{2 \pi i \theta} g$, where at $\theta = 0$ we get the same theory as at $\theta = 1$.  

Next, we will encode axiomatically the way in which the product of observables depends holomorphically on the position of the opens $U \subset \C^4$.  For these purposes, it is helpful to work not with arbitrary opens in $\C^4$, but with poly-discs. A poly-disc in $\C^4$ is the product of four in $\C$.  It is traditional in multi-variable complex analysis to use poly-discs, because the Dolbeault cohomology is known to vanish on a poly-disc. 

Thus, let $D(x,r) \subset \C^4$ be the poly-disc consisting of those $x' \in \C^4$ where $\abs{x_i - x'_i} < r$ for $i = 1,\dots, 4$.  

We let $\mc{F}(x,r)$ be the graded vector space that the factorization algebra assigned to $D(x,r)$. We say $\mc{F}$ is \emph{translation-invariant} if there are isomorphisms $\mc{F}(x,r) \iso \mc{F}(x',r)$ compatible with the factorization product, in the sense that if we translate all the discs involved in the factorization product the result does not change.  

This allows us to write $\mc{F}_r$ for $\mc{F}(x,r)$ for any $x$.

Let
\begin{equation} 
	\op{PDisc}(r^1,\dots,r^n \mid R) \subset (\C^4)^n  
\end{equation}
be the open subsets consisting of $n$ points $x^1,\dots, x^n \in \C^4$ such that the poly-discs $D(x^i, r^i)$ are causally separated and contained in $D(0,R)$.   The factorization product defines a map
\begin{equation} 
	\mc{F}_{r^1} \otimes \dots \otimes \mc{F}_{r^n} \times \op{PDisc}(r^1,\dots,r^n \mid R) \to \mc{F}_R  
\end{equation}
which is linear in all the elements of $\mc{F}_{r^i}$. We say it is \emph{holomorphically translation invariant} if this map depends holomorphically on the point in $\op{PDisc}(r^1,\dots,r^n \mid R)$.

The structure of holomorphic causal factorization algebra on $\C^4$ is in many ways similar to that of a vertex algebra.  Without the causal condition, where the discs are required to be simply disjoint, holomorphic factorization algebras are much more subtle.  This is because Hartog's theorem tells us that the OPE, which is a function on $\C^4 \setminus 0$, extends across the origin.   With the causal condition, we have poles on the light-cone, which is a divisor. Thus, the OPE is a meromorphic function on $\C^4$ whose denominator is a power of $\norm{x}^2$.    

The group $\C^4$ acts on $\PT$.  Given a classical holomorphic theory on $\PT$, with an action of $\C^4$, then it is automatic that the corresponding causal factorization algebra on $\C^4$ is holomorphic.  

We can ask whether this happens at the quantum level. For the twistor string theory which is our main object of study, this is true.
\begin{proposition} 
	The factorization algebra on $\PT$ of quantum observables of $SO(8)$ holomorphic Chern-Simons coupled to type I Kodaira-Spencer theory has a holomorphic action of the group $\C^4$, and so gives a holomorphic causal factorization algebra on $\C^4$.  
\end{proposition}
We will prove this in the appendix \ref{section_twistor_quantization} as part of the proof that this theory is well-defined at the quantum level.

The factorization algebra by itself encodes the operator product.  With a little more data, it also defines the correlation functions. The extra data is a state, which is a linear map 
\begin{equation} 
	\ip{} : 	\mc{F}(\R^4) \to \C 
\end{equation}
which is invariant under translations and any other symmetries of space-time the theory possesses.  For the correlation functions to extend analytically, we need the state to extend to a linear map
\begin{equation} 
	\ip{}: \mc{F}(\C^4) \to \C 
\end{equation}
which is again invariant under translations. In our context, there is no difficulty in building a state like this, although there are many different such states leading to different correlation functions.   The OPEs, however, are universal and independent of the state.

\section{Computing the four-dimensional theory associated to the type I topological string}
\label{section_4d_computation}

In section \ref{section_typeI_intro} we discussed holomorphic Chern-Simons theory with a meromorphic volume form on twistor space.  In this section we will show that the corresponding theory on $\R^4$ is $WZW_4$, plus a K\"ahler potential.

Let us start with the holomorphic Chern-Simons theory field: 
\begin{theorem}
	If $\Obs^{cl}_{\PT,\g}$ be the factorization algebra associated to classical holomorphic Chern-Simons theory on twistor space, with three-form $\Omega = \d v_1 \d v_2 \d z /z^2$ as above and with gauge field vanishing at $z = 0$, $z = \infty$.  
	Then, $\pi_\ast \Obs^{cl}_{\PT,\g}$ is the classical factorization algebra associated to the $WZW_4$ theory with gauge group $G$, where we work in perturbation theory around the identity in $G$. 
\end{theorem}
\begin{proof}
	This result was stated in \cite{WH}.  A proof was given in \cite{2011.04638, 2011.05831}.  So we will not give all the details here. 

	We pause to note that what we will find naturally is the analytically-continued version of $WZW_4$, involving a field which is a map to  the complex group $G$, and not to any real form.  This is because holomorphic Chern-Simons is an analytically-continued theory \cite{Witten:2010zr}  , where the space of fields is a complex manifold and the path integral is taken over a contour.  In perturbation theory, the choice of contour doesn't matter.  Because of this, the corresponding theory on $\R^4$ will also be an analytically continued theory.  In the language of factorization algebras, this simply means we have a factorization algebra over $\C$ and not $\R$.

Let us start by discussing the limit where the Lie algebra $\g$ is made Abelian. Then holomorphic Chern-Simons theory becomes free, and the fields live in $\Omega^{0,1}(\PT, \Oo(-2)) \otimes \op{dim} \g$.  The Penrose transform immediately implies that the push-forward of the factorization algebra of classical observables to $\R^4$ gives a free scalar field valued in $\mf{g}$.

	The content of the result is that the deformation coming from making $\g$ non-Abelian is precisely the $WZW_4$ Lagrangian. 

	As a first step in this direction, we note that we can restrict the holomorphic Chern-Simons $(0,1)$-form to every $\CP^1$ in twistor space. There, it describes a holomorphic bundle trivialized at $0$ and $\infty$. The moduli space of such bundles is $G$.

	Since points $x \in \R^4$ give rise to $\CP^1$'s in twistor space, the holomorphic Chern-Simons gauge field gives rise to a map $\sigma : \R^4 \to G$.  Further, this map entirely encodes the gauge equivalence class of holomorphic Chern-Simons gauge field on each twistor fibre.      As in \cite{Costello:2019tri}, we can express $\mc{A}_{\zbar}$ in terms of $\sigma$ by choosing a gauge, then use the equations of motion $F_{\zbar \vbar_i} = 0$ to find $\mc{A}_{\vbar_i}$ in terms of $\sigma$, and finally plug the resulting expressions into the Chern-Simons action to calculate a Lagrangian for $\sigma$.  This Lagrangian was determined in \cite{WH,2011.04638,2011.05831} to be  $WZW_4$. 

	We will not give all the details, but it will be helpful in what follows to have an explicit representative for some components of the field $\mc{A}$ on twistor space corresponding to $\sigma$.  If $\sigma$ satisfies the equations of motion, then we should build, for each $z \in \CP^1$, a holomorphic bundle on the $\CP^1$ living over $z$. This bundle should be trivialized at $0,\infty$.  The formula for the $(0,1)$ form in the complex structure $z$ is  
	\begin{equation} 
		 \pi_z^{0,1} J^{0,1}.  
	\end{equation}
	Here, $J^{0,1}$ is the $(0,1)$ part of the current for $\sigma$ on $\R^4$, in the complex structure at $z = 0$. We pull this back to be a $1$-form on twistor space and then project this to be a $(0,1)$ form on the fibre of $z$.    Our conventions are such that the complex structure at $z = 0$ is the one we use on $\R^4$. 

	Then this bundle at $z = 0$ is trivialized by $\sigma$, and at $z = \infty$ it is trivial, as there we have the opposite complex structure to that at $z = 0$ so $\pi^{0,1} J^{0,1} = 0$.

	One can check in coordinates that, when restricted to the complex plane living over $z$, we have 
	\begin{equation} 
		\pi_z^{0,1} J^{0,1} = \frac{1}{1 + \abs{z}^2} \d \vbar_i J_{\ubar_i}  	 
	\end{equation}
	The equations of motion of $WZW_4$ tell us that $\partial_{u_i} J_{\ubar_i} = 0$. This implies that for each value of $z$, $\pi^{0,1} J^{0,1}$ defines a holomorphic bundle in that complex structure.

This tells us that, away from a small neighbourhood of $z = 0$, and dropping terms with a $\d \zbar$, we can write the gauge field on twistor space as
	\begin{equation} 
		\mc{A} = \pi_1^{0,1} J^{0,1}.  
	\end{equation}

\end{proof}

\subsection{The push-forward of the closed-string fields}

What of the closed-string fields? The closed-string field on twistor space is a $(2,1)$ form $\eta \in \Omega^{2,1}$ which is $\partial$-closed.  Since $\eta$ is $\partial$-closed, we can locally write $\eta = \partial \gamma$, where $\gamma$ is a $(1,1)$-form. Then, the integral of $\gamma$ over the fibres of the map $\PT \to \R^4$ gives rise to a scalar $\rho$. 

The factorization algebra on $\R^4$ associated to the closed-string fields has the following description.
\begin{theorem} 
	The push-forward to $\R^4$ of the factorization algebra associated to the closed-string fields is a \emph{subalgebra} of the factorization algebra of classical observables associated to a scalar field $\rho$ with the Lagrangian
\begin{equation} 
	\int \omega^2 \rho \Lap^2 \rho +  (\Lap \rho)  (\partial \dbar \rho)^2  - \frac{2}{3} (\Lap \rho)^3 \omega^2 + O(\rho^3). 
\end{equation}
	(where I do not currently know the quartic and higher terms, but they all only depend on the derivatives of $\rho$). 

The subalgebra is that given by functions of the closed one-form $\d \rho$, i.e.\ by operators that only involve derivatives of $\rho$.
\end{theorem}
The proof of this will take the next several sections.  The fact that we only find a sub-algebra is worth elaborating on. 

The reason is that $\rho$ is built from $\partial^{-1} \eta$, where $\eta$ is the $(2,1)$-form on twistor space.  The factorization algebra of observables built from just $\eta$ will give us quantities which measure the closed one-form $\d \rho$, which is obtained from transgressing $\eta$ along twistor $\CP^1$'s.  
We could write the Lagrangian in terms of the closed $1$-form $\alpha = \d \rho$, which is constrained to be closed. Indeed it is easy to see that the Lagrangian is a function only of the closed $(1,1)$ form $B=\partial \dbar \rho$. Then the theorem states that the push-forward of the factorization algebra on twistor space is the algebra of classical observables of this constrained system.

In the appendix \ref{section_twistor_quantization}, generalizing slightly the results of \cite{1905.09269}, we show that if $G = SO(8)$ the coupled open-closed theory has a unique quantization on twistor space.  

By simply pushing forward the factorization algebra, we get a quantization of the theory on $\R^4$.

\subsection{Recollections about the Penrose transform}
Before diving into the details of the calculation, we need to recall some aspects of the Penrose transform, relating Dolbeault cohomology groups on twistor space with solutions to free-field equations on twistor space. As above, we will give twistor space $\PT = \Oo(1)^2 \to \CP^1$ coordinates $v_1,v_2$ on the fibres of the projection to $\CP^1$, and $z$ on the base.  If $x_1,x_2,x_3,x_4$ are coordinates on $\R^4$, and $u_1 = x_1 + \i x_2$, $u_2 = x_3 + \i x_4$, then the (non-holomorphic) isomorphism
\begin{equation} 
	\PT \iso \R^4 \times \CP^1 
\end{equation}
is obtained by sending $z \to z$ and setting
\begin{equation} 
	\begin{pmatrix}
		v_1 \\ v_2    
	\end{pmatrix}
	= \begin{pmatrix} 
		u_1 & \br{u}_2 
		\\ u_2 & - \br{u}_1 
	\end{pmatrix}
	\begin{pmatrix} 
		1 \\ z
	\end{pmatrix}	
\end{equation}
The complex conjugates $\br{v}_i$ are obtained  by conjugating this expression.

Under this coordinate change, 
\begin{equation} 
	\begin{split}
		\partial_{u_1} &=  \partial_{v_1} - \zbar \partial_{\vbar_2}   \\
		\partial_{u_2} &= \partial_{v_2} + \zbar \partial_{\vbar_1} \\
\partial_{\ubar_1} &= \partial_{\vbar_1} - z \partial_{v_2} \\
\partial_{\ubar_2} &= \partial_{\vbar_2} + z \partial_{v_1} 
	\end{split} \label{eqn_vectorfields}
\end{equation}

\subsection{ The closed string fields in the BV formalism}

Before we turn to calculating the $4d$ theory corresponding to the closed-string sector, we need to turn to describe how to introduce gauge symmetries in the closed string sector. This is best done by working in the BV formalism,  including ghosts and anti-fields.  This was done in the original paper of BCOV \cite{hep-th/9309140} and was also the point of view taken in \cite{1905.09269}.

To do this, we upgrade $\mu$ to a field 
\begin{equation} 
	\boldsymbol{\mu} \in \Omega^{0,\ast}(\PT, T \PT ( -2 D_0 - 2 D_\infty))[1]  
\end{equation}
Here $D_0,D_\infty$ are the divisors $z = 0, z = \infty$, where $\boldsymbol{\mu}$ vanishes to second order.   The symbol $[1]$ indicates a shift of ghost number, so that the component in $\Omega^{0,1}(\PT, T \PT)$ is of ghost number zero.   We impose the constraint $\partial \mu = 0$; it is most natural to impose this cohomologically \cite{1905.09269} but we will not do that right now.

The Lagrangian for the super-field $\boldsymbol{\mu}$ is the same:
\begin{equation} 
	\tfrac{1}{2}	\int \Omega \wedge ((\dbar \boldsymbol{\mu} \vee \partial^{-1} \boldsymbol{\mu} \vee \Omega) + \tfrac{1}{6} \int \Omega \wedge ( \boldsymbol{\mu} \vee \boldsymbol{\mu} \vee \boldsymbol{\mu} \Omega)  
\end{equation}
One can check that this satisfies the classical master equation.  

The gauge symmetries of this model correspond to the ghost fields in $\op{Ker} \partial \subset \Omega^{0,0}(\PT, T \PT(-2 D_0 - 2 D_\infty))$. These are vector fields on $\PT$, which only differentiate with respect to the holomorphic coordinates, which vanish to second order at $z = 0$, $z = \infty$, and which preserve the volume form $\Omega$.  These gauge symmetries act on all the other fields by Lie derivative.

The coupling to the open-string fields is most naturally written in the BV formalism too.  Let 
\begin{equation} 
	\boldsymbol{A} \in \Omega^{0,\ast}(\PT, \g(-D_0 - D_\infty))[1] 
\end{equation}
be the BV extension of the open-string field, including ghosts and antifields.  Let 
\begin{equation} 
	\boldsymbol{\eta} = \boldsymbol{\mu} \vee \Omega \in \op{Ker} \partial \subset \Omega^{2,\ast}(\PT ) [1] 
\end{equation}
be the $(2,\ast)$-form corresponding to $\boldsymbol{\mu}$. 

Then, the coupling to the open-string fields $\boldsymbol{A}$ takes the same form as before:
\begin{equation} 
	\int \boldsymbol{\eta} \boldsymbol{A} \partial\boldsymbol{A}.
\end{equation} 
Expanding this out into components will allow us to read off the action of gauge symmetries on various fields.  For instance, the terms involving the $(2,0)$ component of $\boldsymbol{\eta}$ (the ghost field) tell us how that the  closed-string gauge symmetries -- which are certain vector fields -- act on $\boldsymbol{A}$ by Lie derivative.

\subsection{The free closed-string theory on $\R^4$}
The first step in our analysis is to compute the $4d$ field corresponding to the free closed-string theory on twistor space.

Let us choose a $(1,1)$ form $\gamma$ on $\PT$ such that $\partial \gamma = \eta$.  Then, we can define a scalar on $\R^4$ by
\begin{equation} 
	\rho(x) = \int_{\CP^1_x} \gamma 
\end{equation}
where $\CP^1_x \subset \PT$ is the line corresponding to the point $x \in \R^4$.  If we change the chosen primitive $\gamma$, then $\rho$ changes by a constant. That is, the closed $1$-form $\d \rho$ on $\R^4$ only depends on $\eta$.

We will write the closed-string theory in terms of $\rho$, but it is important to bear in mind that only derivatives of $\rho$ are truly local operators on twistor space.  Indeed, we could write $\rho(x)$ in a non-local way as 
\begin{equation} 
	\rho(x) = \int_{M} \eta 
\end{equation}
where $M \subset \PT$ is a $3$-manifold with boundary on $\CP^1_x$.  

We will see that
\begin{equation} 
	\Lap^2 \rho = 0 
\end{equation}
where $\Lap$ is the Laplacian on $\R^4$. Note this is a \emph{fourth-order} equation.

We can see this by an explicit calculation on twistor space, using the representatives of the vector fields $\partial_{u_i}$, $\partial_{\ubar_i}$ given in equation \eqref{eqn_vectorfields}. We have
\begin{equation}
	(\Lap \rho)(x) = \int_{\CP^1_X} ( (\mc{L}_{ \partial_{v_2}}   \mc{L}_{z \partial_{v_1}}) \gamma  -   \mc{L}_{ \partial_{v_1}}   \mc{L}_{z \partial_{v_2}}) \gamma  . 
\end{equation}
Here $\iota_V$ is contraction with a holomorphic vector field on twistor space, and $\mc{L}_V$ is Lie derivative.
Using the fact that $\eta = \partial \gamma$ and the Cartan homotopy formula $[\partial, \iota_V] = \mc{L}_V$, we find
\begin{equation}
	(\Lap \rho) = \int_{\CP^1_x} ( (\mc{L}_{ \partial_{v_2}} \iota_{z \partial_{v_1}}) \eta  -   \mc{L}_{ \partial_{v_1}}   \iota_{z \partial_{v_2}}) \eta  \label{eqn_b+}. 
\end{equation}

An equivalent expression is 
\begin{equation} 
	\Lap \rho = \int \d z \iota_{\partial_{v_1}} \iota_{\partial_{v_2}} \eta \label{eqn_b+2} 
\end{equation}
To relate \eqref{eqn_b+2} to \eqref{eqn_b+} we use by integration by parts,
\begin{equation} 
	\begin{split}
		\int \d z \iota_{\partial_{v_1}} \iota_{\partial_{v_2}} \eta &=  -  \int z \partial ( \iota_{\partial_{v_1}} \iota_{\partial_{v_2}} ) \eta  \\
		&= - \int z \left( \mc{L}_{\partial_{v_1}} \iota_{\partial_{v_2}} - \iota_{\partial_{v_1}} \mc{L}_{\partial_{v_2}}    \right) \eta 
	\end{split}
\end{equation}
using the fact that $\eta$ is $\partial$-closed.  Bringing $z$ inside the brackets, and using the fact that $\mc{L}_{\partial_{v_i}} z = 0$, this gives us the expression for $b_+$ in equation \eqref{eqn_b+}.

Now let us use the  expression \eqref{eqn_b+2} for $\Lap \rho$ to show that it is harmonic. We have
\begin{equation} 
	\partial_{u_1} \partial_{\ubar_1} \Lap  = -\int_{z} \mc{L}_{\partial_{v_1}} \mc{L}_{z \partial_{v_2}}  \iota_{\partial{v_1}} \iota_{\partial{v_2}} \eta \d z.  
\end{equation}
We have the identity
\begin{equation} 
	\mc{L}_{z \partial_{v_1}} \iota_{\partial_{v_1}}=z \mc{L}_{ \partial_{v_1}} \iota_{\partial_{v_1}}  
\end{equation}
and similarly for $\partial_{v_2}$.   This means
\begin{equation} 
	\partial_{u_1} \partial_{\ubar_1} \Lap \rho = -\int_{z} z \mc{L}_{\partial_{v_1}} \mc{L}_{ \partial_{v_2}}  \iota_{\partial{v_1}} \iota_{\partial{v_2}} \eta \d z  
\end{equation}
Similarly
\begin{equation} 
		\partial_{u_2} \partial_{\ubar_2} \Lap \rho = \int_{z} z \mc{L}_{\partial_{v_1}} \mc{L}_{ \partial_{v_2}}  \iota_{\partial{v_1}} \iota_{\partial{v_2}} \eta \d z   
\end{equation}
so that the operator $\Lap \rho$ is in the kernel of the Laplacian, as desired.

\subsection{The two-point function of the closed-string fields}

Next, we will calculate the two point function of these closed-string fields. This requires a little work because of the non-standard kinetic term of the Kodaira-Spencer action.  

We will find that the two-point function is
\begin{equation} 
	\rho(0) \rho(x) = \frac{1}{2 \pi \i} \log \norm{x}^2. 
\end{equation}
Since $\log \norm{x}^2$ is the Green's kernel for the square of the Laplacian, this is consistent with the fact that $\Lap^2 \rho = 0$.

Both can be derived from the Lagrangian 
\begin{equation} 
	\int (\Lap \rho)^2 
\end{equation}
appropriately normalized. Indeed, the two-point function of this Lagrangian is
\begin{equation} 
	\frac{1}{32 \pi^2} \log \norm{x}^2.  
\end{equation}
Therefore, the correctly normalized Lagrangian is
\begin{equation} 
	- \frac{1}{16 \pi \i}  \int (\Lap \rho)^2. 
\end{equation}
Note that, as mentioned earlier, only the derivatives of $\rho$ are local operators on twistor space. Thus, the fact that the two-point function of $\rho$ does not analytically continue to $\C^4$ does not contradict the fact that the OPEs of local operators  
In terms of the closed one-form $\alpha = \d \rho$, this Lagrangian can be written is
\begin{equation} 
	\int (\partial \alpha^{0,1} )_+ (\dbar \alpha^{1,0})_+.  
\end{equation}

Now let us turn to the calculation of the two-point function, whose result we have stated above. We will give a first-principles derivation from twistor space.

The operator $\rho(0)$ on twistor space is given by
\begin{equation} 
	\rho(0) = \int_{\CP^1} \partial^{-1} \eta 
\end{equation}
where we integrate over the $\CP^1$ lying over $0 \in \R^4$. Since the kinetic term for the field $\eta$ is $\tfrac{1}{2} \int \dbar \eta \partial^{-1} \eta$, we find that the field sourced by the operator $\rho(0)$ is a $(2,1)$-form $\eta$ on twistor space satisfying
\begin{equation} 
	\dbar \eta = \delta_{\CP^1}. 
\end{equation}
The $\delta$-current $\delta_{\CP^1}$ is a $(2,2)$-form with distributional coefficients.  In terms of the coordinates $z,v_1,v_2$ introduced above, the equation is $\dbar \eta = \delta_{v_1 = v_2 = 0}$.

There are many different solutions to this equation, although all are gauge equivalent.  Since we are measuring gauge invariant quantities, we can use any solution we choose.  A solution in an axial gauge is
\begin{equation} 
	\eta_0 = \frac{1}{2 \pi \i} 	\delta_{v_1 = 0}  \frac{\d v_2}{v_2} 
\end{equation}

The $2$-point function is
\begin{equation} 
	\ip{\rho(0) \rho(u)} = \int_{\CP^1_u} \partial^{-1} \eta_0. 
\end{equation}
It is easier to compute the derivatives of this.  Recall that we can write the derivatives of $\rho$ in terms of certain integrals of $\eta$ as follows: 
\begin{equation} 
	\begin{split}
		(\partial_{u_i} \rho)(u) &= \int_{\CP^1_{u}} \iota_{v_i} \eta \\
(\partial_{\ubar_1} \rho)(u) &= \int_{\CP^1_{u}} \iota_{z v_2} \eta \\
		(\partial_{\ubar_2} \rho)(u) &= -\int_{\CP^1_{u}} \iota_{z v_1} \eta \\	
	\end{split}
\end{equation}
As usual, $\CP^1_u$ indicates the twistor line living over the point $u \in \C^2 = \R^4$.

Then, the two-point functions between $\rho(0)$ and $\partial \rho(u)$ are the expressions
\begin{equation} 
	\begin{split}
		2 \pi \i \partial_{u_i} \ip{\rho(0)	\rho)(u)} &= \int_{\CP^1_{u}} \iota_{v_i} \delta_{v_1 = 0} \frac{\d v_2}{v_2} \\
			2 \pi \i \partial_{\ubar_1} \ip{\rho(0)	\rho)(u)} &=  \int_{\CP^1_{u}} \iota_{z v_2} \delta_{v_1 = 0} \frac{\d v_2}{v_2} \\
				2 \pi \i \partial_{\ubar_2} \ip{\rho(0)	\rho)(u)} &=  \int_{\CP^1_{u}} \iota_{z v_1} \delta_{v_1 = 0} \frac{\d v_2}{v_2} 
	\end{split}
\end{equation}
It remains to calculate
\begin{equation}
	\begin{split}
	  \int_{\substack{v_1= u_1 + z \ubar_2 \\v_2 = u_2 - z \ubar_1}  } \iota_{ \partial_{v_i}} \delta_{v_1 = 0} \frac{\d v_2}{v_2}  \\
 \int_{\substack{v_1= u_1 + z \ubar_2 \\v_2 = u_2 - z \ubar_1}  } \iota_{z \partial_{v_i}} \delta_{v_1 = 0} \frac{\d v_2}{v_2}.
	\end{split}
\end{equation}
Each of these expressions forms an $SU(2)$ doublet with index $i$, so we take without loss of generality $i = 2$.  In that case, 
\begin{equation} 
	\delta_{v_1 = 0} = \delta_{z = -u_1 / \ubar_2}  
\end{equation}
and $v_2 = \frac{\norm{u}^2}{\ubar_2}$.  Then, 
\begin{equation}
	\begin{split}
		\int_{\substack{v_1= u_1 + z \ubar_2 \\v_2 = u_2 - z \ubar_1}  } \iota_{ \partial_{v_2}} \delta_{v_1 = 0} \frac{\d v_2}{v_2}  &= \int_{z} \frac{\ubar_2}{\norm{u}^2} \delta_{z = - u_1 / \ubar_2} \\
		&=  \frac{\ubar_2}{\norm{u}^2}\\
		&= \partial_{u_2} \log \norm{u}^2.
	\end{split}
\end{equation}
By symmetry, we have
\begin{equation} 
	\int_{\substack{v_1= u_1 + z \ubar_2 \\v_2 = u_2 - z \ubar_1}  } \iota_{ \partial_{v_i}} \delta_{v_1 = 0} \frac{\d v_2}{v_2} = \partial_{u_i} \log \norm{u}^2. 
\end{equation}

Similarly, 
\begin{equation}
	\begin{split}
		\int_{\substack{v_1= u_1 + z \ubar_2 \\v_2 = u_2 - z \ubar_1}  } \iota_{ z \partial_{v_2}} \delta_{v_1 = 0} \frac{\d v_2}{v_2}  &= \int_{z}z  \frac{\ubar_2}{\norm{u}^2} \delta_{z = - u_1 / \ubar_2} \\
		&=  -\frac{u_1}{\norm{u}^2}\\
		&= -\partial_{\ubar_1} \log \norm{u}^2.
	\end{split}
\end{equation}
By symmetry,
\begin{equation}
	\int_{\substack{v_1= u_1 + z \ubar_2 \\v_2 = u_2 - z \ubar_1}  } \iota_{z \partial_{v_1}} \delta_{v_1 = 0} \frac{\d v_2}{v_2} = \partial_{\ubar_1} \log \norm{u}^2.  
\end{equation}
To sum up, we have found
\begin{equation}
	\begin{split}
		\partial_{u_i} \ip{\rho(0), \rho(u)  } &= \frac{1}{2\pi \i}\partial_{u_i}  \log \norm{u}^2 \\
		\partial_{\ubar_i} \ip{\rho(0), \rho(u)  } &= \frac{1}{2 \pi \i} \partial_{\ubar_i}  \log \norm{u}^2.
	\end{split}
\end{equation}
It follows that 
\begin{equation} 
	\ip{\rho(0)\rho(u)} = \frac{1}{2 \pi \i} \log \norm{u}^2 + C 
\end{equation}
The constant $C$ is irrelevant, as only the derivatives of $\rho$ are true local operators. 

\subsection{Coupling between open and closed string fields}
The twistor  space expression for the coupling between closed string and open string fields is 
\begin{equation} 
	\lambda_{\g} C_{hCS}   \int \eta \op{tr}(\mc{A} \partial \mc{A})
\end{equation}
where $\eta$ is the $(2,1)$ form representing a closed string field, $\mc{A}$ is the $(0,1)$ form gauge field on twistor space, $\op{tr}$ is the trace in the fundamental representation of our Lie algebra, $C_{hCS}$ is the constant
\begin{equation} 
	C_{hCS} = \frac{1}{4 (2 \pi \i)^{3/2} \sqrt{6} } 
\end{equation}
and $\lambda_{\g}$ is a certain constant determined by the Lie algebra. In our main example, for $\g = \mf{so}(8)$, we have $\lambda_{\mf{so}_8} = \sqrt{3/2}$, so that the coupling is
\begin{equation} 
	\frac{1}{ 8 (2 \pi \i)^{3/2}  }   \int \eta \op{tr}(\mc{A} \partial \mc{A})
\end{equation}
This expression is dimensionless, which on twistor space means invariant under scaling of the fibres of the map $\Oo(1)^2 \to \CP^1$.  

In the case of self-dual Yang-Mills theory, the term canceling the anomaly has to be normalized by using $C_{YM} = \sqrt{2} C_{hCS}$ instead of $C_{hCS}$. 

In this section we will compute the coupling of the closed string field on $\R^4$ to the fields of the $\sigma$-model.  We will analyze the couplings which are linear in 
\begin{equation} 
	\d \rho = \alpha = \alpha^i \d u_i + \alpha^{\br{i}} \d \ubar_{\br{i}}. 
\end{equation}
 In any such coupling, we can move all derivatives applied to $\alpha$ to the $\sigma$-field.  Therefore, without loss of generality, we can assume that $\alpha$ is constant, so that $\rho$ is linear.

Let us write down a twistor space representative for the linear field $\rho$ with $\d \rho = \alpha$.  on twistor space. It should be a $(1,1)$ form which we call $\gamma$. If we take 
\begin{equation} 
	\gamma = \frac{1}{2 \pi \i}  \delta_{\abs{z} = 1} \frac{\d z}{z} (\alpha^i v_i + \alpha^{\br{2}} z^{-1} v_1 - \alpha^{\br{1}} z^{-1} v_2 ).  
\end{equation}
 then 
\begin{equation}
	\begin{split} 
		\rho =& \int_{\CP^1} \gamma \\
		=& \oint_{\abs{z} = 1} \frac{1}{2 \pi \i} \frac{\d z}{z} \left( \alpha^1(u_1 + z \ubar_2) + alpha^2  (u_2 - z \ubar_1) \right. \\
		&\left. + \alpha^{\br{2}}(u_1/z + \ubar_2) - \alpha^{\br{1}} ( u_2/z - \ubar_1)  ( \right) \\
		&= \alpha^i u_i + \alpha^{\br{i}} \ubar_{\br{i}}.
	\end{split}	
\end{equation}
This tells us that a $(2,1)$ form representing the constant one-form on $\R^4$ given by $\eta = \partial \gamma$, so 
\begin{equation} 
	\eta = -\frac{1}{2 \pi \i}  \delta_{\abs{z} = 1} \frac{\d z}{z} (\alpha^i \d v_i + \alpha^{\br{2}} z^{-1}\d v_1 - \alpha^{\br{1}} z^{-1} \d v_2 ).  
\end{equation}

The open-string field is away from a small neighbourhood of $z = 0$, and dropping terms involving $\d \zbar$,  
\begin{equation} 
	\mc{A} = \pi^{0,1} \d \ubar_{\br{i}} J_{\ubar_{\br{i}}}  
\end{equation}
To compute the coupling between open and closed string sectors,  we need to analyze the integral
\begin{equation} 
	\lambda_{\g} C_{hCS}    \int_{v_1,v_2} \oint_{\abs{z} = 1} \eta \pi^{0,1} \d \ubar_1 \pi^{0,1} \d \ubar_2 \eps_{\br{k} \br{l}}  J_{\br{k}} \d J_{\br{l}}    
\end{equation}
where $\eta$ is as above. To be precise,  the integral over the circle $\abs{z} = 1$ only picks up forms which have a $\d z$, and not a $\d \zbar$.  This is because we treat $\delta_{\abs{z} = 1}$ as a $(0,1)$ form. 

We have
\begin{equation} 
	\pi^{0,1} \d \ubar_{\br{i}} = \frac{1}{1 + \abs{z}^2} \d \vbar_{\br{i}} 
\end{equation}
modulo terms involving $\d \zbar$, which do not contribute. Since $\abs{z} = 1$  we have
\begin{equation} 
	\pi^{0,1} \d \ubar_{\br{i}} = \tfrac{1}{2} \d \vbar_{\br{i}}. 
\end{equation}
We also have
\begin{equation}
	\begin{split}
		\d \vbar_{1} &= \d \ubar_1 + \zbar \d u_2 \\
		\d \vbar_{2} &= \d \ubar_2 - \zbar \d u_1.
	\end{split}
\end{equation}
Therefore
	\begin{equation} 
		\pi^{0,1} \d \ubar_1 \pi^{0,1} \d \ubar_2 = \tfrac{1}{4} \left(  \d \ubar_1 \d \ubar_2 + \zbar \omega + \zbar^2  \d u_1 \d u_2\right). 
	\end{equation}
where we normalize $\omega$ to be $\d u_1 \d \ubar_1 + \d u_2 \d \ubar_2$.   Since we work at  $\abs{z} = 1$ we can replace $\zbar$ by $1/z$. 

Next, let us wedge this with the $(2,1)$-form $\eta$.  We have
\begin{equation}
	\begin{split}
		&	-\frac{1}{2 \pi \i}\left(   \frac{\d z}{z} (\alpha^i \d v_i + \alpha^{\br{2}} z^{-1}\d v_1 - \alpha^{\br{1}} z^{-1} \d v_2 ) \right) \wedge    \tfrac{1}{4} \left(  \d \ubar_1 \d \ubar_2 + z^{-1} \omega + z^{-2}  \d u_1 \d u_2\right) \\
		&	= -\frac{1}{8 \pi \i} \frac{\d z}{z}  \left( \alpha^1  (\d u_1 + z \d \ubar_2) + \alpha^2 (\d u_2 - z \d \ubar_1) - \alpha^{\br{1}}(\d u_2 z^{-1}  -  \d \ubar_1) + \alpha^{\br{2}}(z^{-1} \d u_1 + \d \ubar_2 )     \right) 
			\\ & \wedge\left(  \d \ubar_1 \d \ubar_2 + z^{-1} \omega + z^{-2}  \d u_1 \d u_2\right)
	\end{split}
\end{equation}
We want to perform a contour integral of this, so we only need to retain terms which are proportional to $z^{-1} \d z$. All the terms involving $\alpha^{\br{i}}$ drop out, because they give $z^{-k} \d z$ for $k \ge 2$. The same holds for all terms involving $z^{-2} \d u_1 \d u_2$.  We are left with
\begin{equation} 
	-\frac{1}{8 \pi \i}\frac{\d z}{z} \left( \alpha^i \d u_i \d \ubar_1 \d \ubar_2 + \alpha^1 \d \ubar_2 \omega - \alpha^2 \d \ubar_1 \omega    \right) \label{eqn_wedge}  
\end{equation}
Then, 
\begin{equation} 
	\begin{split} 
		\d \ubar_2 \omega &= \d u_1 \d \ubar_1 \d \ubar_2 \\  		\d \ubar_1 \omega &= - \d u_2 \d \ubar_1 \d \ubar_2. 
	\end{split}
\end{equation}
Therefore equation \eqref{eqn_wedge} becomes
\begin{equation} 
	-\frac{1}{4 \pi \i} \frac{\d z}{z}\left(  \alpha^i \d u_i \d \ubar_1 \d \ubar_2\right).  
\end{equation}
Inserting this into the integral expression defining the open-closed coupling gives
\begin{equation}
	\begin{split}
		-\lambda_{\g}C_{hCS} \frac{1}{2}    \int_{u_i}	\oint_{\abs{z} = 1}\frac{\d z}{z} \left( (\alpha^i \d u_i\right)\d \ubar_1 \d \ubar_2 \eps^{\br{k} \br{l}} J_{\br{k}} \d J_{\br{l}} \\
=-\lambda_{\g} \frac{1}{  (4 \pi)^2 2 \sqrt{6} }   \int_{u_i}	( (\alpha^i \d u_i)\d \ubar_1 \d \ubar_2 \eps^{\br{k} \br{l}} J_{\br{k}} \d J_{\br{l}} \\
		=  \lambda_{\g} C_{hCS} \half \int_{u,\ubar} \op{tr}( J^{0,1} \partial J^{0,1} )  \alpha^{1,0} \label{open_closed1}  
		\end{split}
\end{equation}
After a field redefinition of $\rho$ this becomes
\begin{equation} 
	 \half  \int (\Lap \rho)^2 +    \frac{1}{32 (2 \pi)^{2} }    \int_{u,\ubar} \op{tr}( J^{0,1} \partial J^{0,1} )  \alpha^{1,0} \label{open_closed_full}. 
\end{equation}
This is form of the action is useful when we turn to self-dual Yang-Mills.  However, it can also also write it in a way closer to the $WZW_4$ Lagrangian. 

We will use the equations satisfied by the current:
\begin{equation}
	\begin{split}
		\partial J^{0,1} &= - \dbar J^{1,0} - [J^{1,0}, J^{0,1}] \\
		\dbar J^{0,1} &= - \tfrac{1}{2} [J^{0,1}, J^{0,1}]. 
	\end{split}
\end{equation}
Equation \eqref{open_closed1} can be rewritten
\begin{equation}	
	-\lambda_{\g} C_{hCS} \half    		 \int \op{tr}( J^{0,1} \dbar  J^{1,0} + J^{0,1}[J^{1,0}, J^{0,1}]  )  \alpha^{1,0} 	
\end{equation}
Integrating by parts, to make the $\dbar$ act on $\alpha$, and picking up the term given $\dbar J^{0,1}$, gives us
\begin{equation} 
	\begin{split} 
		\lambda_{\g} C_{hCS} \half    		 \int \op{tr}( \tfrac{1}{2}[J^{0,1}, J^{0,1}]  J^{1,0} - J^{0,1}[J^{1,0}, J^{0,1}]  )  \alpha^{1,0} -  \op{tr} (J^{0,1} J^{1,0}) \dbar \alpha^{1,0} \\
		 = -\lambda_{\g} C_{hCS} \half \int \op{tr}(   \tfrac{1}{6}[J, J]  J]  )  \alpha^{1,0} -  \op{tr}(J^{0,1} J^{1,0}) \dbar \alpha^{1,0} 
	\end{split}
\end{equation}
Here we recognize the $WZW_4$ Lagrangian we started with, suitable normalized, but where the K\"ahler form has been replaced by a multiple of 
\begin{equation} 
	- \dbar \alpha^{1,0} = -\d \alpha^{1,0} = \partial \dbar \rho. 
\end{equation}

Including the closed-string kinetic term, rewriting in terms of $\rho$ where $\alpha^{1,0} = \partial \rho$, and specializing to $\mf{g} = \mf{so}(8)$,  the expression becomes
\begin{equation}
	\begin{split} 
		-\frac{1}{16 \pi \i} \int (\Lap \rho)^2	+  \lambda_{\so(8)} C_{hCS} \half \int_{u,\ubar} \op{tr}( J^{0,1} \partial J^{0,1} )  \alpha^{1,0}  \\
		= -\frac{1}{16 \pi \i}   \int (\Lap \rho)^2 +    \frac{1}{16 (2 \pi \i)^{3/2} }    \int_{u,\ubar} \op{tr}( J^{0,1} \partial J^{0,1} )  \alpha^{1,0} 
	\end{split}	
\end{equation}
By a field redefinition of $\rho$ we can write this as 
\begin{equation} 
	    \half \int (\Lap \rho)^2 +    \frac{1}{16 \pi  }    \int_{u,\ubar} \op{tr}( J^{0,1} \partial J^{0,1} )  \alpha^{1,0} 
\end{equation}

\section{The Green-Schwarz mechanism on twistor space for self-dual Yang-Mills} \label{sec:sdym_twistor}
It is well-known (see e.g.\ \cite{Boels:2006ir}) that holomorphic BF theory on twistor space gives rise to self-dual Yang-Mills theory on space-time.  This theory suffers from the same box-diagram anomaly as holomorphic Chern-Simons, and as we mentioned earlier, a similar (but simpler) version of the Green-Schwarz mechanism cancels the anomaly.  

In this section we will give more details on this point and compute how the closed-string field couples to the fields of self-dual Yang-Mills. 

The fields of holomorphic BF theory plus the closed-string field are
\begin{equation} 
	\begin{split}
		\mc{A} \in \Omega^{0,1}(\PT,\g)[1] & \\
		\mc{B} \in \Omega^{3,1}(\PT,\g)[1] & \\
		\eta \in \Omega^{2,1}(\PT)[1]. 
	\end{split}
\end{equation}
Here, we use calligraphic symbols to for the fields on twistor space to distinguish them from the corresponding fields on $\R^4$. 

The field $\eta$, as in the case of holomorphic Chern-Simons, is constrained to satisfy $\partial \eta=  0$. 
The Lagrangian is
\begin{equation}
	\begin{split}
		\int \op{tr}(\mc{B} F(\mc{A})) &+  c	\int \eta \op{tr} \left( \mc{A} \partial\mc{A} \right )\\
		&+ \int \dbar \eta \partial^{-1} \eta \label{eqn_bf_lagrangian}
	\end{split}
\end{equation}
where $c$ is a constant described below.  The exchange of closed-string fields as in figure \ref{fig:GS_sdym} cancels the one-loop anomaly of holomorphic BF theory, as long as the Lie algebra $\mf{g}$ is such that 
\begin{equation} 
	\op{Tr}_{\mf{g}}(X^4) \propto \op{tr}(X^2)^2. 
\end{equation}
This holds when $\mf{g}$ is one of $\mf{sl}_2$, $\mf{sl}_3$, $\mf{so}(8)$ or any exceptional algebra.
\begin{figure}

\includegraphics[scale=0.22]{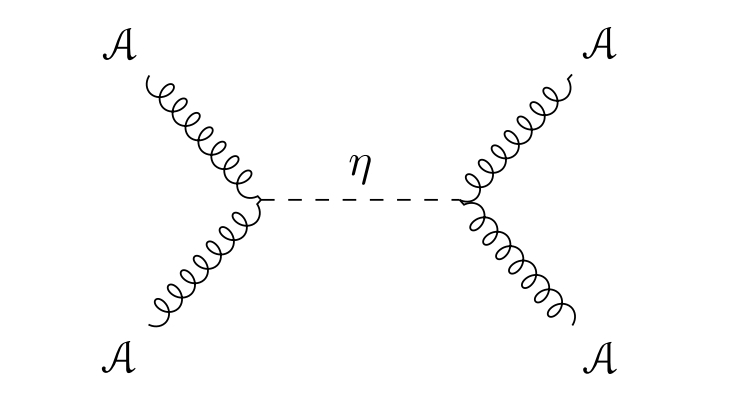}

	\caption{Green-Schwarz anomaly cancellation by exchange of the field $\eta$. \label{fig:GS_sdym}}
\end{figure}
Then, in the appendix \ref{sec:normalization_twistor}, we calculate the relevant constant $c$ to be
\begin{equation}
	\begin{split} 
		c &= \lambda_{\g} C_{YM} \\	
		C_{YM} &= \frac{1}{4 (2 \pi \i)^{3/2} \sqrt{3} }
		\lambda_{g}^2 \op{tr}(X^2)^2 = \op{Tr}(X^4)
	\end{split}	
\end{equation}
where $\op{tr}$ means trace in the fundamental, $\op{Tr}$ in the adjoint.   

Now let us explain what Lagrangian this gives on $\R^4$. We have already seen that the four-dimensional field corresponding to $\eta$ is a field $\rho$ with fourth-order Lagrangian $-\frac{1}{16 \pi \i} \int \rho \Lap^2 \rho$, where $\rho$ is taken up to the addition of a constant. 

Unlike in the case of holomorphic Chern-Simons theory, the closed-string field has no self-interaction.  To completely describe the four-dimensional model, it only remains to fix the interaction between $\rho$ and the fields $A,B$ of four-dimensional self-dual Yang-Mills theory.

This interaction is dimensionless, gauge invariant, and $SO(4)$ invariant.  It must also be linear in $\rho$. This can be seen by considering the tree-level Feynman diagrams one can build from the Lagrangian \eqref{eqn_bf_lagrangian}.  The only Lagrangian of this nature (up to terms which vanish on-shell) is
\begin{equation} 
	\int \d \rho CS(A) = - \int \rho \op{Tr}(F(A)^2). 
\end{equation}
Thus, the four-dimensional theory has Lagrangian 
\begin{equation} 
	\int \op{Tr} BF(A)_+ + 4 \pi i \rho \Lap^2 \rho + c \d \rho CS(A) 
\end{equation}a
To fix the constant $c$, we use the relationship between $WZW_4$ and self-dual Yang-Mills discussed in \cite{2011.04638}.  There is a natural map between the fields associated to the inclusion of sheaves $\Oo(-2) \to \Oo$ on twistor space. Under this map, $A$ be comes $J^{0,1}$; we can view this as a partial gauge fixing for self-dual Yang-Mills.  

The open-closed coupling to cancel the anomaly in holomorphic Chern-Simons is 
\begin{equation} 
	\lambda_{\g}	C_{hCS}  \half \int_{u,\ubar} \op{tr}( J^{0,1} \partial J^{0,1} )  \partial \rho. 
\end{equation}
Because the anomaly for self-dual Yang-Mills is different from that in holomorphic Chern-Simons by a factor of two, we need to use 
\begin{equation} 
	\lambda_{\g}	C_{YM}  \half \int_{u,\ubar} \op{tr}( J^{0,1} \partial J^{0,1} )  \partial \rho. 
\end{equation}
If $A = J^{0,1}$, this is the same as 
\begin{equation} 
	\lambda_{\g}	C_{YM}  \int_{u,\ubar}  \op{tr}(\half  A \d A + \tfrac{1}{6}A[A,A] )  \d  \rho = \lambda_{\g} C_{YM}\half \int CS(A) \d \rho.  
\end{equation}
Including the kinetic term for $\rho$, and writing out $C_{YM}$ explicitly, this is
\begin{equation} 
	-\frac{1}{16\pi \i} \int (\Lap \rho)^2	\lambda_{\g}  \frac{1}{8 (2 \pi \i)^{3/2} \sqrt{3} }  \int CS(A) \d \rho.  
\end{equation}
Note that we can absorb some factors into a redefinition of $\rho$, giving us
\begin{equation} 
	\frac{1}{2 } \int (\Lap \rho)^2 + 	\lambda_{\g}  \frac{1}{8 \pi  \sqrt{3} }  \int CS(A) \d \rho.  
\end{equation}

In particular, for $\mf{sl}_2$, $\lambda_{\g} = \sqrt{8}$ so that the coupling (including the kinetic term for $\rho$) is
\begin{equation} 
	\frac{1}{2 } \int (\Lap \rho)^2	+  \frac{1}{\sqrt{24} }  \int CS(A) \d \rho.  
\end{equation}

\section{Interlude: A general method to determine terms in the Lagrangian, and counter-terms, in twistorial theories}
\label{sec:interlude}
So far, we have computed the kinetic term of the K\"ahler potential and how it couples to the field of the $WZW_4$ model, to linear order in $\rho$.  Computing the self-interaction of $\rho$ is rather challenging.  (In the self-dual Yang-Mills case, there is no self-interaction term for $\rho$). 

In this section I will introduce a general technique for doing such calculations.  This technique also applies in principle at the quantum level, to constrain the counter-terms in any twistorial theory.  

As we have seen, in a twistorial theory, the correlation functions are entire analytic functions.  In particular, the operator product of two local operators $\mc{O}_1(0)$, $\mc{O}_2(x)$ can never involve any terms with $\log \norm{x}$.  This constrains the possible terms in the Lagrangian very tightly, and often allows us to fix them by a recursive argument.  

\subsection{Fixing twistorial Lagrangians in scalar theories}
We will start by illustrating this method with a collection of free scalars $\phi^r$.  Let us try to build an interaction term for these fields with no logarithms in the OPE. We will start by trying to find a cubic interaction. If it has no derivatives, then it is of the form
\begin{equation} 
	I^{rst}\int \phi_r \phi_s \phi_t 
\end{equation}
where $I$ is some symmetric tensor.  We will see that such a term is not allowed. 

To compute the OPE of the operators $\phi_r$, $\phi_s$ in the presence of this interaction, we recall that the Green's function is
\begin{equation} 
	\ip{\phi_r(0) \phi_s(x)} = \delta_{rs}	\frac{1}{ \norm{x -y}^2 } 
\end{equation}
(here we are absorbing factors of $\pi$ into various normalizations).

The first contribution to the OPE comes from the Feynman diagram in figure \ref{figure_cubic_ope}.  
\begin{figure}	

\includegraphics[scale=0.22]{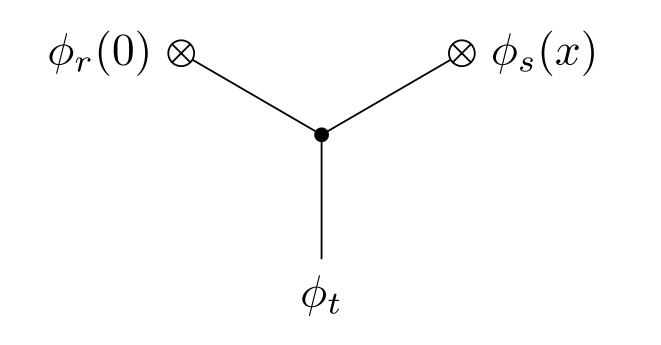}
		\caption{The first diagram contributing to the OPE. For Feynman diagrams computing OPEs, vertices labeled by $\otimes$ correspond to local operators.  External lines are labeled by on-shell field configurations (here $\phi_t$), and do not have propagators.  \label{figure_cubic_ope} }
\end{figure}

This gives  
\begin{equation}
			\phi_r(0) \phi_s(x) \sim  I_{rs}^t\int_{x'} \frac{1}{\norm{x'} \norm{x - x'} }\phi_t(x') \d^4 x' 	
\end{equation}
We are interested in the singular terms when $x$ is small, and which don't involve any derivatives applied to $\phi_t$. We thus approximate $\phi_t$ by its value at $x = 0$; the error terms (by Taylor's theorem) involve derivatives of $\phi$. Then, the integral
\begin{equation}
			  I_{rs}^t \phi_t(0)  \int_{x'} \frac{1}{\norm{x'} \norm{x - x'} }\d^4 x' 	
\end{equation}
 is the convolution of the Green's function with itself.   The result is therefore the Green's function for the square of the Laplacian, which is proportional to
 \begin{equation} 
	  \log \norm{x}. 
 \end{equation}
We conclude that (restoring factors of $\pi$)  we have the OPE 
\begin{equation} 
	\phi_r(0) \phi_t(x) \sim  \frac{1}{16 \pi^2} \log \norm{x} I_{rs}^t \phi_t(0) + \dots 
\end{equation}
where $\dots$ indicates diagrams involving multiple vertices.

We conclude that a cubic action must have at least one derivative. The same argument shows that any polynomial action with no derivatives will have a logarithm in the OPE.

What about terms with one derivative? Suppose we have an interaction of the form
\begin{equation} 
	I^{rst}_i \phi_r \phi_s \partial_{x_i} \phi_t 
\end{equation}
where $I^{rs}_i$ is symmetric under permutation of $r$ and $s$. (Further, a tensor $I^{rst}_i$ which is totally symmetric under permutation of $r,s,t$ gives a Lagrangian which is a total derivative and hence does not contribute).

Then, the diagram in figure \ref{figure_twistor_ope} gives rise to the OPE
\begin{equation}
	\begin{split}
		\phi_r(0) \phi_s(x) \sim&  I^{rst}_i  \int_{x'} \frac{1}{\norm{x'}^2 \norm{x-x'}^2 } \partial_{x'_i} \phi_t(x') \d^4 x' \\
		&+ I^{rts}_i  \int_{x'} \frac{1}{\norm{x'}^2}\left(\partial_{x'_i}  \frac{1}{ \norm{x-x'}^2 }\right)  \phi_t(x') \d^4 x' \\
		&+ I^{str}_i  \int_{x'} \left( \partial_{x'_i} \frac{1}{\norm{x'}^2}\right)  \frac{1}{ \norm{x-x'}^2 }  \phi_t(x') \d^4 x' 
	\end{split}
\end{equation}
The first term on the right hand side can be computed using the method of before, giving us a factor of $\log \norm{x}$.  The second term can, by integration by parts, be written as a sum of the integrals in the first and third terms.  The third term is 
\begin{equation} 
	-2 I^{str}_i   \int_{x'} \frac{1}{\norm{x'}^4}\frac{\phi_t(x') x_i'}{ \norm{x-x'}^2 }  \phi_t(x') \d^4 x'. 
\end{equation}
We can Taylor expanding $\phi_t(x') = \phi_t(0) + x'_j \partial_j \phi_t(0) + \dots$. Dimensional analysis tells us that the only logarithmic contributions come from the first derivatives of $\phi_t$, giving the integral
\begin{equation} 
	-2  I^{str}_i \int_{x'} \frac{1}{\norm{x'}^4}\frac{\partial_j \phi_t(0) x'_j x'_i}{ \norm{x-x'}^2 }  \phi_t(x') \d^4 x'.  
\end{equation}
Since $\log \norm{x}$ is $SO(4)$ invariant, we can evaluate the coefficient of $\log$ in the integral by projecting onto the $SO(4)$ invariants which is
\begin{equation} 
	-2  I^{str}_i \int_{x'} \frac{1}{\norm{x'}^4}\frac{\partial_i \phi_t(0) \norm{x'}^2  }{ \norm{x-x'}^2 } \phi_t(x') \d^4 x' \sim -\frac{2}{8 \pi} \partial_i \phi(0) I^{str}_i \phi_t(0) \log \norm{x} . 
\end{equation}
Collecting all the terms shows us that the only logarithmic terms in the  OPE are
\begin{equation} 
	\phi_r(0) \phi_s(x) \sim \frac{1}{8 \pi} \log \norm{x}   \partial_i \phi_t(0) \left(  I^{rst}_i - 2 I^{str}_i +  I^{rts}_i \right)   
\end{equation}
This vanishes only when $I^{rst}_i$ is symmetric in $r,s,t$, which only happens if the Lagrangian is a total derivative.

We conclude that, up to trivial Lagrangians, it is not possible to have a cubic Lagrangian with a single derivative with no logarithmic OPEs.

\subsubsection{Lagrangians with two derivatives}
If we have \emph{two} derivatives it is possible to build a cubic Lagrangian where the diagram in Figure \ref{figure_cubic_ope} does not contribute a logarithmic OPE.  

The $WZW_4$ Lagrangian, when written in terms of $\phi = \log \sigma$, is of this form: we can expand the Lagrangian as
\begin{equation} 
	\int \phi_a \Lap \phi_a + f^{abc} \omega_{ij} \int \phi_a \partial_{x_i} \phi_b \partial_{x_j} \phi_c + O(\phi^4) 
\end{equation}
where $\omega$ is the K\"ahler form on $\R^4$. 

We will see that requiring there be no logarithmic OPEs, plus some symmetry considerations, leads us to this Lagrangian.

The most general cubic Lagrangian with two derivatives is of the form
\begin{equation} 
	I^{rst}_{ij} \phi_r \partial_i \phi_s \partial_j \phi_t \label{eqn:cubic_2del} 
\end{equation}
(A Lagrangian with two derivatives applied to a single field can be brought to this form by integration by parts).  We can assume that $I^{rst}_{ij}$ is symmetric under simultaneous permutation of $i,j$ and $s,t$.

Let us pause to see why such an expression is trivial if it is Lorentz invariant. If this Lagrangian is Lorentz invariant, it must be
\begin{equation} 
	 	I^{rst} \phi_r \partial_i \phi_s \partial_i \phi_t  
\end{equation}
We can absorb any Lagrangian involving $\Lap \phi_r$ into a field redefinition, so such terms can be dropped. This means the Lagrangian is
\begin{equation} 
	I^{rst} \int \phi_r \Lap(\phi_s \phi_t) = I^{rst} \int (\Lap \phi_r) \phi_s \phi_t 
\end{equation}
so it can be absorbed by a field redefinition.

Let us assume that our interaction \eqref{eqn:cubic_2del} is invariant under the group $U(2)$. Since it is not invariant under $SO(4)$,  this means that the space-time indices can only be contracted with the K\"ahler form, so it is of the form
\begin{equation} 
	\omega_{ij} 	I^{rst} \phi_r \partial_i \phi_s \partial_j \phi_t 
\end{equation}
where $I^{rst}$ is anti-symmetric in $s$ and $t$.  

One can compute explicitly the logarithmic terms in the OPE, using the methods above, and see that they vanish. It is  easy to see this by symmetry reasons.  The OPE must be of  the form
\begin{equation} 
	\phi_r(0) \phi_s(x) = I_{rs; ij}^t \partial_i \partial_j \phi_t(0). 
\end{equation}
Invariance under $U(2)$ implies that the operator $I_{rs;ij}^t \partial_i \partial_j$ is a multiple of the Laplacian, and so the OPE vanishes once we impose the linearized equations of motion.  This is sufficient because we are currently considering only terms which involve the interaction once.

It is not difficult to show that any $U(2)$ invariant Lagrangian with more than $2$ derivatives is an expression involving $\Lap \phi$, and so can be absorbed into a field redefinition.

If we further ask that our fields live in a Lie algebra $\g$ and that the Lagrangian is invariant under the adjoint action, we find that the only possible cubic term is that appearing in the $WZW_4$ action.

\subsection{Quartic Lagrangians}

One can, of course, go further and try to constrain all of the terms in the tree-level Lagrangian by this method.    At the quartic level, terms with $\le 2$ derivatives automatically lead to log divergences.  This means that terms with $\le 1$ derivatives are excluded.  

Terms with $2$ derivatives, however, are forced on us to cancel the logarithmic term in the OPE associated to diagram figure \ref{figure_ope_cancellation}. 

  \begin{figure}



\includegraphics[scale=0.22]{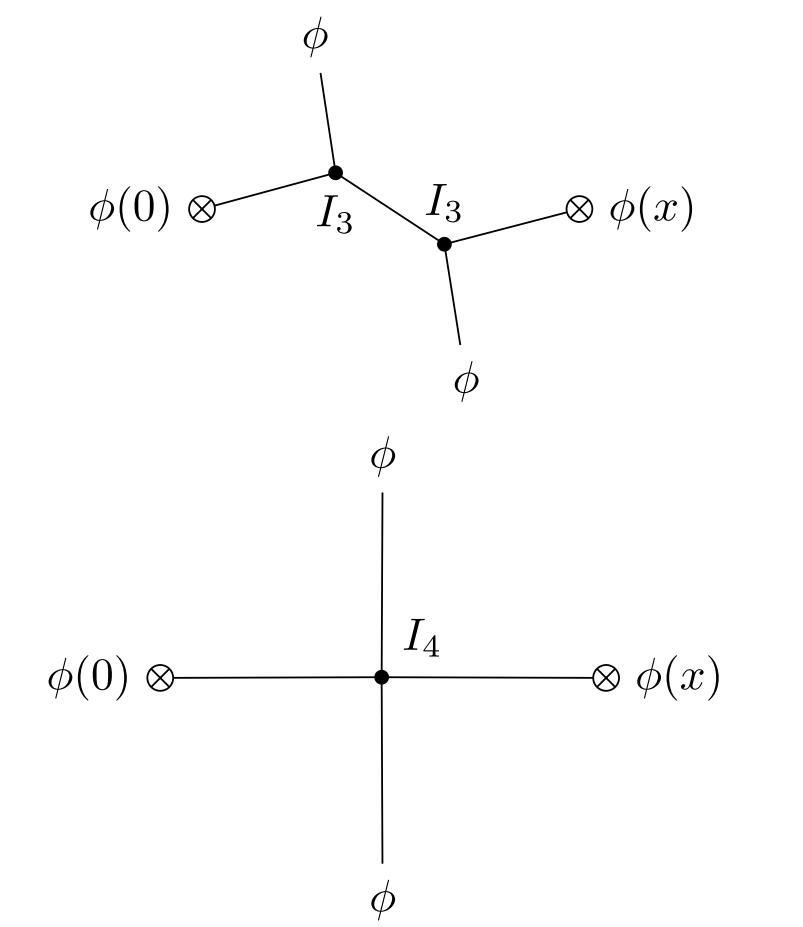}

	\caption{ The logarithmic OPEs associated to these two diagrams must cancel, which controls the coefficient of the quartic interaction $I_4$.  \label{figure_ope_cancellation} }
\end{figure}

To see this, we note that the logarithmic divergence in a $U(2)$ invariant  term with $2$ derivatives must be of the form
\begin{equation} 
	\phi_r(0) \phi_s(x) \sim J^{rstu} \omega_{ij}  \partial_i \phi_t (0) \partial_j \phi_u(0) \log \norm{x}  
\end{equation}
for some tensor $J$.   This form for the logarithmic OPE follows from $U(2)$ invariance and dimensional analysis. An explicit computation, similar to the ones given above, shows that $J^{rstu}$ encodes all the information of the quartic interaction up to total derivatives. 

The fact that the first diagram in  Figure \ref{figure_ope_cancellation} gives an OPE of the same form again follows from $U(2)$ invariance and dimensional analysis; the precise coefficients can be determined by evaluating some integrals explicitly.

From this, we see that the quartic terms with $\le 2$ derivatives are completely determined by the cubic terms.  Calculating the precise coefficient of the quartic term using this method is rather painful.  However, from what we have seen so far, we can prove the following:
\begin{proposition}
	The $WZW_4$ Lagrangian (expanded as a series in $\phi = \log \sigma$) is the unique Lagrangian for a scalar field in the Lie algebra $\g$, which is:
\begin{enumerate} 
	\item Invariant under $U(2)$ and the adjoint action of $G$
	\item Has as quadratic term the standard kinetic term for a scalar field.
	\item Has only two derivatives. 
	\item At one loop order, there are no logarithmic OPEs. 
\end{enumerate}
\end{proposition}
\begin{proof} 
	As we sketched above, term in the Lagrangian involve $k$ copies of the field $\phi$ leads to a logarithmic OPE which must cancel log OPEs coming from terms in the Lagrangian with $<k$ copies of $\phi$.    This uniquely fixes the term with $k$ copies of $\phi$. 
\end{proof}
One can hope that, at loop level, something similar happens. That is, when we include the closed-string fields and restrict to $G = SO(8)$, all counter-terms are uniquely fixed by the requirement that there are no log OPEs.

\section{The closed-string interaction}
\label{section_closed_interaction}
In this section, we will show how to derive the closed-string cubic interaction.  We will do this by first using a twistor space analysis to narrow down the possible cubic terms to two, and then use the method of section \ref{sec:interlude} to fix the relative coefficients of these terms.

The first step is to note that the Kodaira-Spencer interaction makes sense for a $(2,1)$ form on twistor space with logarithmic poles on the divisor $z=  0$, $ z = \infty$. (Recall that we have defined our fields just to be regular $(2,1)$-forms on $\PT$).  This can be seen explicitly in coordinates, say, near $z = 0$. A $(2,1)$ form with logarithmic poles corresponds to a Beltrami differential where the $\partial_{v_i}$ components are divisible by $z$ and the $\partial_z$ component is divisible by $z^2$.  The Kodaira-Spencer interaction acquires a zero of order $4$ from $\partial_{v_1} \partial_{v_2} \partial_z$, and a pole of order $4$ because it involves the holomorphic volume form twice. Therefore, there is no overall pole or zero.

A closed $(2,1)$ form with logarithmic poles on twistor space corresponds to closed $(1,1)$-form on $\R^4$, which is $\partial \dbar \rho$ if $\rho$ is the K\"ahler potential field corresponding to a $(2,1)$-form on twistor space with no poles.

To see this, we note that if $\eta$ is a closed $(2,1)$ form with log poles on twistor space, then
\begin{equation} 
	\int_{z} \iota_{\partial_{v_i}} \eta 
\end{equation}
is not a gauge invariant operation under the gauge symmetry $\eta \mapsto \eta + \dbar \chi$, for a $(2,0)$ form with log poles $\chi$.  This is because the zero in $\iota_{\partial_{v_i}}$ at $z = \infty$ cancels the pole in $\eta$ at $\infty$, but the pole in $\eta$ at $z = 0$ is not canceled.

The gauge invariant measurements we can make are of the form
\begin{equation} 
	B^{i \br{j}} = \omega^{\br{j}}_{j}	\int_z \iota_{ \partial_{v_i}} \mc{L}_{z \partial_{v_j}} \eta. 
\end{equation}
It is clear that if $\eta$ has no poles, then $B^{i \br{j}} = \partial_{u_i} \partial_{\ubar_{\br{j}}} \rho$.

Since the closed-string interaction on twistor space extends to $(2,1)$ forms with logarithmic poles, it must be written entirely in terms of $B = \partial \dbar \rho$.   Note that the self-dual part of $B$ is simply $\Lap \rho$.

There are two $U(2)$ invariant cubic interactions of the correct dimension. These are
\begin{equation} 
	\int (\partial \dbar \rho)^2 \Lap \rho \  \ \int \omega^2 (\Lap \rho)^3  
\end{equation}
Therefore, the Lagrangian is a linear combination
\begin{equation} 
	C	\int (\partial \dbar \rho)^2 \Lap \rho + D \int \omega^2 (\Lap \rho)^3  \label{closed_interaction} 
\end{equation}
We will determine the ratio of the coefficients $C,D$. We will fix the ratio of the coefficients by requiring that we do not get a logarithmic contribution to the OPE. 

We first note that the Kodaira-Spencer propagator on twistor space can be written as $(\partial \boxtimes 1) \dbar^{-1}$, where $\dbar^{-1}$ is the Green's kernel for the $\dbar$ operator on $\Omega^{2,\ast}$. Viewed as as a kernel, $\dbar^{-1}$ has one leg in $\Omega^{1,\ast}$ and one leg in $\Omega^{2,\ast}$.  We then apply $\partial: \Omega^{1,\ast} \to \Omega^{2,\ast}$ to get the propagator. 

The field $\rho$ corresponds, on twistor space, to a $(1,1)$-form $\gamma$ satisfying $\partial \gamma = \eta$, where $\eta$ is the $(2,1)$-form Kodaira-Spencer field.  This discussion shows that the operator $\rho = \int_{\CP^1} \gamma$ makes sense on twistor space as part of any tree-level Feynman diagram where there are no propagators which connect any two such operators  to each other. 

\begin{figure}

\includegraphics[scale=0.22]{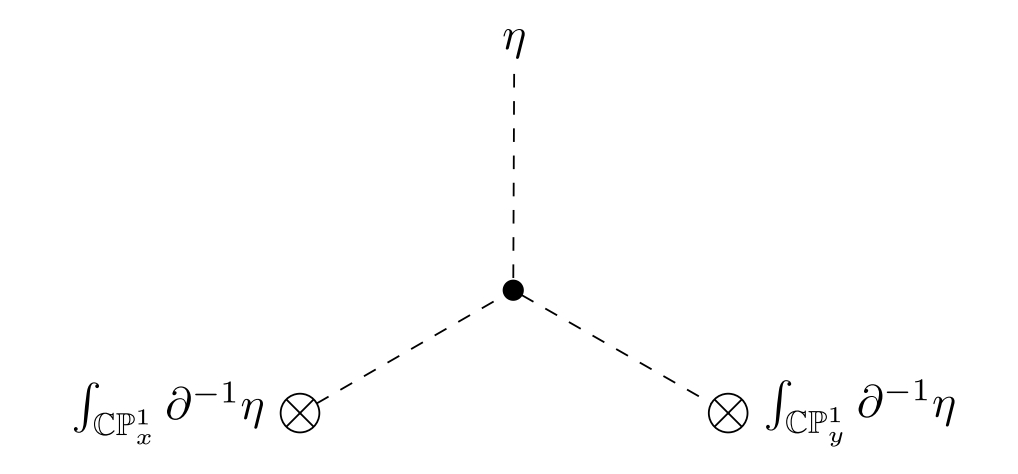}
		\caption{The OPE on twistor space.  \label{figure_twistor_ope}}
\end{figure}

We are interested in the diagram in Figure \ref{figure_twistor_ope}.  This diagram makes sense where the two operators are placed at any $\CP^1$ in $\PT$. If we place one operator at $0$, then the diagram makes sense when the other operator is placed in any $\CP^1$, so the result must analytically extend to $\C^4$. 
Each term in the Lagrangian \eqref{closed_interaction} contributes a logarithmic term in the OPE which cancels if $C,D$ are chosen correctly. We will calculate the contribution of each term to the OPE.  

Let us normalize the two-point function of $\rho$ to be $ \log \norm{x}$, which is proportional Green's function of $\Lap^2$.   Let us first determine the logarithmic OPE associated to $D \int (\Lap \rho)^3$.  This is  
\begin{equation} 
	\rho(0) \cdot \rho(u,\ubar) \sim  D \int 6 (\Lap \log \norm{u'} ) (\Lap \rho(u',\ubar')) (\Lap \log \norm{u' - u} ). 
\end{equation}
The factor of $6$ comes from the symmetries of the vertex. 

Now $ \Lap \log \norm{u' - u}$ is $2 \norm{u-u'}^{-2}$,  so the integral becomes
\begin{equation} 
	\rho(0) \cdot \rho(u,\ubar) \sim  D \int 6 \cdot 4  \frac{1}{ \norm{u'}^2 }(\Lap \rho(u',\ubar')) \frac{1}{ \norm{u' - u}^2 } 
\end{equation}
This is (up to a factor) the convolution of the Green's function with itself, as we discussed in section \ref{sec:interlude}, so that we end up with
\begin{equation} 
	\rho(0) \cdot \rho(u,\ubar) \sim  D \int  6 \cdot 4  \pi^2   \log \norm{u} \Lap \rho(u,\ubar).
\end{equation}

For the other possible vertex $C \int (\partial \dbar \rho)^2 \Lap \rho$, there are three possibilities depending on where the line labeled by $\Lap \rho$ goes.   Let us first consider the case when the external line is $\Lap \rho$. Then, the OPE looks like
\begin{equation} 
	\int_{u', \ubar'}2   \Lap \rho (u',\ubar') \left( \partial \dbar \log (\norm{u'}^2 )\right)\wedge \left( \partial \dbar   \log (\norm{u' - u})^2\right) . 
\end{equation}
The factor of $2$ is a symmetry factor. By integration by parts, we have
\begin{equation}
\begin{split}
	&2 	 \int_{u', \ubar'}  \Lap \rho (u',\ubar') \left( \partial \dbar \log (\norm{u'}^2 )\right)\wedge \left( \partial \dbar   \log (\norm{u' - u})^2\right) \\ &= - 2   \int_{u', \ubar'} \left( \partial  \Lap \rho (u',\ubar') \right) \left(  \dbar \log (\norm{u'}^2 )\right)\wedge \left( \partial \dbar   \log (\norm{u' - u})^2\right)  
\end{split}
\end{equation}
where we have used the fact that $\partial^2 = 0$.   We can do the same with the $\dbar$ operator giving us
\begin{equation}
\begin{split}
	&\int_{u', \ubar'}  \Lap \rho (u',\ubar') \left( \partial \dbar \log (\norm{u'}^2 )\right)\wedge \left( \partial \dbar   \log (\norm{u' - u})^2\right)  \\ &=   \int_{u', \ubar'} \left( \partial \dbar \Lap \rho (u',\ubar') \right) \left(   \log (\norm{u'}^2 )\right)\wedge \left( \partial \dbar   \log (\norm{u' - u})^2\right)  
\end{split}
\end{equation}
This expression has at least four derivatives acting on the external field $\rho$, so for dimensional reasons must result in an OPE of the form
\begin{equation} 
	\rho(0) \cdot \rho(u,\ubar) \sim (\partial \dbar \Lap \rho) F(u,\ubar) + \dots 
\end{equation}
where $F$ is an expression of weight $2$ under rescaling $u$.  Thus, $F$ is non-singular, and this does not contribute to the OPE at all.

The remaining possibility is that the external line is labeled by $\partial \dbar \rho$, so the amplitude contributing to the OPE is
\begin{equation}
\begin{split}
	2 	\int_{u', \ubar'} (\partial \dbar \rho (u',\ubar')) \wedge \left( \Lap \log (\norm{u'}^2 )\right)\wedge \left( \partial \dbar   \log (\norm{u' - u})^2\right)\\
	+  2  	\int_{u', \ubar'} (\partial \dbar \rho (u',\ubar')) \wedge \left( \partial \dbar \log (\norm{u'}^2 )\right)\wedge \left( \Lap   \log (\norm{u' - u})^2\right) .  
\end{split}
\end{equation}
We can write $\partial \dbar \rho$ as the self-dual and anti-self dual part, 
\begin{equation} 
	\partial \dbar \rho = (\partial \dbar \rho)_- + \omega \Lap \rho. 
\end{equation}
The anti-self-dual part can not contribute to a logarithmic OPE, because the operator multiplying $\log \norm{u}^2$ must be $SU(2)$ invariant.  Therefore, only the self-dual part can contribute.  The self-dual amplitude is the same as the one we encountered when we studied the vertex $(\Lap \rho)^3$, giving us 
\begin{equation} 
	\rho(0) \cdot \rho(u,\ubar) \sim  C \int   4 \cdot 4  \pi^2   \log \norm{u} \Lap \rho(u,\ubar). 
\end{equation}
There is a different prefactor from the vertex with coefficient $D$, which had $6 \cdot 4 \pi^2 D $.    The ratio of $4/6$ comes from the fact that with the vertex with coefficient $C$, $4$ of the $6$ ways to position the external lines lead to a logarithmic divergence.  

Therefore, the overall logarithmic OPE is of the form
\begin{equation} 
	\rho(0) \cdot \rho(u,\ubar) \sim 24 \pi^2 (\tfrac{2}{3} C + D)  ( \Lap \rho) \log (\norm{u}^2) + \dots  
\end{equation}
where $\dots$ are terms with no branch cuts.  This vanishes if $D = - \tfrac{2}{3} C$. 

The overall coefficient of 
\begin{equation} 
	-\frac{2}{3} \int (\Lap \rho)^3 \omega^2 + \int (\Lap \rho) (\partial \dbar \rho)^2	 
\end{equation}
can not be determined quite so easily.  It is possible in principle to compute it using these methods,  but using loop level diagrams in a version of the Green-Schwarz mechanism, as discussed in sections \ref{section_spacetime_GS} and \ref{sec:wzw4_gs}.

\subsection{Geometric understanding of the equations of motion}
The full Lagrangian, up to an overall factor and after rescaling $\rho$, is
\begin{equation} 
	\int \omega^2 \rho \Lap^2 \rho +  (\Lap \rho)  (\partial \dbar \rho)^2  - \frac{2}{3} (\Lap \rho)^3 \omega^2  
\end{equation}
Varying $\rho$, the equations of motion are
\begin{equation}
2 \omega^2 \Lap^2 \rho +  \Lap ( (\partial \dbar \rho)^2 )  +  2 (\partial \dbar \Lap \rho) (\partial \dbar \rho) - 2 \omega^2 \Lap (\Lap \rho \Lap \rho). 	
\end{equation}
In terms of the closed $(1,1)$-form $B = \partial \dbar \rho$, $B_+ = \Lap \rho \omega$, the equations of motion become
\begin{equation}
	\begin{split}
		0 &= 2 \omega \Lap B_+ +  \Lap ( B\wedge B) + 2 B\wedge  \Lap B - 2 \Lap (B_+\wedge B_+) \\
		&= 2 \omega \Lap B_+ + \Lap (B_- \wedge B_-) -  \Lap ( B_+ \wedge B_+ ) + 2 B_+ \wedge \Lap B_+ + 2 B_- \wedge \Lap B_-.  
	\end{split}
\end{equation}
This equation is equivalent to asking that the scalar curvature of the K\"ahler metric $\omega + B$ vanishes, modulo $B^3$.

We will normalize so that $\ast (\omega^2) = 1$. Recall that the Ricci curvature of the deformed K\"ahler metric $\omega + B$ is given by
\begin{equation}
	\begin{split}
		R &= \partial \dbar f \\
		f &= \log \left( \frac{(\omega + B)^2}{\omega^2}  \right) \\
		&= \log \left( 1 + 2\ast(B_+ \wedge \omega) + \ast (B_+ \wedge B_+) + \ast(B_- \wedge B_-)    \right) \\
		&= 2\ast(B_+ \wedge \omega) - 2 \ast(B_+ \wedge B_+)  + \ast(B_+ \wedge B_+) + \ast (B_- \wedge B_-) + O(B^3) \\
		&= 2\ast(B_+ \wedge \omega) -  \ast(B_+ \wedge B_+)  + \ast (B_- \wedge B_-) + O(B^3) 
	\end{split}
\end{equation}

The scalar curvature is 
\begin{equation}
	\begin{split}
		S &=(\omega + B) \wedge R \\
		&= \omega^2  \Lap f  + B_+ \omega  \Lap f + B_- \wedge \partial \dbar f \\
		&= 2 \omega \Lap B_+ + 2 B_+ \Lap B_+ -  \Lap(B_+ \wedge B_+) +   \Lap(B_- \wedge B_-) + 2 B_- \wedge (\partial \dbar (B_+/\omega))  +  O(B^3).   
	\end{split}
\end{equation}
Now, $B_+ = \omega \Lap \rho$ where $\rho$ is the K\"ahler potential, so that 
\begin{equation} 
B_- \wedge \partial \dbar (B_+ / \omega) = B_- \wedge \Lap B_-. 
\end{equation}
Therefore
\begin{equation} 
	S = 2 \omega \wedge \Lap B_+ + 2 B_+ \wedge \Lap B_+ -  \Lap(B_+ \wedge B_+) +   \Lap(B_- \wedge B_-) + 2 B_- \wedge \Lap B_-  +  O(B^3).    
\end{equation}
This matches precisely the equations of motion of our Lagrangian.

We conclude that, modulo $B^3$, the equations of motion for the closed-string field theory are that the scalar curvature of the K\"ahler metric $\omega + B$ vanishes. 

There is an important caveat here.  I do not currently know the ratio between the coefficients of the terms
\begin{equation} 
	\begin{split}
		\int (\partial \dbar \rho) \op{Tr} J^{1,0} J^{0,1} \\
		(\Lap \rho)  (\partial \dbar \rho)^2  - \frac{2}{3} (\Lap \rho)^3 \omega^2 . 
	\end{split}
\end{equation}
These coefficient can in principle be determined by the Green-Schwarz mechanism on space-time as discussed in sections \ref{section_spacetime_GS} and \ref{sec:wzw4_gs}.  

The coefficient of the first term tells us which multiple of $\partial \dbar \rho$ should be viewed as the deformation of the K\"ahler form. This term, after all, couples $\partial \dbar \rho$ to the stress-energy tensor of the $\sigma$-model.  

The interpretation of the equations of motion in terms of the scalar curvature is really only valid if the relative coefficients between the two terms is $1$.  Calculating this relative coefficient involves some challenging one-loop Feynman diagram computations.

\section{The Green-Schwarz mechanism on space-time for self-dual Yang-Mills}
\label{section_spacetime_GS}

In this section we will perform a two-loop cancellation of an operator product expansion in self-dual Yang-Mills theory. We will find that there is a non-trivial logarithmic divergence, which can be canceled by once we introduce the closed-string field $\rho$. This is a space-time incarnation of the Green-Schwarz mechanism on twistor space.

To write down the OPE we will compute, let us introduce some notation. Decompose $\op{Spin}(4) = SU(2)_+ \times SU(2)_-$, and let $S_+$, $S_-$ be the fundamental representations of $SU(2)_+$, $SU(2)_-$.  Introduce indices $\alpha$ for basis of $S_+$ and $\dot{\alpha}$ of $S_-$.  We use the convention that the twistor line $\CP^1$ is a rotated by $SU(2)_+$.  

We will compute the OPE of the stress-energy tensor 
\begin{equation} 
	T_{\alpha \dot{\alpha} \beta \dot{\beta}} = B^a_{\alpha \beta} F^a_{\dot{\alpha} \dot{\beta}}.	 
\end{equation}
with itself.  

We will find that there is a logarithmic term in the OPE, at two loops, of the form
\begin{equation} 
	T_{\alpha \dot{\alpha} \beta \dot{\beta}}(0) 	T_{\gamma \dot{\gamma} \delta \dot{\delta}}(w)  \sim
	  \eps_{\alpha \gamma} \eps_{\beta \delta}F^a_{\dot{\alpha} \dot{\beta}}F^b_{\dot{\gamma} \dot{\delta}}   (F^a \wedge F^b)  		 \log \norm{w}^2 + \text{ rational }\label{eqn_log_ope_T}.
\end{equation}
It is important to note that in equation \eqref{eqn_log_ope_T} we have not written down all the terms in the OPE, only the lowest-order logarithmic term.  There are also rational function terms which appear at one loop, and logarithmic terms at higher loops.

To compute this OPE directly on $\R^4$ is challenging. From the space-time perspective, this logarithmic OPE is associated to the diagrams in figure \ref{Figure:spacetime_BB_ope}.
\begin{figure}
\includegraphics[scale=0.25]{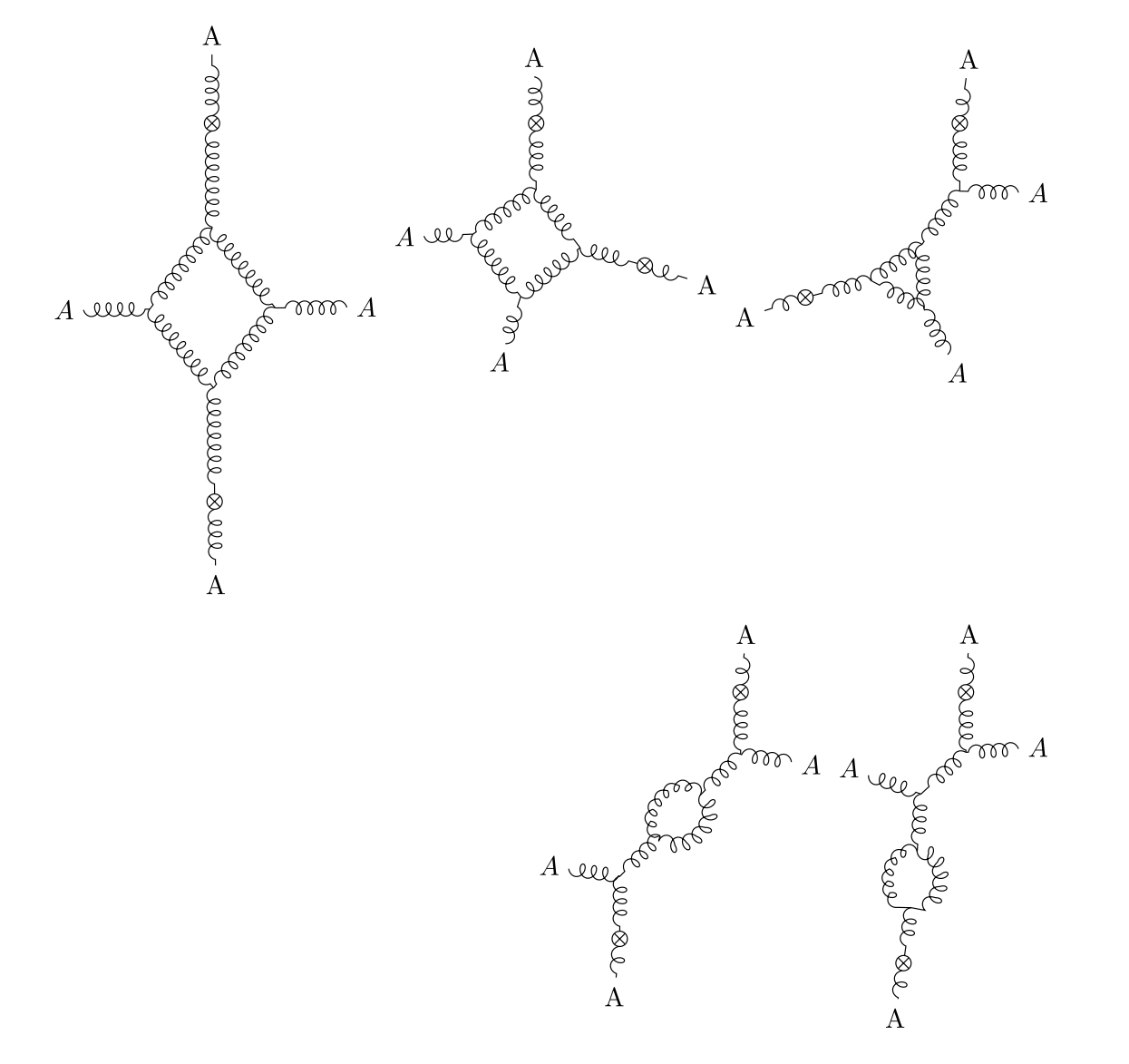}

	\caption{Diagrams giving rise to the $T(0) T(x)$ OPE. Here, the vertices labeled by $\otimes$ are where the operator $T$ is placed.   \label{Figure:spacetime_BB_ope}}
	
\end{figure}

Instead we will give a twistor space derivation.  We will write 
\begin{equation} 
	\begin{split}
		\mc{A} & \in \Omega^{0,1}(\PT, \mf{g}) \\
		\mc{B} & \in \Omega^{3,1}(\PT, \mf{g}) 
	\end{split}
\end{equation}
for the fields on twistor space, and use italic font $A,B$ for the fields on $\R^4$.   

Recall that there is a non-local quartic expression on twistor space which is introduced to cancel the one-loop anomaly:
\begin{equation} 
	\frac{1}{ (2 \pi)^4 32 }  \int_{x \in \R^4} \int_{z,z' \in \CP^1}\op{Tr}( \mc{A}(x,z) \d \mc{A}(x,z)\frac{\d z \d z'}{(z - z')^2}  \mc{A}(x,z') \d \mc{A}(x,z') ) + O(\mc{A}^5). \label{eqn_nonlocal_vertex} 
\end{equation}
Here the trace is in the adjoint representation.  From the twistor space perspective, local vertices by themselves can not give rise to logarithmic OPEs. This is only possible with non-local vertices.  

We will therefore calculate the logarithmic term in the OPE by representing the stress-energy tensor $T$ on twistor space, and  using \eqref{eqn_nonlocal_vertex}, as in the diagram Figure \ref{figure_twistor_ope}.
\begin{figure}

\includegraphics[scale=0.22]{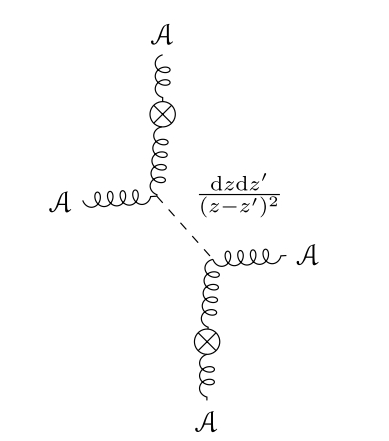}

	\caption{This figure represents the twistor space computation of the logarithmic OPE, including all the diagrams from Figure \ref{Figure:spacetime_BB_ope}.  We represent the twistor space representation of the operator $B(0)^2$ as the vertices $\otimes$. We represent the quartic non-local vertex as a product of cubic vertices connected by $\frac{\d z \d z'}{(z-z')^2}$. \label{figure_twistor_ope_T}}	
\end{figure}

The first step is to represent the operator $T$ on twistor space. In fact, we only need to describe the non-gauge invariant operator $B_{\alpha \beta}$.  The field $\mc{B}$ is a $(0,1)$-form on twistor space twisted by $\Oo(-4)$. On each $\CP^1$, $\Oo(-4) = \K_{\CP^1}^{\otimes 2}$.  To construct a $(1,1)$ form on $\CP^1$, we need to contract $\mc{B}$ with a vector field on $\CP^1$. Holomorphic vector fields on $\CP^1$ are sections of $\Oo(2)$, and live in the adjoint representation of $SU(2)_+$.     We can write a vector field as $V_{\alpha \beta}$. Then, 
\begin{equation} 
	B^a_{\alpha \beta}(0) = \int_{\CP^1} \iota_{V_{\alpha \beta}} \mc{B}^a.\end{equation}

We will compute the OPE associated to the diagram in Figure \ref{figure_twistor_ope_T}. Let us pick three points $u_i^{(k)}$, $k = 1,2,3$ on $\R^4$. We will place the non-gauge invariant operators at $u^{(1)}$ and $u^{(3)}$, and place the non-local vertex on twistor space anomaly at $u^{(2)}$.  At the end of the computation we will integrate over $u^{(2)}$. 

Since the propagator attaches to only the $B$ in the stress-energy tensor, and not to $F$,  as an intermediate step we compute the OPE of the non-gauge invariant operator $B_{\alpha \beta}^a$ with $B_{\gamma \delta}^b$.  Further, because everything is invariant under $SU(2)_+$, we lose no generality by taking the OPEs between $B_{12}^a$ and $B_{12}^b$.  The propagator is the field sourced by the operator $\int_{\CP^1} \iota_{V_{12}} \mc{B}^a$ on twistor space. 

We will do this in the coordinates $z,v_1,v_2$ we have been using on twistor space up to now. The vector field $V_{12}$ is
\begin{equation} 
	V_{12} = z \partial_z. 
\end{equation}
If we trivialize the bundle $\Oo(-4)$ on a patch of twistor space by using the meromorphic volume form $\d v_1 \d v_2 \d z / z^2$, so that $\mc{B}$ is a $(0,1)$-form, then operator $B_{12}$ is represented by 
\begin{equation} 
	\int_{\CP^1} \frac{\d z}{z} \mc{B}. \label{eqn_twistor_operator}
\end{equation}
In these coordinates, the kinetic term of the Lagrangian is
\begin{equation} 
	\int \mc{B} \dbar \mc{A} \d v_1 \d v_2 \frac{\d z}{z}. 
\end{equation}
From this we see that the field sourced by the operator \eqref{eqn_twistor_operator} placed at $u,\ubar$ is  a field $\mc{A}$ which satisfies
\begin{equation} 
	\dbar \mc{A} = z \delta_{v_1 = u_1 + z \ubar_2} \delta_{v_2 = u_2 - z \ubar_1}.  
\end{equation}
We will solve this equation, in an axial gauge, for the operators placed at $u^{(1)}$ and $u^{(3)}$:
\begin{equation}
\begin{split}
	\mc{A}^{(3)} &= \frac{1}{2 \pi \i} z \delta_{v_1 = u_1^{(3)} + z \ubar_2^{(3)} } \frac{1}{ v_2 - u_2^{(3)} + z \ubar_1{(3)}  }\\
	\mc{A}^{(1)} &=  \frac{1}{2 \pi \i}  z\delta_{v_2 = u_2^{(1)} - z \ubar_1^{(1)} } \frac{1}{ v_1 - u_1^{(1)} - z \ubar_1{(1)}  }\\
\end{split}
\end{equation}

We will let $\mc{A}$ be the background field placed at the two external lines of Figure \ref{figure_twistor_ope_T}. We take $\mc{A}$ to be on-shell, so that it can be expressed in terms of the four-dimensional gauge field $A$. We can then insert the fields $\mc{A}^{(1)}, \mc{A}^{(3)}$, $\mc{A}$ into the non-local vertex on twistor space, giving the expression
\begin{equation} 
	\frac{1}{(2 \pi)^4 8 }  \int_{u^{(2)},\ubar^{(2)}}	\int_{z,z'}\mc{A}^{(1)}(u^{(2)},\ubar^{(2)},z) \d \mc{A}(u^{(2)},\ubar^{(2)},z) \frac{\d z \d z'}{(z-z')^2} \d \mc{A}(u^{(2)},\ubar^{(2)},z) \mc{A}^{(3)}(u^{(2)},\ubar^{(2)},z) \label{eqn_nonlocalamplitude} 
\end{equation}
This expression will be a function of $u^{(13)}$, $\ubar^{(13)}$ and of the background field $\mc{A}$. Since this field is on-shell, we can express it in terms of a four-dimensional gauge field $A$, and we find the following:
\begin{proposition}
	The singular part in the amplitude of \eqref{eqn_nonlocalamplitude} as $u^{(13)}, \ubar^{(13)} \to 0$ is a non-zero multiple of
	\begin{equation} 
		\log \norm{u^{13}}	\int_{\R^4} F^a(A) \wedge F^b(A) 
	\end{equation}
	where $a,b$ are the colour indices of the non-gauge invariant operators $B^a$, $B^b$ whose OPE we are taking.
	\label{prop:ope_sdym}
\end{proposition}
The proof is a rather lengthy direct computation which we place in the appendix \ref{sec:twistor_ope}. We do not determine the coefficient of the logarithmic divergence using this method, but will determine it shortly in another way.

This result immediately implies that we
\begin{equation} 
	T_{\alpha \dot{\alpha} \beta \dot{\beta}}(0) 	T_{\gamma \dot{\gamma} \delta \dot{\delta}}(w)  \sim
	C  \eps_{\alpha \gamma} \eps_{\beta \delta}F^a_{\dot{\alpha} \dot{\beta}}F^b_{\dot{\gamma} \dot{\delta}}   (F^a \wedge F^b)  		 \log \norm{w}^2 + \text{ rational }
\end{equation}
for a constant $C$ which we will determine shortly. It would be very interesting to perform this calculation directly from the space-time perspective, by computing from first principles the amplitude of the diagrams \ref{Figure:spacetime_BB_ope}.   This seems quite a bit more challenging than the computation from the twistor point of view.

\subsection{Canceling the logarithmic divergence} 
The general analysis from twistor space tells us that this logarithmic OPE can be canceled by an exchange of closed string fields. Here we will verify on $\R^4$ that this is indeed the case.  This will allow us to fix the overall factor. 

The diagram we will use is Figure \ref{Figure:BB_closedstring}. 
\begin{figure}

\includegraphics[scale=0.22]{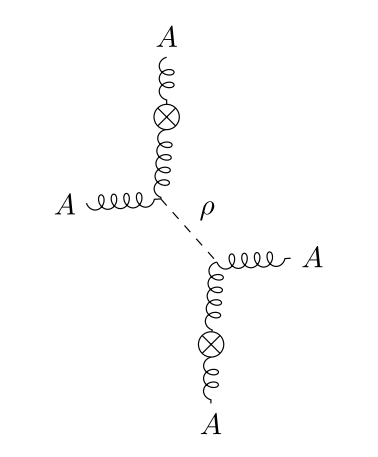}

	\caption{The exchange of  the field $\rho$ leads to a logarithmic term in the $T (0) T(x)$  OPE which cancels the term coming from the diagrams in \ref{Figure:spacetime_BB_ope}. The stress energy tensor is placed at the vertex labeled $\otimes$.
	\label{Figure:BB_closedstring}}
\end{figure}

To perform this computations, we will use coordinates $w_{\alpha \dot{\beta}}$ where
\begin{equation} 
	\begin{split} 
		w_{1\dot{1}} &= u_1\\
		w_{1 \dot{2}} &= u_2 \\
		w_{2 \dot{1}} &= \ubar_2\\
		w_{2 \dot{2}} &= -\ubar_1
	\end{split}
\end{equation}
The inner product in this basis is
\begin{equation} 
-\half \eps^{\alpha \beta}  \eps^{\dot{\alpha} \dot{\beta}} w_{\alpha \dot{\alpha}} w_{\beta \dot{\beta}}  
\end{equation}

Then, self-dual two-forms have two undotted indices,
\begin{equation} 
	B= B^{\alpha \beta} \eps^{\dot{\alpha} \dot{\beta} }\d w_{\alpha \dot{\alpha}} \d w_{\beta \dot{\beta}}  
\end{equation}
It will be helpful to write
\begin{equation} 
	B^{\alpha \beta} = \tfrac{1}{4} \eps_{\dot{\alpha} \dot{\beta}} \iota_{\partial_{\alpha \dot{\alpha}} } \iota_{\partial_{\beta \dot{\beta}}} B   
\end{equation}

We have scalar but non-gauge invariant operators $B^a B^b$ defined by
\begin{equation} 
	 B^a B^b \d^4 x = B\wedge B.
\end{equation}
In our coordinates, we have
\begin{equation} 
	B^a B^b = \tfrac{1}{16}   B^a_{\alpha \beta} B^b_{\gamma \delta} \eps^{\alpha \gamma} \eps^{\beta \delta} 
\end{equation}
Similarly, we have
\begin{equation} 
	F^a F^b =  \frac{1}{16}F^a_{\dot{\alpha}_1 \dot{\beta}_1}	F^b_{\dot{\alpha}_2 \dot{\beta}_2}  	\eps^{\dot{\alpha}_1 \dot{\alpha}_2}	\eps^{\dot{\beta}_1 \dot{\beta}_2}. 
\end{equation}

In the appendix, we compute the normalization of the two-point function of $F$ and $B$. We need the following results. Let $A_{\alpha \beta}$ be the field sourced by $B_{\alpha \beta}$.  We have
\begin{equation}
	\begin{split} 
		\ip{ F_{\dot{\alpha} \dot{\beta}} B_{\alpha \beta} } &= 		 \frac{1}{32 \pi^2 \norm{w}^6 } \left(  w_{\beta \dot{\beta}} w_{\alpha \dot{\alpha}} + w_{\alpha \dot{\beta}} w_{\beta \dot{\alpha}} \right) \\
		F \wedge \d A_{\alpha \beta} &= 2 \d^4 x  \eps_{\alpha \lambda} \eps_{\beta \mu} 	F_{\dot{\lambda} \dot{\mu}} \dpa{w_{\lambda \dot{\lambda}}}  \dpa{w_{\mu \dot{\mu}} }  G(w) 
	\end{split}
\end{equation}

The first computation we will do is the following. 
\begin{lemma}
	The logarithmic term in the OPE of the non-gauge invariant operator $B_{\alpha\beta}^a$ with $B_{\gamma \delta}^b$ coming from the exchange of the field $\rho$ as in Figure \ref{Figure:BB_closedstring} is 
	\begin{equation} 
	B^a_{\alpha_1 \beta_1} (w^{(1)}) B^b_{\alpha_2 \beta_2}(w^{(2)} ) \sim \frac{\lambda_{\g}^2}{(2 \pi)^4 3 \cdot 2^5   } \left(  \eps_{\alpha_1 \beta_2} \eps_{\alpha_2 \beta_1} + \eps_{\alpha_1 \alpha_2} \eps_{\beta_1 \beta_2} \right) F^a \wedge F^b  		 \log \norm{w^{(1)} - w^{(2)}}^2 + \text{ rational }	
	\end{equation}
\end{lemma}

Since we know that self-dual Yang-Mills with the axion field has no logarithmic OPEs, we conclude that this must be canceled by an OPE in just self-dual Yang-Mills with the opposite sign:
	\begin{equation}
B^a_{\alpha_1 \beta_1} (w^{(1)}) B^b_{\alpha_2 \beta_2}(w^{(2)} ) \sim -\frac{\lambda_{\g}^2}{(2 \pi)^4 3 \cdot 2^5   } \left(  \eps_{\alpha_1 \beta_2} \eps_{\alpha_2 \beta_1} + \eps_{\alpha_1 \alpha_2} \eps_{\beta_1 \beta_2} \right) F^a \wedge F^b  		 \log \norm{w^{(1)} - w^{(2)}}^2 + \text{ rational }
	\end{equation}

\begin{proof}
Let us compute the amplitude of the Feynman diagram in Figure \ref{Figure:BB_closedstring}.   We label the internal vertices in the diagram by $w^{(3)}$, $w^{(4)}$.  The propagator is $\frac{1}{2 \pi \i} \log \norm{w^{(34)}}^2$.  The vertices are $\lambda_{\g} C_{YM} \half \int \d \rho \op{tr} (A \d A)$ where
\begin{equation} 
	C_{YM} = \frac{1}{4 (2 \pi \i)^{3/2} \sqrt{3} }
\end{equation}
Therefore the amplitude us  
\begin{equation} 
	\lambda_{\g}^2 C_{YM}^2 \frac{1}{2 \pi \i}	\eps^{\alpha_1 \alpha_3} \eps^{\alpha_2 \alpha_4} \int_{w^{(3)}, w^{(4)} }\op{tr}( \d A_{\alpha_1 \alpha_2} (w^{(3)})  \d A^{bg} (w^{(3)} )) \op{tr}( \d A_{\alpha_3 \alpha_4}  (w^{(4)} ) \d A^{bg}(w^{(4)})) \log \norm{w^{(34)}}^2 
\end{equation}
In this expression, $A^{bg}$ denotes the background field, which we assume to be on-shell.  

This brings us to
\begin{equation} 
	\frac{\lambda_{\g}^2}{(2 \pi)^4 48 }	\int_{w^{(3)}, w^{(4)} } \int_{w^{(3)}, w^{(4)} }\left( \d A_{\alpha_1 \alpha_2} (w^{(3)})  \d A^{bg} (w^{(3)} ) \d A_{\alpha_3 \alpha_4}  (w^{(4)} ) \d A^{bg}(w^{(4)})\right) \log \norm{w^{(34)}}^2 
\end{equation}

We abuse notation slightly and write
\begin{equation} 
	F = \d A^{bg}. 
\end{equation}
The other terms in the field strength don't contribute. As, we are computing the OPE to quadratic order in the background field -- higher order terms are determined by gauge invariance.  The field strength $F$ has only anti-self dual components 
\begin{equation} 
	F = F^{\dot{\alpha} \dot{\beta}} \d w_{\alpha \dot{\alpha}} \d w_{\beta \dot{\beta}} \eps^{\alpha \beta}. 
\end{equation}
Using 
\begin{equation} 
	F \wedge \d A_{\alpha \beta} = 	 2 \d^4 x \eps_{\alpha \lambda} \eps_{\beta \mu} 	F_{\dot{\lambda} \dot{\mu}} \dpa{w_{\lambda \dot{\lambda}}} \dpa{w_{\mu \dot{\mu}} }  G(w)  
\end{equation}
we find that our amplitude is 
\begin{equation}
	\begin{split}
		\frac{4\lambda_{\g}^2 }{(2 \pi)^4 48}	\int_{w^{(3)}, w^{(4)}} F_{\dot{\alpha}_1 \dot{\beta}_1}(w^{(3)} ) \d^4 x^{(3)} \d^4 x^{(4)}  
	\dpa{w^{(3)}_{\alpha_1 \dot{\alpha}_1 }} \dpa{w^{(3)}_{\beta_1 \dot{\beta}_1 } } G(w^{(1)} - w^{(3)}) \\	
	F_{\dot{\alpha}_2 \dot{\beta}_2}(w^{(4)}) 	\dpa{w^{(4)}_{\alpha_2 \dot{\alpha}_2 }} \dpa{w^{(4)}_{\beta_2 \dot{\beta}_2 } }  G(w^{(2)} - w^{(4)})     \log \norm{w^{(34)}}^2
	\end{split}
\end{equation}
Logarithmic OPEs can only involve two derivatives in total applied to the background fields.  This means that we can re-express the amplitude by integration by parts, dropping any terms where derivatives apply to the curvature  $F$ of the background field, and we can take $F$ outside the integral.  

Also using the fact that $(\partial_{w^{(3)}_{\alpha \dot{\alpha}}} + \partial_{w^{(4)}_{\alpha \dot{\alpha}_i}} ) \log \norm{w^{(34)}}^2 = 0$, we can rewrite the amplitude as 
\begin{equation}
	\begin{split}
		F_{\dot{\alpha}_1 \dot{\beta}_1}	F_{\dot{\alpha}_2 \dot{\beta}_2} 		\frac{4\lambda_{\g}^2}{(2 \pi)^4 48 }	\int_{w^{(3)}, w^{(4)}} \d^4 x^{(3)} \d^4 x^{(4)}   \dpa{w^{(3)}_{\alpha_1 \dot{\alpha}_1}} \dpa{w^{(3)}_{\beta_2 \dot{\beta}_2 } }  G(w^{(1)} - w^{(3)}) \\ 
	\dpa{w^{(4)}_{\alpha_2 \dot{\alpha}_2 }} \dpa{w^{(4)}_{\beta_1 \dot{\beta}_1 } }  G(w^{(2)} - w^{(4)})      \log \norm{w^{(34)}}^2
	\end{split}
\end{equation}
If we write $I_{\dot{\alpha}_1 \dot{\beta}_1 \dot{\alpha}_2 \dot{\beta}_2 \alpha_1 \beta_1 \alpha_2 \beta_2}(w^{(1)} - w^{(2)})$ for the integral above, then we have shown so far that
\begin{equation} 
	B_{\alpha_1 \beta_1}(w^{(1)}) B_{\alpha_2 \beta_2}(w^{(3)}) =  I_{\dot{\alpha}_1 \dot{\beta}_1 \dot{\alpha}_2 \dot{\beta}_2 \alpha_1 \beta_1 \alpha_2 \beta_2}(w^{(1)} - w^{(2)})  F^{\dot{\alpha}_1 \dot{\beta}_1} F^{\dot{\alpha}_2 \dot{\beta}_2} .  
\end{equation}

Now, since any logarithmic terms must be Lorentz invariant, logarithmic terms in the integral above can occur only if the dotted indices are contracted by an $\eps$ tensor.   Note that there is only one way to contract the indices. Therefore, we have 
\begin{equation}
	\begin{split} 
		I_{\dot{\alpha}_1 \dot{\beta}_1 \dot{\alpha}_2 \dot{\beta}_2 \alpha_1 \beta_1 \alpha_2 \beta_2} (w^{(1)} - w^{(2)})& 
		\\
		=\frac{1}{16}   \eps_{\dot{\alpha}_1 \dot{\beta}_2}  \eps_{\dot{\alpha}_2 \dot{\beta}_1}&  \eps^{\dot{\gamma}_1 \dot{\delta}_2}    \eps^{\dot{\gamma}_2 \dot{\delta}_1} \eps_{\alpha \beta_2} \eps_{\alpha_2 \beta_1} \eps^{\gamma_1 \delta_2} \eps^{\gamma_2 \delta_1}       I_{\dot{\gamma}_1 \dot{\delta}_1 \dot{\gamma}_2 \dot{\delta}_2 \alpha_1 \beta_1 \alpha_2 \beta_2} (w^{(1)} - w^{(2)})
	\\	&+ \text{ rational } 	 
	\end{split}	
\end{equation}

Therefore  
\begin{multline} 
	 		B_{\alpha_1 \beta_1} (w^{(1)}) B_{\alpha_2 \beta_2}(w^{(2)} )
			\\ \sim \eps_{\alpha_1 \beta_2} \eps_{\alpha_2 \beta_1} F^{\dot{\alpha}_1 \dot{\beta}_1} F^{\dot{\alpha}_2 \dot{\beta}_2}  		 \frac{\lambda_{\g}^2}{(2 \pi)^4 3 \cdot 2^6   }
			 \int_{w^{(3)}, w^{(4)}}  \eps_{\gamma_1 \delta_2} \eps_{\dot{\gamma}_1 \dot{\delta}_2} \d^4 x^{(3)} \d^4 x^{(4)}  
		\dpa{w^{(3)}_{\gamma_1 \dot{\gamma}_1}} \dpa{w^{(3)}_{\delta_2 \dot{\delta}_2 } }  G(w^{(1)} - w^{(3)}) \\ 
		\eps_{\gamma_2\delta_1} \eps_{\dot{\gamma}_2 \dot{\delta}_1}	\dpa{w^{(4)}_{\gamma_2 \dot{\gamma}_2 }} \dpa{w^{(4)}_{\delta_1 \dot{\delta}_1 } }  G(w^{(2)} - w^{(4)})      \log \norm{w^{(34)}}^2  + \text{ rational }
\end{multline}
(We are implicitly symmetrizing with respect to the $\alpha_1 \beta_1$ and $\alpha_2 \beta_2$ indices on the right, to keep the expression short).

Here we see that the differential operators multiplying $G(w^{(1)} -w^{(3)})$ and $G(w^{(2)} - w^{(4)})$ are $\half \Lap$.  Since we have normalized $G$ so that $\d^4 x\Lap G(w) = \delta_{w = 0}$, we find (re-introducing colour indices and using our normalization of the scalar operator $F^a \wedge F^b$)  
\begin{equation} 
	B^a_{\alpha_1 \beta_1} (w^{(1)}) B^b_{\alpha_2 \beta_2}(w^{(2)} ) \sim \frac{\lambda_{\g}^2}{(2 \pi)^4 3 \cdot 2^5   } \left(  \eps_{\alpha_1 \beta_2} \eps_{\alpha_2 \beta_1} + \eps_{\alpha_1 \alpha_2} \eps_{\beta_1 \beta_2} \right) F^a \wedge F^b  		 \log \norm{w^{(1)} - w^{(2)}}^2 + \text{ rational }
\end{equation}
\end{proof}

We have computed the logarithmic OPE of a non-gauge invariant operator in self-dual Yang-Mills, using the fact that when we introduce the field $\rho$ it must cancel logarithmic OPEs.  This will allow us to normalize the logarithmic OPE of the stress-energy tensor. Recall that  the (trace-free) stress-energy tensor of self-dual Yang-Mills is, up to normalization,
\begin{equation} 
	T_{\alpha \dot{\alpha} \beta \dot{\beta}} = B^a_{\alpha \beta} F^a_{\dot{\alpha} \dot{\beta}}.	 
\end{equation}
In self-dual Yang-Mills, the operator $F$ has no non-trivial OPEs with itself, and the OPE with $B$ is rational.  Therefore, the only way to get a logarithmic term in the $TT$ OPE is to contract the two copies of the operator $B$. This gives us
\begin{equation}
	\begin{split} 
		 	T_{\alpha \dot{\alpha} \beta \dot{\beta}}(0) 	T_{\gamma \dot{\gamma} \delta \dot{\delta}}(w)  \sim
	\frac{\lambda_{\g}^2}{(2 \pi)^4 3 \cdot 2^4   }   \eps_{\alpha \gamma} \eps_{\beta \delta}F^a_{\dot{\alpha} \dot{\beta}}F^b_{\dot{\gamma} \dot{\delta}}   (F^a \wedge F^b)  		 \log \norm{w}^2 \\ 
		+ \text{ rational }
\end{split}
\end{equation}

\subsection{Scaling of the operator $B^2$ in self-dual Yang-Mills}
Self-dual Yang-Mills is a non-unitary CFT.  Like in any CFT, scaling on $\R^4$ is a symmetry of the theory, and so acts on the Hilbert space (or the space of local operators).  In a unitary CFT, this symmetry is a diagonalizable operator on the Hilbert space, whose spectrum gives the anomalous dimensions of operators.

In a non-unitary CFT like self-dual Yang-Mills, there is no reason for this matrix to be diagonal, and indeed it is not.  The quantum scaling operator turns out to be a sum of the classical scaling operator, plus a strictly upper-triangular matrix.  To be more precise, a local operator is a trace of a product of derivatives of $B$ and $F$.  The classical scaling operator is diagonal with eigenvalue the number of derivatives plus twice the number of $B$'s and $F$'s (as these have dimension $2$).  The quantum correction is triangular in this basis, and always decreases the number of $B$'s. 

We will not compute all the terms in the scaling operator.  Instead, we will compute the quantum correction to the scaling of the operator $B^2$. 

In general, quantum corrections to scaling of $B^2$ arise when there are logarithmic terms in the UV cut-off for the renormalization of this composite operator.

In self-dual Yang-Mills plus our axion field, the twistor space origin guarantees that there are no quantum corrections to the classical answer.  Therefore any quantum contribution to the scaling of the operator $B^2$ can be canceled by an exchange of an axion.  This means that we can compute the quantum correction to scaling by studying logarithmic counter-terms that occur in the renormalization of the operator $B^2$, in diagrams exchanging an axion.

The calculation of the logarithmic OPE given above can be repurposed easily to do this. The relevant diagram is in Figure \ref{Figure:Scaling_B}.  
\begin{figure}

\includegraphics[scale=0.22]{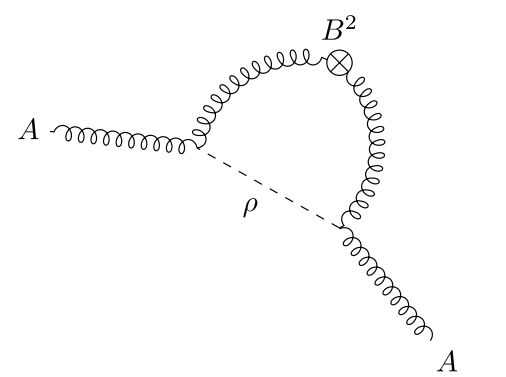}

	\caption{This diagram has a logarithmic UV divergence, contributing to the RG flow of the operator $B^2$.  
	\label{Figure:Scaling_B}}
\end{figure}
We can compute this by performing a point-splitting regularization, replacing $B^2(0)$ by 
\begin{equation} 
	B^2_{reg} = 	\frac{1}{16}  B_{\alpha {\beta}} (0) B_{{\gamma} {\mu}} (\epsilon) \eps^{{\alpha} {\gamma}} \eps^{{\beta}{\mu}}   
\end{equation}j
We already know the logarithmic term in the  OPE of $B$ with itself, giving
\begin{equation} 
	B^2_{reg} =   -\frac{\lambda_{\g}^2}{(2 \pi)^4 3 \cdot 2^4   }    F \wedge F  		 \log \eps^2  + \text{ rational }
\end{equation}
That is, renormalization of the operator $B^2$ involves a two-loop logarithmic counter-term. The matrix of anomalous dimensions involves applying the operator $\eps \partial_{\eps}$. The rational terms do not contribute, as they are canceled by the classical dimension of all operators.  The quantum contribution tells us that the flow of $B^2$ is 
\begin{equation} 
	B^2 \mapsto B^2 - \frac{\lambda_{\g}^2}{(2 \pi)^4 3 \cdot 2^3} F^2  
\end{equation}

\subsection{Computing the two-point function of the operator $B^2$}
The fact that the operator $B^2$ flows by adding on a multiple of $F^2$ gives us surprisingly strong constraints on the correlation functions of the operator $B^2$ with itself.  

Let us start by computing the logarithmic term in the two-point function $\ip{B^2(0) B^2(x)}$.   This correlation function only includes $3$-loop contributions, so \emph{a priori} should be highly non-trivial. 

One way to compute this would be to use the fact that any logarithmic terms in the self-dual YM correlation functions must be canceled by the exchange of the axion field $\rho$.  There are two diagrams involving a $\rho$ exchange which can contribute to this correlation function, depicted in Figure \ref{Figure:BBcorrelator}.  However, one of them is rational, so that the logarithmic correlation function can be computed by a single diagram.
\begin{figure}

\includegraphics[scale=0.25]{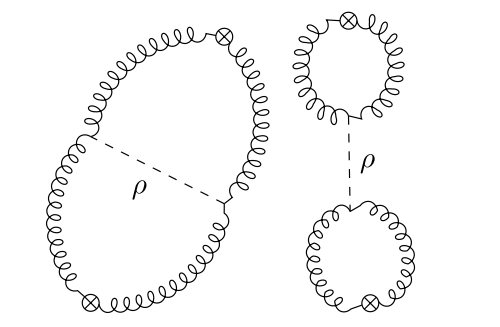}

	\caption{The two diagrams that must cancel any logarithms in the two-point function of $B^2$. The second diagram turns out to be a rational function, so does not play a role.  \label{Figure:BBcorrelator}}

\end{figure}

Instead of using this method, we will use our knowledge of the RG flow of the operator $B^2$.  We have
\begin{equation} 
	(\sum w_{\alpha \dot{\alpha}}\partial_{w_{\alpha \dot{\alpha}}} - 8) \ip{B^2(0) B^2(w)} =  - \frac{\lambda_{\g}^2}{(2 \pi)^4 3 \cdot 2^3}   \ip{ F^2(0), B^2(w) }   
\end{equation}
Now, 
\begin{equation} 
	\begin{split} 
		\ip{F^2(0), B^2(w)} &=  2^{-7} \ip{ F_{\dot{\alpha}_1 \dot{\beta}_1}(0)    B_{\alpha_1 \beta_1}(w) } \ip{ F_{\dot{\alpha}_2 \dot{\beta}_2}(0) B_{\alpha_2 \beta_2}(w) } \eps^{\dot{\alpha}_1 \dot{\alpha}_2} \eps^{\dot{\beta}_1 \dot{\beta}_2} \eps^{\alpha_1 \alpha_2} \eps^{\beta_1 \beta_2}   \\
		&= - \frac{1}{2^{11} (2 \pi)^4 \norm{w}^{8} }.	
	\end{split}
\end{equation}
Therefore,
\begin{equation} 
	(\sum w_{\alpha \dot{\alpha}}\partial_{w_{\alpha \dot{\alpha}}} - 8) \ip{B^2(0) B^2(w)} =   \frac{\lambda_{\g}^2}{(2 \pi)^8  3 \cdot 2^{14} \norm{w}^8 }    
\end{equation}
It follows that
\begin{equation} 
	\ip{B^2(0), B^2(w)} =   \frac{\lambda_{\g}^2}{(2 \pi)^8  3 \cdot 2^{14} \norm{w}^8 }   \log \norm{w} + \text{ rational }  
\end{equation}

\section{The Losev-Moore-Nekrasov-Shatashvili scheme for one-loop quantization of $WZW_4$} 
\label{section_nonlocal_vertices}

When we computed the logarithmic OPE in self-dual Yang-Mills, a key role was taken by the non-local vertex which one needs to adjoin to the theory on twistor space to cancel the gauge anomaly. This vertex is
\begin{equation} 
	\frac{1}{(2 \pi)^4 64} 	\int_{x \in \R^4} \int_{z,z' \in \CP^1} \mc{A}(x,z) \d \mc{A}(x,z)\frac{\d z \d z'}{(z - z')^2}  \mc{A}(x,z') \d \mc{A}(x,z'). \label{eqn_nonlocal_general} 
\end{equation}
However, for the twistor presentation of $WZW_4$, there is more than one choice of non-local vertex that cancels the anomaly. This is because $\mc{A}$ is required to have a zero at $z = 0$, $z = \infty$, so there is more freedom in the choice of meromorphic volume form $(z-z')^{-2} \d z \d z'$.

The most general choice of non-local vertex is of the form
\begin{equation} 
	\int_{x \in \R^4} \int_{z,z' \in \CP^1} \mc{A}(x,z) \d \mc{A}(x,z)\omega(z,z')  \mc{A}(x,z') \d \mc{A}(x,z'). 
\end{equation}
	where $\omega(z,z')$ is a meromorphic volume form on $\CP^1 \times \CP^1$, with the following properties:
\begin{enumerate}
	\item It is anti-symmetric under the exchange $z \mapsto z'$. The anti-symmetry arises because we are symmetrizing the Lie algebra indices, and 
	\item It has second-order poles on the diagonal $z=z'$, whose polar part is determined by anomaly cancellation.
	\item It can have second-order poles when $z,z'$ is zero or infinity.  These cancel the second-order zeroes in $\mc{A} \partial \mc{A}$.
	\item Since $WZW_4$ is invariant under the $U(1)$ rotating $z$, we ask that $\omega(z,z')$ also has this feature.
\end{enumerate}
Properties (2) and (3) means that $\omega(z,z')$ is well-defined up to the addition of an expression like
\begin{equation} 
	\alpha(z) \beta(z')  
\end{equation}
where $\alpha$, $\beta$ are meromorphic $1$-forms on $\CP^1$ with second-order poles at $0,\infty$.  There are three such meromorphic $1$-forms, namely $\d z$, $z^{-1} \d z$, $z^{-2} \d z$.  Anti-symmetry under the exchange of $z,z'$ and $U(1)$ invariance means that the most general two-form appearing in the non-local vertex is of the form
\begin{equation} 
	\frac{\d z \d z'}{(z-z')^2} + c_1 \d z \d z'( z^{-2} + (z')^{-2}) + c_2 \frac{\d z \d z'}{z z'}.
\end{equation}
Here $c_1,c_2$ are coupling constants.

This tells us that we should expect two potential one-loop counter-terms in $WZW_4$, with coupling constants $c_1,c_2$.  This is indeed the case. We can write them out to quartic order in the field $\phi = \log \sigma$.  They are 
\begin{equation} 
	\begin{split}
		\op{Tr}(\t^{(a} \dots \t^{d)}) \omega_{i \br{l}} \omega_{k \br{j}}   \int_{\R^4} \phi_a \phi_b \partial_{u_i} \partial_{\ubar_{\br{j}}} \phi^c \partial_{u_k} \partial_{\ubar_{\br{l}}} \phi_d\\ 
\op{Tr}(\t^{(a} \dots \t^{d)})  		 \omega_{i \br{k}} \omega_{j \br{l}}   	 \int_{\R^4} \phi_a \phi_b \partial_{u_i} \partial_{u_j}\phi^c   \partial_{\ubar_{\br{k}}}  \partial_{\ubar_{\br{l}}} \phi_d  
	\end{split}
\end{equation}
where $\omega$ is the K\"ahler form \footnote{There are of course other ways to write these expressions using integration by parts.  A certain linear combination of these terms is $SO(4)$ invariant, and is the four-derivative term added to the usual $\sigma$-model Lagrangian in the in the Skyrme model. }.


Losev, Moore, Nekrasov and Shatashvili   \cite{Losev:1995cr} showed that the one-loop counter-terms are uniquely fixed by requiring that we do not break certain symmetries present at tree-level. In this section, we will prove their result by twistor methods. We will show that the only way to preserve these symmetries is to set the coupling constants $c_1,c_2= 0$.  

This will allow us, in the next section, to show that the Green-Schwarz mechanism canceling logarithmic OPEs holds for $WZW_4$ with the group $SO(8)$.

\subsection{Infinite dimensional symmetries in $WZW_4$ from the twistor perspective}
In \cite{Losev:1995cr}, it was shown that $WZW_4$ has a very large symmetry group at the classical level.  The field $\sigma : \R^4 \to G$ can be multiplied on the left by a holomorphic function $F(u) : \R^4 \to G$, and on the right by an anti-holomorphic function $\til{F}(\br{u}) : \R^4 \to G$, and the equations of motion are preserved.

This, of course, is very similar to the large symmetry group of the familiar two-dimensional WZW model, and explains the name $WZW_4$. 

This symmetry group is easy to understand from the twistor perspective. A solution to the equations of motion of the theory on  twistor space  is a holomorphic $G$-bundle on $\PT$, trivialized at $z = 0$ and $z = \infty$.  At $z = 0,\infty$ we can multiply the trivialization by any holomorphic map from $\C^2 \to G$, where $\C^2$ is the fibre at $z = 0$ or $\infty$ of the fibration $\PT \to \CP^1$.  From the point of view of $\R^4$, changing the trivialization at $z = 0$ changes the map $\sigma : \R^4 \to G$ by left multiplying by a holomorphic map, and the change of trivialization at $\infty$ modifies it by right multiplication by an anti-holomorphic map.

We would like these symmetries to persist to the quantum level. To understand  quantum symmetries, it is better to think of coupling to a background gauge field in a gauge-invariant manner. To give an action of the group of  holomorphic maps from $\R^4 \to G$ is the same as to couple to a background gauge field
\begin{equation} 
	A_L \in \Omega^{0,1}(\R^4, \mf{g}).
\end{equation}
This gauge field must satisfy $F^{0,2}(A_L) = 0$.  In $WZW_4$, $A_L$ is coupled to the field $\sigma$ by
\begin{equation} 
	\int \omega \op{Tr} ( A_L \sigma^{-1} \partial \sigma)  . 
\end{equation}
The equations of motion of $WZW_4$ are $\omega \dbar (\sigma^{-1} \partial \sigma) = 0$, and these imply that this is invariant under gauge transformations of $A_L$. 

Similarly, the to give an action of the group of anti-holomorphic maps is to couple to a field $A_R \in \Omega^{1,0}(\R^4,\mf{g})$ in a gauge invariant way. The coupling is
\begin{equation} 
	\int \omega \op{Tr} (A_R (\dbar \sigma) \sigma^{-1} )  
\end{equation}
and again, this is gauge invariant by the equivalent formulation of the equations of motion $\omega \partial ( (\dbar \sigma)\sigma^{-1} ) = 0$. 

From the point of view of twistor space, coupling to the background gauge fields $A_L$, $A_R$ is achieved by modifying the boundary conditions at $z = 0,\infty$.  We ask that the holomorphic Chern-Simons gauge field $\mc{A}$ satisfies
\begin{equation} 
	\mc{A}_{z = 0} = A_L \ \ \mc{A}_{z = \infty} = A_R. 
\end{equation}
This makes sense, because $A_L$ is a $(0,1)$ form in the complex structure at $z = 0$, and $A_R$ is a $(0,1)$ form in that at $z = \infty$.  

The holomorphic Chern-Simons action is gauge-invariant when we impose these modified boundary conditions. Indeed, the gauge variation of the holomorphic Chern-Simons action is
\begin{equation} 
	\int \d v_1 \d v_2 \frac{\d z}{z^2} \dbar ( \chi F(\mc{A})  
\end{equation}
where $\chi$ is the gauge parameter, which we assume to vanish at $z = 0$, $z = \infty$.  Since, by assumption, $F(\mc{A})$ also vanishes at $z = 0$, $z = \infty$, integration by parts tells us that this expression vanishes.

The action is also covariant under gauge transformations of the fields $A_L$, $A_R$ at $z = 0$, $z = \infty$.  Consider the effect of a gauge transformation 
\begin{equation} 
	A_L \mapsto A_L + \dbar\psi + [A_L,\psi] 
\end{equation}
We need to show that holomorphic Chern-Simons theory does not change when we modify the boundary conditions in this way. To see this, let us extend $\psi : \R^4 \to \mf{g}$ to a map $\chi : \PT \to \mf{g}$,  with $\chi = \psi$ at $z = 0$.  The choice of extension does not matter, because any two extensions differ by one which vanishes at $z = 0$, which is therefore a gauge transformation in the bulk.  

We can transform the bulk field $\mc{A}$ by $\chi$. Under this transformation, the Chern-Simons action changes to
\begin{equation} 
	\int \d v_1 \d v_2 \frac{\d z}{z^2} \dbar (\chi F(\mc{A})).	 
\end{equation}
Integration by parts, together with the fact that $F(\mc{A}) = 0$ at $z = 0$, tells us that this is
\begin{equation} 
	\int_{z = 0} \d v_1 \d v_2 \eta \partial_z F(\mc{A}) . 
\end{equation}
The equation $\partial_z F(\mc{A}) = 0$ holds by the equations of motion of the bulk theory. 

We conclude that we can modify the boundary conditions as above, in a way both invariant under bulk gauge transformations and compatible with boundary gauge transformations. This encodes the infinite dimensional symmetry of $WZW_4$ in the twistor space picture.

\subsection{Aside: $w_{\infty} \oplus \br{w}_{\infty}$ symmetries of the closed-string sector}
As an aside, we note that a variant of this analysis applies to the closed-string sector.  In the closed-string sector, our fundamental field is a Beltrami differential $\mu$ on $\PT$, which is divergence free and which is such that the $(2,1)$-form
\begin{equation} 
	\mu \vee \d v_1 \d v_2 \frac{\d z}{z^2} 
\end{equation}
has no poles at $z = 0$, $z = \infty$.  This means that $\mu$ has a second-order zero at $z = 0$, $z = \infty$.

We can modify the boundary conditions by asking that $\mu$ approaches some fixed Beltrami differential at $0$,$\infty$, instead of tending to zero. The simplest modification is to choose a Beltrami differential
\begin{equation} 
	\mu_L = \mu^{\br{i}}_{L,j} \d \vbar_i \partial_{v_j} 
\end{equation}
at $0$ (and similarly $\mu_R$ at $\infty$).  We ask that this is divergence free,
\begin{equation} 
	\partial_{v_j}  \mu^{\br{i}}_{L,j} = 0. 
\end{equation}
We also ask that $\mu_L$ satisfies the integrability condition
\begin{equation} 
	\dbar \mu_L + \tfrac{1}{2}[\mu_L,\mu_L] =0, 
\end{equation}
so that it defines an integrable holomorphic-symplectic structure on $\C^2$. We will similarly choose a Beltrami differential $\mu_R$ at $z = \infty$. 

The modified boundary conditions are that 
\begin{equation} 
	\begin{split}
		\mu &= \mu_L + O(z^2)\\  
		\mu &= \mu_R + O(z^{-2}) 
	\end{split}
\end{equation}
Just as in Chern-Simons theory, the closed-string theory is gauge invariant with these modified boundary conditions. Further, the coupled system is invariant under the gauge transformations of the boundary Beltrami differential $\mu_L$, $\mu_R$ by divergence-free vector fields on the boundary:
\begin{equation} 
	\mu_L \mapsto \dbar V - [V,\mu] 
\end{equation}
where $V = V_i \partial_{v_i}$ is a vector field on the boundary satisfying $\partial_i V_i = 0$. 

This tells us that the four-dimensional theory corresponding to the closed-string fields can be coupled to background Beltrami differentials in a gauge invariant way.  

As in  Chern-Simons theory and $WZW_4$, these couplings reflect infinite dimensional symmetries of the four-dimensional model.  The symmetries are those gauge transformations of $\mu_L$, $\mu_R$ which preserve the zero field configuration $\mu_L = 0$, $\mu_R = 0$.  These gauge transformations are holomorphic (or anti-holomorphic) divergence-free vector fields on $\C^2$.  Thus, they form a copy of $w_{\infty}$, at $z = 0$, and its complex conjugate $\br{w}_{\infty}$, at $z = \infty$.   

It would be very interesting to explicitly verify that the four-dimensional Lagrangian associated to the closed-string fields has this $w_{\infty} \oplus \br{w}_{\infty}$ symmetry.  Probably this will fix the quartic and higher terms in the self-interaction of the K\"ahler potential that are not determined in this paper.  

\subsection{Constraining the non-local vertex using the infinite-dimensional symmetries}
We want to show that the infinite-dimensional symmetries of $WZW_4$ are violated if we add  to the non-local vertex the terms 
\begin{equation} 
	\int_{x \in \R^4} \int_{z,z' \in \CP^1} \mc{A}(x,z) \d \mc{A}(x,z)\omega(z,z')  \mc{A}(x,z') \d \mc{A}(x,z') \label{eqn_nonlocal_general2} 
\end{equation}
where $\omega(z,z')$ is a linear combination
\begin{equation} 
	c_1 \d z \d z'( z^{-2} + (z')^{-2}) + c_2 \frac{\d z \d z'}{z z'} 
\end{equation}
To see this, let us study gauge invariance of the Lagrangian where we as before ask that $\mc{A} = A_L$ at $z = 0$, $z = \infty$. 

Let us vary $\mc{A}$ by a gauge transformation $\mc{A} \mapsto \mc{A} + \dbar \chi + [\mc{A},\chi]$. We don't require that $\chi(0) = 0$, instead we also vary $A_L \mapsto A_L + \dbar \chi(0) + [A_L,\chi(0)]$.  The terms in equation \eqref{eqn_nonlocal_general2} fail to be gauge invariant under such a transformation, because they will have a boundary term at $z = 0$ (and at $z = \infty$). We will focus on the boundary terms at $z = 0$, as those at $z = \infty$ are the same.

The gauge variation of the non-local vertex associated to $\d z \d z' (zz')^{-1}$ is
\begin{equation}
	\int_{u,z,z'} \dbar \chi(z) \partial \mc{A}(z) \frac{\d z \d z'}{zz'} \mc{A}(z') \partial \mc{A}(z').  	
\end{equation}
Integration by parts shows that this is
\begin{equation}
	\int_{u,z'}  \chi(0) \partial A_L \frac{\d z'}{z'} \mc{A}(z') \partial \mc{A}(z').  	
\end{equation}
This expression is clearly non-zero.

Similarly, if we use the two-form $\d z \d z'(z^{-2} + (z')^{-2})$, we get
\begin{equation}
	\int_{u,z,z'} \dbar \chi(z) \partial \mc{A}(z) \d z \d z'(z^{-2} + (z')^{-2})  \mc{A}(z') \partial \mc{A}(z').  	
\end{equation}
Integration by parts brings this to
\begin{equation}
	\begin{split}
	\int_{u,z,z'} \partial_z (\chi(z) \partial \mc{A}(z))_{z = 0}   \d z' \mc{A}(z') \partial \mc{A}(z') \\
		=\int_{u,z,z'} \chi(0) \partial_z \partial \mc{A}(0) \d z'  \mc{A}(z') \partial \mc{A}(z') + \int_{u,z,z'}\partial_z \chi(0)  \partial A_L(0)  \d z' \mc{A}(z') \partial \mc{A}(z')  
	\end{split}
\end{equation}
The last term, involving $\partial_z \chi(0)$, tells us that the action is not gauge invariant even under purely bulk gauge transformations, which vanish at $z = 0$.

We conclude that the only non-local vertex we can use, which cancels the anomaly and is compatible with the infinite-dimensional symmetries of the model, is that associated to the two-form $\frac{\d z \d z'}{(z-z')^2}$. 

\section{The Green-Schwarz mechanism for $WZW_4$}
\label{sec:wzw4_gs}
We have  shown that the non-local vertex canceling the anomaly on twistor space must be the same one used in self-dual Yang-Mills.  This allows us to compute the two-loop OPEs of the operator $\phi = \log \sigma$.

We will show the following.
\begin{proposition}
There is a two-loop logarithmic OPE	
	\begin{equation} 
		\phi^a(0) \phi^b(u) \sim \log \norm{u} (\partial \dbar \phi^c)(0) \wedge ( \partial \dbar \phi^d) (0) \op{Tr}( \t_{(a} \dots t_{d)}). 
	\end{equation}
	\label{prop:log_wzw}	
\end{proposition}

The first step in the proof is to calculate the field sourced by the operator $\phi(0)$. The twistor space representation of this operator is 
\begin{equation} 
	\phi(0) + O(\phi^2) = \int_{\CP^1} \frac{\d z}{z} \mc{A}. 
\end{equation}
(We are only interested in the lowest-order logarithmic terms in the OPE, when expanded as a series in $\phi$; thus adding higher-order expressions to our operator will not change the result). 

As in our discussion in the case of self-dual Yang-Mills,  the field sourced by the operator placed at $u,\ubar$ is  a field $\mc{A}$ which satisfies
\begin{equation} 
	\dbar \mc{A} = z \delta_{v_1 = u_1 + z \ubar_2} \delta_{v_2 = u_2 - z \ubar_1}.  
\end{equation}
We will place one operator at $u^{(1)}$, and one at $u^{(3)}$, and solve the source equations in an axial gauge:
\begin{equation}
\begin{split}
\mc{A}^{(3)} &= z \delta_{v_1 = u_1^{(3)} + z \ubar_2^{(3)} } \frac{1}{ v_2 - u_2^{(3)} + z \ubar_1{(3)}  }\\
\mc{A}^{(1)} &=  z\delta_{v_2 = u_2^{(1)} - z \ubar_1^{(1)} } \frac{1}{ v_1 - u_1^{(1)} - z \ubar_1{(1)}  }\\
\end{split}
\end{equation}
These are the same expressions used in the case of self-dual Yang-Mills.  The OPE calculation is then the same. In the self-dual Yang-Mills case, we found the OPE of the non-gauge invariant operator $B_{12}$ was
\begin{equation} 
	\log \norm{u^{(12)}} \left( F(A)^2\right) 
\end{equation}
We can immediately transform this into the result for $WZW_4$, noting that the non-local vertex for $WZW_4$ is the same as that for self-dual Yang-Mills upon setting $A = J^{0,1}$. This leads to the result stated in the above proposition.

\subsection{Canceling the log OPE with closed-string fields}
To show that the introduction of the closed-string fields cancels this anomaly, we will follow the analysis used for self-dual Yang-Mills. There, the field $\rho$ is coupled to the gauge field by $\int \rho F(A)^2$, and this leads to the cancellation of the log OPE.

To show that this works for $WZW_4$ as well, it suffices to show that the coupling between $\rho$ and the fields of $WZW_4$ are the same as those between self-dual Yang-Mills and $\rho$, upon setting $A = J^{0,1}$.  
The coupling is, as in \eqref{open_closed1}
\begin{equation} 
	\int \op{Tr}( J^{0,1}\partial J^{0,1}) \partial \rho. 
\end{equation}
This is the same as $2 \int CS(J^{0,1}) \d \rho$. Indeed, the cubic term in the Chern-Simons $3$-form drops out because we have a $(0,1)$ form, and the quadratic term is a form of type $(1,2)$. As such, wedging it with $\dbar \rho$ gives zero, so we can replace $\d \rho$ by $\partial \rho$.

The field sourced by $\phi$ in $WZW_4$ is, when turned into a field of self-dual Yang-Mills, the same as the field 
sourced by $B_{12}$. This tells us that the cancellation we computed for self-dual Yang-Mills applies here without any change.

\subsection{Potential cancellation of other logarithmic OPEs}
We have seen that one of the terms in the usual Green-Schwarz mechanism on twistor space leads to a Green-Schwarz like cancellation of logarithmic OPEs.  There are two other terms as part of the Green-Schwarz cancellation on twistor space.  These two other terms were not calculated explicitly in \cite{1905.09269}, rather it was shown by an inductive argument that they can be canceled by some choice of local coupling between the open and closed string fields.

These anomalies, and their cancellation, are sketched in Figure \ref{Figure:GS_other_terms}. 
\begin{figure}
\includegraphics[scale=0.25]{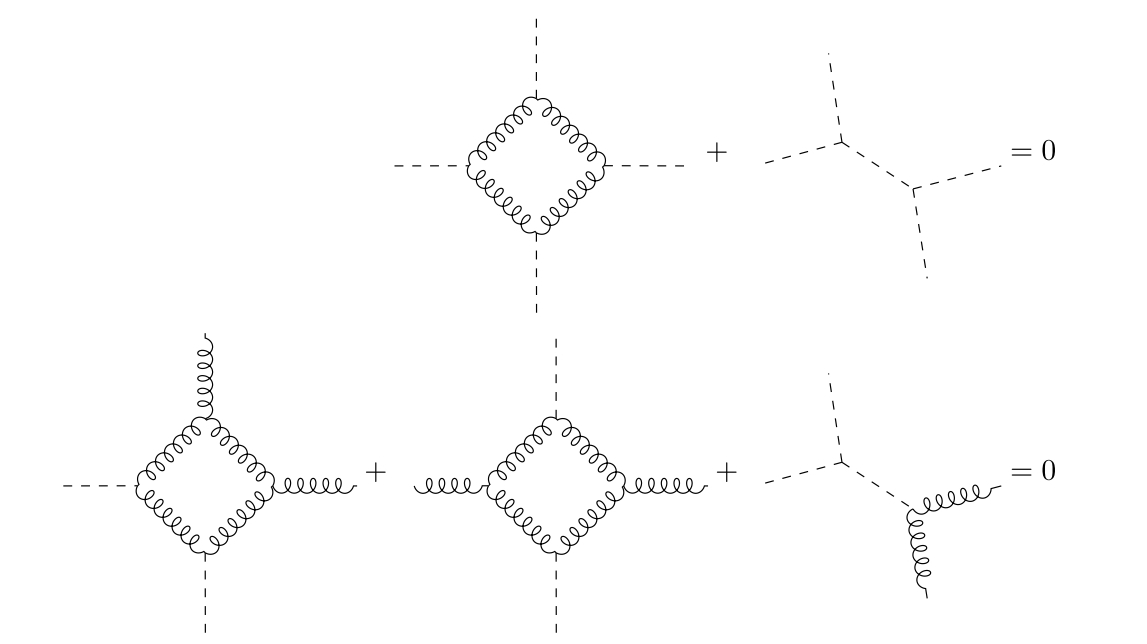}

	\caption{The two other terms in the Green-Schwarz anomaly cancellation on twistor space.  The dashed lines are closed-string fields. The second line constraints the coefficient of the self-interaction of the closed-string field, and the first line should constrain the dimension of the gauge Lie algebra to be $28$.     
	\label{Figure:GS_other_terms}}
\end{figure}

It would be very interesting, but challenging, to investigate the space-time version of these terms in the Green-Schwarz cancellation.

\section{$D$-strings and integrability in $WZW_4$}
In the beautiful paper \cite{2011.04638}, it was shown that the two-dimensional principal chiral model, for arbitrary values of the WZW term, can be obtained by dimensional reduction of the 4d WZW model on the geometry $T^2 \times \R^2$, as the radius $R$ of the $T^2$ is sent to zero.

The coefficient of the WZW term is determined by the ratio of the K\"ahler volume of $T^2$ and the Riemannian volume of $T^2$.  If they are the same, so that the $T^2$ is holomorphic,  then the resulting model is the conformal WZW model.  In the other extreme, if the $T^2$ is Lagrangian, so that its K\"ahler volume is zero, the model is the principal chiral model with no WZW term.   

Since the PCM is integrable, we can ask if the 4d uplift remains integrable.  One manifestation of integrability of the 2d theory is that there is a Lax matrix $\mc{L}(z)$. The path-ordered exponential of the Lax matrix defines a line defect which is topological, i.e.\ can be freely moved around. These topological line defects commute with each other for different values of $z$.

From the perspective of four-dimensional Chern-Simons theory \cite{1303.2632,1709.09993}, these topological line defects arise from Wilson lines.  As shown in \cite{2011.04638, 2011.05831}, holomorphic Chern-Simons on twistor space reduces to four-dimensional Chern-Simons. The uplift of the Wilson lines of four-dimensional Chern-Simons are surface defects wrapping holomorphic curves in twistor space.  

Thus, a natural four-dimensional counterpart of integrability is the presence of surface defects which wrap planes holomorphic in some complex structure $z$.  The 4d version of the statement that these are topological is that the defects vary holomorphically as a function of their position.   The 4d version of the statement that these operators commute with each other is that the OPE between these defects is non-singular for different values of $z$.   

There is an important distinction between $4d$ and $2d$ integrability. In $WZW_4$,  we necessarily have gravitational fields, which are not present in $2d$.  These gravitational fields -- together with the fact that the model arises from a string theory -- forces us to have $D$-strings, rather than simply surface defects. The distinction is that with $D$-strings, the fluctuations of the object in space-time are dynamical.

\subsection{$D$-strings and integrability}
In Witten \cite{Witten:2003nn} and Berkovits'  \cite{Berkovits:2004hg}  proposals for twistor string theory, $D$-strings on twistor space played an important role.

Here we will consider similar $D$-strings in our set-up. We find that they wrap surfaces in $\R^4$, and behave like strings in the four-dimensional model. 

Two dimensional integrable models have particles which satisfy the Yang-Baxter equation.  We expect that the strings we study in the four-dimensional model satisfy a similar Yang-Baxter equation. We derive this at a formal level, by noting that the strings do not touch each other in twistor space.  A more detailed computation of the Yang-Baxter equation is left for future work.  

Because the theory only lifts to twistor space for $G = SO(8)$, we do not expect that integrability holds for other groups\footnote{Except, as mentioned before, for the possibility of $G_2 \times G_2$.}. 

The $D$-strings we are discussing are branes in the type I topological string on twistor space. The worldvolume of the brane is a holomorphic curve in $\PT$. In the embedding of this theory in the physical type IIB string, these come from $D3$ branes on an $O7^-$ plane with $4$ $D7$'s, in an $\Omega$-background.  

Consider a curve $C$ in a Calabi-Yau threefold $X$.  The theory on a single $D1$ brane wrapping the curve is the following.  
\begin{enumerate} 
	\item 	 There are $8$ pairs symplectic bosons $\beta^{(i)}_1$, $\beta^{(i)}_2$,  for $i = 1 \dots 8$.  These are of spin $1/2$ and have the OPE
		\begin{equation} 
			\beta^{(i)}_r \beta^{(j)}_s = \delta^{ij} \eps_{rs}.  
		\end{equation}
		These are BRST reduced by the diagonal action of the group  $SL_2$.  There is a remaining $SO(8)$ flavour symmetry. 
	\item There is an additional pair of symplectic bosons $\chi_r$  transforming in the normal bundle $N$ to $C$ in $X$. (Note that $\wedge^2 N = K_C$, which is what is needed to define symplectic bosons living in a rank $2$ bundle).  These describe motion of the curve $C$ in $X$.
\end{enumerate}
If we do not consider the symplectic bosons describing the motion of $C$, the chiral algebra is that associated by the work of \cite{1412.7541} to a well-known $N=2$ superconformal theory in dimension $4$.  This is the $SU_2$ gauge theory with $4$ fundamental hypermultiplets, in which the obvious $SU_4$ flavour symmetry is enhanced to $SO(8)$.  The Higgs branch of the theory is the moduli space of charge $1$ instantons on $\R^4$ for the group $SO(8)$.  The chiral algebra is closely related to the curved symplectic bosons for this holomorphic symplectic manifold.  

This theory is coupled to the closed-string fields $A \in \Omega^{0,1}(X,\mf{so}(8))$ and $\eta \in \Omega^{2,1}(X)$ by
\begin{equation} 
	\int_{C} A J + \int_{C} \partial^{-1} \eta  + \dots 
\end{equation}
where $J$ is the current for the $SO(8)$ flavour symmetry and  $\dots$ indicates an infinite number of terms which depend on the fields $\chi_r$ moving the position of the brane. (These can be determined by taking the first two terms for the brane $C + \chi_r$ whose position has been varied by $\chi_r$, and expanding in series in $\chi_r$).    

We note that the coupling to the closed-string field $\eta$ is not quite local.  If we write the curve $C$ as the boundary of a $3$-manifold $M$, than this term can be written as $\int_M \eta$. Because $\eta$ is closed on-shell, this expression does not change if we vary $M$ continuously. 

We need to explain how to modify the theory on the $D1$ brane when we work on $\PT$.   The only thing we need to change is the behaviour of the symplectic bosons $\chi_r$ that describe the normal variation of the brane.  Take a curve $C \subset \PT$ and let $D \subset C$ be the divisor where the function $z$ is zero or infinity. (We will assume that $C$ does not lie entirely in the locus $z = 0$, $z = \infty$). The three-form $\Omega$ on $\PT$ has poles of order two at $z= 0$, $z = \infty$, so that if $N$ is the normal bundle to $C$ in $\PT$, we have an isomorphism
\begin{equation} 
	\wedge^2 N \iso K_C (2D). 
\end{equation}
Therefore the symplectic bosons describing the variation of the brane should like in $N(-D)$, as we have $\wedge^2 (N(-D)) = K_C$.  This means that position of the brane at $z = 0$, $z = \infty$ is not dynamical. 

One can of course place $D1$ branes which wrap the fibres of the map $\PT \to \R^4$.  These will correspond to certain local operators in the theory on $\R^4$.  Because the theory on the brane involves the integral of $\partial^{-1} \eta$ along the brane, the operator measures the value of the K\"ahler potential $\rho$ at a point in $\R^4$.  For these vertical $D1$ branes, the symplectic bosons describing the variation of the position live in $\Oo(-1)$, and so have no zero modes.  This means that the position of these defects is not dynamical, and they should not be thought of as $D$-strings.  

For the purposes of understanding integrable features of these theories, we are more interested in placing $D1$ branes along non-compact curves which live at a particular $z \in \CP^1$.  Given a polynomial $F(v_1,v_2)$ and $z_0 \in \CP^1$, we can place a surface defect on the locus $z = z_0$, $F(v_1,v_2) = 0$.  This gives us a surface defect in the corresponding theory on $\R^4$ which wraps the submanifold cut out by the real and imaginary parts of the equation
\begin{equation} 
	F(u_1 + z_0 \ubar_2, u_2 - z_0 \ubar_1) = 0. 
\end{equation}
This surface defect is given by coupling to the four-dimensional system a system of gauged symplectic bosons $\beta^{(i)}_r$ as above, together with a second system of symplectic bosons $\chi_1,\chi_2$ describing the normal variation of the brane.  

Let us describe how this theory is coupled to the $4d$ fields. A representative on twistor space for the field $\sigma$ is given (away from a small neighbourhood of $z = 0$) by 
\begin{equation} 
	A^{0,1} = J_{\ubar_i} \pi^{0,1} \d \ubar_i 
\end{equation}
where $\pi^{0,1}$ indicates the projection onto a $(0,1)$-form on twistor space.  We then integrate this against the two-dimensional current $J_{2d}$, viewed as a $(1,0)$ form on the brane.  Wedging $\pi^{0,1} \d \ubar_i$ with $J_{2d}$ is the same as wedging $\d \ubar_i$ with $J_{2d}$. Therefore, we find that the symplectic bosons are coupled to $\sigma$ by the expression
\begin{equation} 
	\int_{C} \op{Tr} J_{2d} J_{4d}^{(0,1)}. 
\end{equation}

Let us compute how this defect couples to the K\"ahler potential (to leading order).  From the expression on twistor space, it is clear that, to leading order, the coupling between the defect and the K\"ahler potential $\rho$ involves the integral of some derivatives of $\rho$ over the defect, and does not involve the fields on the defect.  (There are higher order terms which do involve the fields on the defect).  

We will calculate that, to leading order in $\rho$, the Lagrangian is 
\begin{equation} 
	\int_{C} \d^2 C \Lap \rho  \label{eqn_curve_coupling} 
\end{equation}
where $\d^2 C$ is the Riemannian metric on $C$ coming from the flat metric on $\R^4$.

To check this, we note that the Riemannian volume of $C$ is obtained by considering the K\"ahler form in which $C$ is holomorphic, and integrating it over $C$. The curve $C$ is holomorphic in the complex structure $z_0$ for which holomorphic coordinates are $w_1 = u_1 + z_0 \ubar_2$, $w_2 = u_2 - z_0 \ubar_1$.   Before deforming the flat metric by the $\partial \dbar \rho$, the K\"ahler form in complex structure $z_0$ is 
\begin{equation} 
	\omega_{z_0} =  \frac{1}{1 + \abs{z_0}^2 } \d w_i \d \wbar_{i}.
\end{equation}
The prefactor is there to ensure that $\omega_{z_0}^2 = \omega^2$, where $\omega$ is the initial K\"ahler form.  Therefore, the expression \eqref{eqn_curve_coupling} can be written as 
\begin{equation} 
	\int_{C} \omega_{z_0} \Lap \rho = \int_{C} \omega_{z_0} \partial_{z_0} \dbar_{z_0} \rho.	 
\end{equation}

We will verify that this is the correct coupling by computing the two-point function of the operator $\rho$ with the defect.  Since the operator $\rho$ is defined by the integral over a twistor fibre of $\partial^{-1} \eta$, we can do this by studying the field sourced by the defect on twistor space. 

The field sourced by the defect satisfies the equation
\begin{equation} 
	\dbar \eta = \delta_{C} 
\end{equation}
where $C \subset \PT$ is the curve supporting the brane on twistor space.  A solution to this equation is
\begin{equation} 
	\eta = \delta_{z = z_0} \frac{\d F(v_1,v_2) }{F(v_1,v_2)}  
\end{equation}
(where as above $C$ is the curve $z = z_0$, $F(v_1,v_2) = 0$. As a representative for $\partial^{-1} \eta$ we can take
\begin{equation} 
	\partial^{-1} \eta = \delta_{z = z_0} \log F(v_1,v_2).  
\end{equation}
To compute the two-point function of the operator $\rho(u_i,\ubar_i)$ with the defect we need to integrate this against the curve where $v_1 = u_1 + z \ubar_2$, $v_2 = u_2 - z \ubar_1$. The result of this integral is the field sourced by the defect on $\R^4$, which is 
\begin{equation} 
	\rho = \log F(u_1 + z_0 \ubar_2, u_2 - z_0 \ubar_1)	 
\end{equation}

Note that this satisfies the equations of motion. In fact, $\log F$ is a harmonic function.  To see this, we note that $\Lap \log F$ is one of the components of $\partial \dbar \log F$ where we define $\partial$ and $\dbar$ in the complex structure for which $F$ is holomorphic.

From this we can see that the term in the Lagrangian coupling $\rho$ to the defect is (to linear order)
\begin{equation} 
	\int \rho \Lap^2 \log F \d^2 u \d^2 \ubar = \int (\Lap \rho) \Lap \log F \d^2 u \d^2 \ubar 
\end{equation}
This is an integral over the curve $C$ because $\Lap \log F$ is a certain derivative of the delta-function on $C$.  

We can make this more explicit.  Let $w_1 = u_1 + z_0 \ubar_2$, $w_2 = u_2 - z_0 \ubar_1$ as above, and $\omega_{z_0}$ be the K\"ahler form in this complex structure. 
Then,  
\begin{equation} 
	 \Lap \log F \d^2 w \d^2 \wbar = \omega_{z_0} \partial_{z_0} \dbar_{z_0} \log F  
 \end{equation}
where the subscript $z_0$ indicates that $\partial$, $\dbar$ are taken in this complex structure.  Now, since $F$ is holomorphic in this complex structure,
\begin{equation} 
	\partial_{z_0} \dbar_{z_0} \log F = \delta_{C}. 
\end{equation}
Thus, the term in the Lagrangian is
\begin{equation} 
	\int_{C} \omega_{z_0} \Lap \rho,  
\end{equation}
as desired.

\subsection{The Yang-Baxter equation for defects}
So far we have described how the fields of the $D$-string couple to the four-dimensional theory. In this section we will derive, at a very formal level, the Yang-Baxter equation satisfied by the scattering of $D$-strings. I hope to come back to a more detailed discussion of this point in the future.

The derivation given here is very similar to the derivation of the Yang-Baxter equation from four-dimensional Chern-Simons theory \cite{1303.2632,1709.09993}. 

To understand the YBE, we should first consider what happens when two $D$-strings which don't touch in $\PT$ cross when viewed as defects in $\R^4$. Since the strings don't touch in $\PT$, this configuration is well defined and leads to a configuration of defects in $\R^4$. This means that, in $\R^4$, we have a pair of surface defect sand a way of coupling them at their intersection points. 

The operator coupling them is a local operator in the system consisting of the two $D$-strings and the four-dimensional gauge fields. These local operators can be computed to leading order in perturbation theory, as they arise by integrating out the fields of the theory on twistor space.  In \cite{1709.09993}, the local operator at the intersection of two lines is the $R$-matrix. The local operator at the intersection of the $D$-strings plays the same role, and can be thought of as the $R$-matrix in the four-dimensional context.

In \cite{1709.09993}, a semi-classical analysis leads us explicitly to the familiar classical $r$-matrix.  I expect similar expressions here, although I have not performed the computations in detail. Schematically, if $J_a^{(1)}$, $J_a^{(2)}$ are the currents for the $SO(8)$ symmetry of the $D$-strings, I expect that the exchange of an open-string field on twistor space should lead to a local operator which is roughly $J_a^{(1)} J_a^{(2)} \frac{1}{z_1 - z_2}$.  The context we are dealing with is quite a bit more complicated than that of \cite{1709.09993}, because of the presence of closed-string field and the fact that we also need to consider how the operator at the intersection of the two defects depends on the fields of $WZW_4$.

Now let us state the Yang-Baxter equation. Suppose we have three $D$-strings placed at values $z_1,z_2,z_3$ of the spectral parameter, and, for simplicity, placed on straight complex lines. We can take the lines to be at 
\begin{equation} 
	\begin{split} 
		v_1 &= 0 \ \ z = z_1 \\
		v_2 &= 0 \ \ z = z_2 \\
		v_1 &= c \ \ z = z_3
	\end{split}
\end{equation}
Let us denote the theory on the $D$-strings by $\mc{S}^{(i)}$, $i = 1,2,3$ and let $R^{(ij)}$ be the local operator living at the intersection of $S^{(i)}$ with $S^{(j)}$.  

Then, on $\R^4$, the $D$-strings lives on the planes
\begin{equation} 
	\begin{split}
		u_1 + z_1 \ubar_2 &= 0 \\
  		u_2 - z_2 \ubar_1 &= 0 \\
		u_1 + z_3 \ubar_2 &= c. 
	\end{split}
\end{equation}
The three intersection points of these planes are the origin and the two points
\begin{equation} 
	\begin{array}{rclrcl}
		u_1 &= &\frac{c}{1 + z_3 \zbar_2}  &   u_2 &= &\frac{z_2 \br{c} } { 1 + \zbar_3 z_2} \\  
		u_1 &=& \frac{z_1 c}{z_1 - z_3} &  u_2 &=& \frac{\br{c}}{\zbar_3 - \zbar_1} 	
	\end{array}
\end{equation}
Evidently as $c \to 0$, the three intersection points coincide. 

The statement of the Yang-Baxter equation is that \emph{there are no singularities in the system consisting of the $4d$ theory plus $D$-strings,  as $c \to 0$}.  

More concretely, we have operators $\mc{O}^{(ij)}$ in the coupled system, living at the intersections of the $D$-strings.  The Yang-Baxter equation is a statement about the a surprising lack of singularities in the OPEs of the operators $\mc{O}{(ij)}$. 

We should compare this to the ordinary Yang-Baxter equation.  There, we have three line defects which we can move past each other without introducing any singularities.  In that case, the line defects are topological and the system does not depend at all on the position of the lines. In our case, the surface defects are not topological, so there is dependence on the parameter $c$.  For both the topological YBE, and the YBE we are considering here, the main content is that the system has no singularities in the parameter $c$ controlling the position of the lines.  The only difference is the nature of the singularities: in the familiar case, the only such singularities could be discontinuities.  In our case, the only possible singularities are poles, because the system depends holomorphically on $c$.

\appendix

\section{Normalizing the open-closed coupling}
\label{sec:normalization_twistor}
In what follows, we will take $\op{Tr}$ to mean the trace in the adjoint representation and $\op{tr}$ to be that in the fundamental representation.  We fix a Lie algebra $\g$ so that $\op{Tr}(X^4)$ is proportional to $\op{tr}(X^2)^2$.  This happens for $\g = \sl_2$, $\sl_3$, $\mf{so}(8)$ or an exceptional algebra.

According to \cite{Gwilliam:2018lpo}, the anomaly associated to the rectangle diagram is
\begin{equation} 
	\frac{1}{(2 \pi \i)^3 4!  }  \int \op{Tr} \left( \chi (\partial \mc{A})^3 \right) \label{eqn_normalizationsdym} 
\end{equation}
Compared to their analysis we get an extra factor of $\half$ from the diagram automorphism when work in holomorphic Chern-Simons, so that then the normalization is  
\begin{equation} 
	\frac{1}{(2 \pi \i)^3 2 \cdot  4!  }  \int \op{Tr} \left( \chi (\partial \mc{A})^3 \right) \label{eqn_normalization_hcs} 
\end{equation}
In what follows we will normalize the open-closed coupling in the holomorphic Chern-Simons case.  The self-dual Yang-Mills case will acquire an extra factor of $\sqrt{2}$, coming from the extra factor of $2$ above.

If we normalize the kinetic term of the closed-string field to be
\begin{equation} 
\tfrac{1}{2} \int \dbar \eta \partial^{-1} \eta 
\end{equation}
and  normalize the interaction between $\eta$ and $\mc{A}$ to be
\begin{equation} 
	\lambda_{\g} C    \int \eta \op{tr}( \mc{A} \partial \mc{A} )  
\end{equation}
for some universal constant $C$ we have yet to determine.  The Lie algebra dependent constant $\lambda_{\g}$ satisfies 
where, for $X \in \mf{g}$, we have
\begin{equation} 
	\lambda_{\g}^2 \op{tr}(X^2)^2 = \op{Tr}(X^4). \label{eqn_lambda} 
\end{equation}
With these conventions, let us compute the anomaly for the exchange of closed string fields. 
\begin{figure}


\includegraphics[scale=0.21]{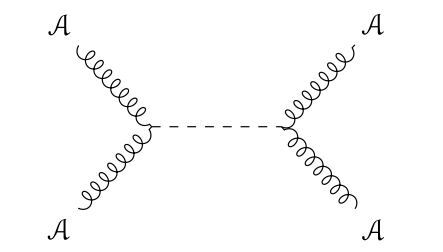}

	\caption{Two open-string fields exchanging a closed-string field gives a tree-level anomaly.  \label{Figure:closed_anomaly_app} }  
\end{figure}

Let us look at  the $\eta$- field sourced by the operator 
\begin{equation} 
	 	\lambda_{\g} C    \int \eta \op{tr}( \mc{A} \partial \mc{A} )  
\end{equation}
Varying $\eta$ and applying the operator $\partial$ gives us 
\begin{equation} 
	\dbar \eta = 	\lambda_{\g} C \op{tr} (\partial \mc{A}^2). 
\end{equation}
When we consider the diagram in Figure \ref{Figure:closed_anomaly_app}, a gauge variation sending $\mc{A} \to \mc{A} + \dbar \chi$ and then integrating by parts gives us
\begin{equation} 
	-2 (\lambda_{\g} C)^2  \int \op{tr} \left( \chi \partial \mc     {A}  \right) \op{tr} \left( (\partial \mc{A})^2   \right). 
\end{equation}
The factor of $2$ comes because there are $4$ external lines, and we can apply gauge variation to any of them, but there is a diagram automorphism switching the two vertices.  

We conclude that in the holomorphic Chern-Simons case
\begin{equation} 
		2 C^2 = \frac{1}{(2 \pi \i)^3 2 \cdot  4!  }   
\end{equation}
and in the self-dual Yang-Mills case we have 
\begin{equation} 
	2 C^2 = \frac{1}{(2 \pi \i)^3 4! }. 
\end{equation}

Now let us discuss how to determine $\lambda_{\g}$.  For all of the algebras we are considering, we know that the two sides of  equation \eqref{eqn_lambda} are proportional. To determine the constant, we only need to restrict $X$ to a particular value, which we can take to be a cocharacter.   If $w_a$ are the weights of the cocharacter $X$ in the adjoint representation, and $w_i$ are the weights in the fundamental representation, then
\begin{equation} 
	\lambda_{\g}^2 = \frac{\sum w_a^4}{ (\sum w_i^2)^2 }. 
\end{equation}

Let us do some sample computations.  Taking $\g = \mf{sl}_2$ and $X$ to be a generator of the Cartan, the weights of the fundamental representation are $1,-1$ and of the adjoint are $2,0,-2$.  Therefore
\begin{equation} 
	\lambda_{\sl_2}^2 = \frac{32 }{4   } = 8. 
\end{equation}
For $\mf{sl}_3$, we can can take $X$ to be an element of the Cartan of a copy of $\mf{sl}_2$. Then the non-zero weights of the fundamental representation are $1,-1$ and of the adjoint are $2,-2,1,-1,1,-1$.  Therefore,
\begin{equation} 
	\lambda_{\sl_3}^2 = \frac{  36   } { 4  } = 9.  
\end{equation}

Next, for $\mf{so}(8)$, let us decompose $\C^8$ into isotropic subspaces $\C^4_+ \oplus \C^4_-$. We can choose $X$ to be a cocharacter so that $\C^4_{\pm}$ have weights $\pm 1$.  Then, by construction, in the fundamental we have four weights each of $\pm 1$. In the adjoint, we have six weights each of $\pm 2$ and $16$ weight zero vectors. 

We conclude that
\begin{equation} 
	\lambda_{\so_8}^2 = \frac{ 96  } { 64  } = \frac{3} {2}  . 
\end{equation}

Putting this together, we find that (for instance) the coupling between the open and closed fields for $\g = \so(8)$ holomorphic Chern-Simons is normalized as 
\begin{equation} 
	\frac{\sqrt{3}}{ \sqrt{2} \sqrt{ (2 \pi \i)^3 4 \cdot  4! } }    \int \eta \op{tr}( \mc{A} \partial \mc{A} ) =  \frac{1}{8 (2 \pi \i)^{3/2} }   \int \eta \op{tr}( \mc{A} \partial \mc{A} )  .
\end{equation}

\section{Canceling the anomaly by a non-local term}
\label{sec:nonlocal}
In this appendix we  will verify that the expression
\begin{equation}
		\frac{1}{(2 \pi)^4 64  }   \int_{x \in \R^4} \int_{z,z' \in \CP^1} \op{Tr}_{adj} \left( \mc{A}(x,z) \d \mc{A}(x,z)\frac{\d z \d z'}{(z - z')^2}  \mc{A}(x,z') \d \mc{A}(x,z') \right)
\end{equation}
cancels anomaly on twistor space for either  holomorphic Chern-Simons theory, and the same expression multiplied by $2$ cancels the anomaly for self-dual Yang-Mills. The anomaly for holomorphic CS is
\begin{equation} 
		\frac{1}{(2 \pi\i)^3 48 }  \int_{\PT}\op{Tr}_{adj} \chi (\partial \mc{A})^3. 
\end{equation}

More formally, we can write the non-local vertex as follows. Consider the manifold
\begin{equation} 
	X = \PT \times_{\R^4} \PT 
\end{equation}
with coordinates $(x,z,z')$ where $x$ is a coordinate on $\R^4$ and $z,z'$ on $\CP^1$. The space $X$ consists of points in $\R^4$ with a pair of points $z,z'$ on the corresponding $\CP^1$ in $\PT$.  Let $\pi_1,\pi_2$ be the two projections of $X$ to $\PT$.

Since 
\begin{equation} 
	X = \R^4 \times \CP^1 \times \CP^1 
\end{equation}
we can give $X$ to natural complex structures, one by identifying $\R^4$ and the first copy of $\CP^1$ with twistor space, and one using the second copy of $\CP^1$. We let $\partial$, $\dbar$ be the operators on forms on $X$ associated to the first complex structure, where $\R^4$ is given a complex structure associated to $z' \in \CP^1$. We let $\partial'$, $\dbar'$ be the corresponding operators in the other complex structure. 

The one-loop anomaly for holomorphic Chern-Simons on twistor space is canceled by the  expression
\begin{equation} 
	\int_{X} \op{Tr}_{adj} \pi_1^\ast \left(  \mc{A} (x,z) \partial \mc{A}(x,z)\right)  \frac{\d (z-z')}{z - z'}\pi_2^\ast \left(  \partial \mc{A}(x,z') \mc{A}(x,z') \right)    + \dots 
\end{equation}
where $\dots$ indicates terms which are quintic and higher in $\mc{A}$ and which which involve $3$ or more points in $\CP^1_x$. These terms don't enter our discussion.   $\op{Tr}_{adj}$ indicates the trace in the adjoint representation, and $\pi_i^\ast$ indicates pull-back from one of the two copies of $\PT$. (In what follows we will drop the colour factors and the trace from the notation, as they don't play a role). 

We are interested in the gauge variation of this under linearized gauge transformations $\mc{A} \mapsto \mc{A} + \dbar \chi$.  Under such a transformation, the integral varies by
\begin{equation}
	\begin{split}
		& 2 \int_{X}    \partial \mc{A} (x,z) \dbar \chi (x,z)  \frac{\d z \d z'}{(z - z')^2} \partial'  \mc{A}(x,z') \mc{A}(x,z') \\
		&+2 \int_{X}    \partial \dbar \chi (x,z)  \mc{A} (x,z)  \frac{\d z \d z'}{(z - z')^2} \partial'  \mc{A}(x,z') \mc{A}(x,z')  
	\end{split}
\end{equation}
(We have used the symmetry between $z,z'$ here).  These two expressions are in fact the same, as we can use integration by parts, together with the fact that $\partial ( (z-z')^{-2} \d z \d z' ) = 0$, to equate them.  We thus find the gauge variation is
\begin{equation} 
	4 \int_{X}    \partial \mc{A} (x,z) \dbar \chi (x,z)  \frac{\d z \d z'}{(z - z')^2} \partial'  \mc{A}(x,z') \mc{A}(x,z').  
\end{equation}
Next, we can integrate by parts to move the operator $\dbar$ around. Since our fields are on-shell,  we can drop the term containing $\dbar \mc{A}$. This gives us
\begin{equation} 
	-4 \int_{X}    \partial \mc{A} (x,z) \chi (x,z) \dbar \left( \frac{\d z \d z'}{(z - z')^2}\right)  \partial'  \mc{A}(x,z') \mc{A}(x,z').  
\end{equation}
Next, we note that 
\begin{equation} 
	\dbar \frac{\d z \d z'} {(z-z')^2} = 2 \pi \i \partial_{z'} \d z \delta_{z' = z}. 
\end{equation}
(At this point, we could also have used a point-splitting regulator where $\abs{z-z'} = \eps$. If we did this, the integration by parts would produce a boundary term which is an integral over this counter, and we would find the same result).

Then, our gauge variation becomes
\begin{equation}
	\begin{split}
		-8 \pi \i \int_{\PT}&    \partial \mc{A} \chi  \d z \partial_z \left( \partial  \mc{A}  \mc{A}  \right) \\
		=& -8 \pi \i \int_{\PT}    \partial \mc{A} \chi  \d z (\partial_z  \partial  \mc{A} ) \mc{A}  \\
		&+ 8 \pi \i \int_{\PT}    \partial \mc{A} \chi    \partial  \mc{A} ) \d z \partial_z \mc{A} . 
	\end{split}
\end{equation}
The term with two copies of $\mc{A}$ which are not differentiated cancels. This is because $\mc{A} \wedge \mc{A}$ is anti-symmetric in the colour indices, but we are symmetrizing them. This means the gauge variation is
\begin{equation}
	\begin{split}
	 + 8 \pi \i \int_{\PT}&    \partial \mc{A} \chi    \partial  \mc{A} ) \d z \partial_z \mc{A} \\
		&= \frac{8 \pi \i}{3} \chi (\partial \mc{A})^3.
	\end{split}
\end{equation}
which is proportional to the anomaly.

\section{Quantization to of type I topological strings on twistor space}
\label{section_twistor_quantization}
In \cite{1905.09269} we showed that holomorphic Chern-Simons with gauge group $SO(8)$, coupled to type I Kodaira-Spencer theory, admits a canonically defined quantization on $\C^3$. Here we will generalize this result to twistor space with the boundary conditions discussed earlier, and also show that the quantum theory is compatible with the symmetries of twistor space.  
\begin{theorem} 
Holomorphic Chern-Simons theory with gauge group $SO(8)$ coupled to type I Kodaira- Spencer theory admits a canonically defined quantization at all orders in the loop expansion on $ \PT$, with these boundary conditions.

Further, the action of the group $GL_2 \ltimes \C^4$ on $\PT$ acts holomorphically on the theory, where the multiples of the identity in $GL_2$ scale the coupling constant.
\end{theorem}
This theorem is rigorous in the framework of \cite{costello2011renormalization, costello2021factorization} , and so in particular produces a factorization algebra on $\PT$.

The proof of this theorem is by a careful analysis of the cohomology groups which contain possible anomalies or counter-terms.  In \cite{1905.09269}, we found a surprising cancellation of anomalies between open and closed-string sectors. We will find that a version of this statement continues to hold when we have the boundary conditions.

\subsection{Green-Schwarz mechanism with boundary}
One can ask whether the Green-Schwarz anomaly cancellation continues to hold when we impose these boundary conditions. This is a slightly delicate question, as in other circumstances boundary conditions can introduce anomalies. 

We find that there are no extra anomalies coming from the boundary conditions, leading to the following result.
\begin{theorem} 
Holomorphic Chern-Simons theory with gauge group $SO(8)$ coupled to type I Kodaira- Spencer theory admits a canonically defined quantization at all orders in the loop expansion on $ \PT$, with these boundary conditions. 
\end{theorem}
This theorem is rigorous in the framework of \cite{costello2011renormalization,costello2021factorization}, and so in particular produces a factorization algebra on $\PT$.

The proof of this theorem is by a careful analysis of the cohomology groups which contain possible anomalies or counter-terms.  In \cite{1905.09269}, we found a surprising cancellation of anomalies between open and closed-string sectors. We will find that a version of this statement continues to hold when we have the boundary conditions. 

\begin{proof}
Since the construction of a quantization is a local question -- more precisely, the collection of possible quantizations forms a homotopy sheaf on space-time \cite{costello2011renormalization} -- we can analyze each boundary condition separately. 

In \cite{1905.09269}, we considered quantization not just of the $SO(8)$ holomorphic Chern-Simons theory, but more generally of the $OSp(8 + 2N \mid 2N)$ theory for varying $N$. More precisely, we framed the problem of constructing the quantum open-closed string theory with gauge group $OSp(8 + 2 N \mid 2 N)$ for all $N$, in a compatible way. The compatibility arises from the natural embeddings of $OSp(8+2M \mid M)$ into $OSp(8 + 2N \mid N)$ when $M < N$.  If we can construct the quantum system in this uniform way for all $N$, we can in particular restrict to a quantization for $SO(8)$.

It was helpful to phrase the problem in this uniform-in-$N$ way because trace relations at finite $N$ make the cohomology computations much more difficult.  If we work uniformly in $N$, trace relations go away. 

	We solved the uniform-in-$N$ quantization problem on $\C^3$ by obstruction theory. Let us sketch the argument.  

	We will work order-by-order in a slightly non-standard loop expansion.  For the purposes of this proof, what we will call loop number is determined as follows.  Take a connected Feynman diagram $\gamma$ of the open-closed theory. Since we are working with a string field theory, we can view it as being associated to the topological type of a unoriented connected surface, $\Sigma(\gamma)$. This is obtained in the standard way by viewing the closed-string vertex as a pair of pants and the closed-string propagator as a cylinder, and using the $SO$ version of the double line notation to interpret the open-string sector as a Riemann surface. 

	The closed-string external lines of the graph $\gamma$ are viewed as punctures on $\Sigma(\gamma)$, i.e.\ they have been removed from the surface.   If the surface $\Sigma(\gamma)$ has Euler characteristic $\chi$ (where the punctures are not considered part of the surface), then the loop number is $1 - \chi$.

	More prosaically, with this convention, a diagram with $k$ closed-string vertices, $l$ open-string vertices, $m$ open-closed vertices, $n$ open-string propagators and $p$ closed-string propagators has loop number $n-l + k +1$.  Note that the two terms in the Green-Schwartz cancellation appear at loop number one, because one has a pair of open-closed vertices connected by a closed-string propagator, and the other has four open-string vertices and four open-string propagators. 

The dimension of a Lagrangian $\mc{L}$ determines where it can appear in the loop expansion. Expressions of dimension $3(L-1)$ appear at $L$ loops.  Here, we are treating the closed-string field with its natural dimension when viewed as a $(2,1)$-form $\eta$.  If we do this, then the closed-string kinetic term $\int \dbar \eta \partial^{-1} \eta$ is dimensionless, as is the open-closed interaction $\eta A \partial A$. The closed-string interaction is of dimension $3$ and the entire open-string Lagrangian is of dimension $-3$.    

	Since we are working at large $N$, so that trace relations are not imposed, it makes sense to ask how many traces in the open-string fields are involved in any Lagrangian.  A Lagrangian with $k$ traces can appear as a counter-term or anomaly for diagrams $\gamma$ where the surface $\Sigma(\gamma)$ has $k$ boundary components. 

	The computation in \cite{1905.09269} was of the cohomology of the space of possible Lagrangians, where the differential is the BV anti-bracket with the classical Lagrangian.  The computation used a spectral sequence\footnote{A standard technique for computing cohomology where we write the differential as a sum of two terms one of which is diagonal in some basis, and the other is upper-triangular.  The diagonal term dominates and the first term in the spectral sequence is obtained by taking cohomology with respect to it.} where, in the first term, we only consider terms in the differential which preserve the sum of the number of traces plus the number of closed string fields.  We found that in dimension $> 0$ the cohomology vanished.  

	This cohomology has two gradings: the dimension $d$, and the sum of the number $m$ of traces with the number $n$ of closed-string fields.  It is important to note that only Lagrangians with $d - m - n \ge -1$  can contribute as anomalies or counter-terms.  This is a purely combinatorial fact. To see this, consider a connected Feynman diagram $\gamma$ and corresponding connected surface $\Sigma(\gamma)$, with $m$ boundaries, $n$ punctures, and $\chi(\Sigma(\gamma)) = 1-d$. Fill in all the punctures and boundaries with discs, to get $\Sigma'$. The Euler characteristic of $\Sigma'$ is $\chi(\Sigma') = \chi(\Sigma) + m + n$.  Since $\chi(\Sigma') \le 2$, we have
\begin{equation} 
	 m + n \le 1 + d 
\end{equation}

	Now let us consider the case where we have the three-form $z^{-2} \d v_1 \d v_2 \d z$, and we work near $z = 0$. 
	Let $\boldsymbol{\til{A}}=\mbf{A}/z$ be the super field for the open string fields. This is the natural open string field because $\mbf{A}$ vanishes at $z=0$.  We let $\boldsymbol{\eta}$ that for the closed string fields. The components of $\boldsymbol{\eta}$ are
	\begin{equation} 
		\boldsymbol{\eta}_{12} = \iota_{\partial_{v_1}} \iota_{\partial_{v_2}} \boldsymbol{\eta} \ \ \boldsymbol{\eta}_{iz} =   \iota_{\partial_{v_1}} \iota_{\partial_{v_2}} \boldsymbol{\eta}
	\end{equation}
	Note that $\boldsymbol{\eta}_{12}$ is of charge $0$ under rotation of $z$, whereas  $\boldsymbol{\til{A}}$ and $\boldsymbol{\eta}_{i z}$  are both  of charge $-1$.   

The Lagrangian takes the form
	\begin{equation} 
		\begin{split}		\tfrac{1}{2}	\int \boldsymbol{\til{A}} \dbar \boldsymbol{\til{A}} \d v_1 \d v_2 \d z + \tfrac{1}{6}	\int \boldsymbol{\til{A}}[\boldsymbol{\til{A}},\boldsymbol{\til{A}}] z \d v_1 \d v_2 \d z \\ + \int \boldsymbol{\eta}_{12} \boldsymbol{\til{A}} \partial_z (z \boldsymbol{\til{A}}) \d v_1 \d v_2 z \d z \\
			+ \int \boldsymbol{\eta}_{i z} \boldsymbol{\til{A}} \partial_{v_j}\boldsymbol{\til{A}} \d v_1 \d v_2 z^2 \d z \\
			+ \int \dbar \boldsymbol{\eta} \partial^{-1} \boldsymbol{\eta} \\
			+ \int \boldsymbol{\eta}_{12} \boldsymbol{\eta}_{1z} \boldsymbol{\eta}_{2z} \d v_1 \d v_2 z^2 \d z
		\end{split} \label{eqn_oc_lagrangian}
	\end{equation}
	Note that the open-string Lagrangian is of weight $-1$ under rotation of $z$ (where $z$ has weight $1$ and $\boldsymbol{\til{A}}$ has weight $-1$). The closed-string kinetic term and the terms which couple open and closed fields are of weight $0$. The closed string interaction is of weight $1$.

	Therefore, as on $\C^3$, a Feynman diagram of $L$ loops (using our modified notion of loop number) must be of charge $L$ under rotation of $z$.

Further, charge zero is only achieved when for the kind of diagrams that appear in the Green-Schwarz cancellation: one-loop diagrams built from open-string vertices and propagators, or diagrams where two open-closed vertices exchange a closed-string propagator.

	If we work near $z = 0$ but remove the divisor where $z = 0$,  we can perform the coordinate change $z \mapsto z^{-1}$ to bring us to the previous situation with volume form $\d v_1 \d v_2 \d z$.   The analysis from \cite{1905.09269} then applies, and we find that there are no possible Lagrangians of non-zero charge under rotation of $z$.  The Green-Schwarz anomaly is the only uncharged Lagrangian of ghost number $1$, and there are no uncharged Lagrangians of ghost number $0$ or less.

We need to show that the sub-complex consisting of Lagrangians where we include the defect has the same cohomology. (Recall that the cohomology of a subcomplex can be larger than that of the complex it is contained in).

	The analysis above means that   we need to compare the cochain complex of possible Lagrangians and counter-terms in the bulk-defect system (including $z = 0$) to those of the bulk system (with $z = 0$ removed), only in charges $\ge 0$. 

The computation of possible Lagrangians is a cohomological problem, and, as such, can be analyzed by means of a spectral sequence.  To describe the possible Lagrangians, we will first pass to the cohomology with respect to the BV bracket with the kinetic term in the actions.  This acts by $\dbar$ on all operators.  Passing to this cohomology means that only holomorphic derivatives can appear.  Further, all the $\d \zbar$, $\d \vbar_i$  in the Dolbeault complexes of the fields $\boldsymbol{A}$, $\boldsymbol{\eta}$ are simply wedged together. This is because  the terms involving $\partial_{\zbar}$, $\partial_{\vbar_i}$  or contraction with $\partial_{\zbar}$, $\partial_{\vbar_i}$ cancel in the cohomology of the $\dbar$ operator. 

Now consider  the most general Lagrangian $\mc{L}(\boldsymbol{A},\boldsymbol{\eta})$, after passing to $\dbar$ cohomology as above.    It takes the form 
	\begin{equation} 
		\int z^k \d z\d v_1 \d v_2  D_{1} \boldsymbol{\til{A}} \dots D_{l} \boldsymbol{\til{A}} D_{l+1} \boldsymbol{\eta} \dots D_{l+m} \boldsymbol{\eta} \label{eqn_general_lagrangian} 
	\end{equation}
	where the $D_i$ are holomorphic differential operators, that is polynomials in $\partial_{v_i}$ and $\partial_z$. Further, $k \ge 0$.

	If we include the defect at $z = 0$, then $k \ge 0$; removing this defect allows us to include Lagrangians with $k < 0$. 

Lagrangians of this form form a cochain complex, where the differential is the BV anti-bracket with the interacting terms in the Lagrangian we start with. Because we are writing this in terms of the fields $\boldsymbol{\til{A}}$, the interacting term for the open-string fields involves a $z$.  This term in the differential imposes gauge invariance (cohomologically).  The Dirichlet defect at $z = 0$ is present only in the way we have multiplied the interaction term in holomorphic Chern-Simons by $z$. While operators living at $z = 0$ may fail to be gauge invariant, some version of gauge invariance is still present for operators that are integrated over $z$.

We will show that in charge $\ge 0$, Lagrangians in the bulk and bulk-defect system are the same. Since the cohomology of the bulk complex vanishes in charge $> 0$, the same will hold for the bulk-defect complex.

	To see that the two complexes are the same in charge $> 0$, consider the general expression of equation \eqref{eqn_general_lagrangian}.  Suppose that this Lagrangian is of charge $\ge 0$. All expressions in the Lagrangian except $z^k \d z$ derivatives contribute charge $\le 0$.  Therefore, to have non-negative charge, we must have $k > 0$. 

In the bulk system, the cohomology vanishes for charge $> 0$. This argument implies that it vanishes for the bulk-defect system too.

Let us now consider the charge zero complex. Feynman diagrams which can contribute at charge $0$ are  one loop and have only open-string external lines. This means they are either single trace or double trace.  In the bulk system, there are two possible Lagrangians at ghost number one (i.e.\ these are potential anomalies). These are 
\begin{equation} 
	\int \op{tr} \left( \boldsymbol{A} \partial \boldsymbol{A}  \right) \op{tr}   \left( \partial  \boldsymbol{A} \partial \boldsymbol{A}  \right) \ \ \int \op{tr} \left( \boldsymbol{A} (\partial \boldsymbol{A})^3 \right). 
\end{equation}
Here trace is taken in the fundamental representation.  Both anomalies are coordinate independent, so that this expression makes sense in the neighbourhood of the defect at $z = 0$. The anomaly which involves two traces always cancels by the Green-Schwarz mechanism. The anomaly involving only one trace vanishes only for $SO(N+8 \mid N)$.

To complete the proof, we need to show that there are no extra anomalies that occur because of the boundary condition at $z = 0$.  But, as we have already seen, the charge zero complex including $z = 0$ is isomorphic to that when we remove this locus, so there are no extra anomalies (or possible one loop counter-terms).

We can apply this argument in a neighbourhood of $z=\infty$ as well, to construct a unique quantization of the coupled system in this region.  The two quantum systems on the overlapping region agree, because in this region we know the quantization is unique.  We conclude that there is a unique quantization on all of $\PT$.

Now let us prove that the action of the group $G = GL_2 \ltimes \C^4$, which is a symmetry of the classical system, persists at the quantum level.  This is a consequence of the cohomology vanishing.  Proving this statement carefully  requires a little work, however, because this symmetry group does not preserve the gauge-fixing condition.  We need to formulate carefully what it means to have a symmetry of a theory which does not preserve the gauge-fixing condition.

This was done in \cite{costello2011renormalization}.  The idea is to think of a family of theories over the space of gauge fixing conditions, which is locally constant.  This concept was explained in \cite{costello2011renormalization}, and we will review some details. There is a space (we view it as a simplicial set) of gauge fixing conditions $\mc{GF}$ whose points are K\"ahler metrics on $\PT$.  An $n$-simplex in this simplicial set is a smooth family of K\"ahler metrics on $\PT$ parameterized by the $n$-simplex $\tr^n$.  

In \cite{costello2011renormalization}, the concept of QFT over the base dg ring $\Omega^\ast(\mc{GF})$ was formulated.  Any such field theory gives rise to a factorization algebra on $\PT$ over the base ring $\Omega^\ast(\mc{GF})$. Passing to cohomology, it gives a locally-constant family of factorization algebras on $\PT$, parameterized by $\mc{GF}$.  Since $\mc{GF}$ is contractible, this is the same as a since factorization algebra on $\PT$.  

The fact that the family of factorization algebras parameterized by $\mc{GF}$ is locally constant means that, if we take a path between two gauge fixing conditions, we get an isomorphism between the factorization algebras obtained by quantizing the theories in the two gauges.   

The obstruction theory arguments of \cite{costello2011renormalization} imply that, since $\mc{GF}$ is contractible and has no cohomology, the cohomology groups controlling quantizations of the classical theory theories over the base ring $\Omega^\ast(\mc{GF})$ is the same as that of quantizing the theory with a fixed gauge. 

Therefore, there is a unique quantum theory over the base ring $\Omega^\ast(\mc{GF})$. In particular, this shows that the factorization algebra on $\PT$ built from the quantum theory is independent of the choice of gauge.

Next, let us turn to the symmetries.  The group $G = GL_2 \ltimes \C^4$ acts on the simplicial set $\mc{GF}$ of gauge fixing conditions.  Thus, it makes sense to ask for a quantum theory with an action of $G$: it means that we take a theory over the base ring $\Omega^\ast(\mc{GF})$ with a $G$-action covering the natural $G$-action on $\Omega^\ast(\mc{GF})$.  

Once we have formulated the notion of symmetry which does not preserve the gauge fixing condition, it is completely formal to prove that the classical symmetries extend to the quantum level. We will do this by obstruction theory. The only fact we will need is that the cohomology of the obstruction-deformation complex is zero in degrees $\le 0$. 

Let us denote by $\mc{T}$ the quantum theory over the base ring $\Omega^\ast(\mc{GF})$. Given an element $g \in G$, we can form a new theory $g^\ast \mc{T}$. We need to give an equivalence of theories $\mc{T} \simeq g^\ast \mc{T}$.  Since, classically, $G$ is a symmetry, $\mc{T}$ and $g^\ast \mc{T}$ are the same at the classical level, so that they are both quantizations of the classical theory.  

However, we have seen that there is a unique quantization, by cohomology vanishing.  It follows that there is an equivalence 
$\mc{T} \simeq g^\ast \mc{T}$.  (By equivalence we are using the fact proved in \cite{costello2011renormalization} that the collection of quantizations of a fixed classical theory forms a simplicial set in a natural way, where one-simplices are equivalences, two-simplices are equivalences between equivalences, and so on).

Since the cohomology groups of the obstruction-deformation complex vanish in degrees $\le 0$, the equivalence $\mc{T} \simeq g^\ast \mc{T}$ is unique up to contractible choice. It follows from this that if $g_1,g_2 \in G$, the two equivalences we get between $\mc{T}$ and $g_1^\ast g_2^\ast \mc{T}$ coincide.  

The family of theories $g^\ast \mc{T}$ depends smoothly on $g \in G$, and so defines a family of theories over the base ring $\cinfty(G) \what{\otimes} \Omega^\ast(\mc{GF})$.  Classically this family of theories is a trivial family over $G$. The obstruction argument above applies without any change to show that we can get an equivalence $g^\ast \mc{T} \sim \mc{T}$ of theories which varies smoothly with $g$.

This shows that we get a $G$-equivariant factorization algebra on $\PT$ from the quantum theory. The final step is to show that, at the infinitesimal level, the action of the anti-holomorphic part of the Lie algebra is homotopically trivial. This is true, because it is true classically and all obstructions to finding a trivialization at the quantum level also vanish. 
\end{proof}

\section{Calculation of logarithmic OPEs from twistor space}
\label{sec:twistor_ope}
In this appendix I will give a proof of propositions \ref{prop:log_wzw} and \ref{prop:ope_sdym}.

Both propositions are basically the same. We take points $u^{(i)}$ on $\R^4$, for $i = 1,2,3$ and $(0,1)$ forms on twistor space given by
\begin{equation}
\begin{split}
	\mc{A}^{(3)} &= \frac{1}{2 \pi \i} z \delta_{v_1 = u_1^{(3)} + z \ubar_2^{(3)} } \frac{1}{ v_2 - u_2^{(3)} + z \ubar_1{(3)}  }\\
	\mc{A}^{(1)} &= \frac{1}{2 \pi \i} z\delta_{v_2 = u_2^{(1)} - z \ubar_1^{(1)} } \frac{1}{ v_1 - u_1^{(1)} - z \ubar_1{(1)}  }\\
\end{split}
\end{equation}
These are the fields sourced by operators in either self-dual Yang-Mills or $WZW_4$ on $\R^4$.

We then insert these fields into two legs of the non-local vertex
\begin{equation}
			\frac{1}{ 8 \pi \i   (4 \pi)^4 }  \mc{A}^a(z) \partial \mc{A}^b (z) \frac{\d z \d z'}{(z - z')^2} \mc{A}^c(z') \partial \mc{A}^d(z') \op{Tr}( \t^{(a} \dots \t^{d)}). 
\end{equation}
We will insert the fields sourced by the operators into the legs which to not have a de Rham operator applied to them. This is legitimate, because, by integration by parts, we can move the operator $\partial$ from one of the two legs associated to coordinate $z$ to the other.  This choice adds an overall factor of $4$.

Thus, we need to compute (dropping the Lie algebra indices)
\begin{equation}
			\frac{1}{ 2 \pi \i   (4 \pi)^4 }  \mc{A}^{(1)}(z) \partial \mc{A} (z) \frac{\d z \d z'}{(z - z')^2} \mc{A}^{(3)}(z') \partial \mc{A}(z')
\end{equation}
Before we proceed further, we need to determine the background field $\mc{A}$.

Let us first discuss the difference between the computation in the case of $WZW_4$ or self-dual Yang-Mills.  Every on-shell field configuration of self-dual Yang-Mills can be realized from one of $WZW_4$ \cite{2011.04638}, by
\begin{equation} 
	A =  J^{0,1}. \label{eqn_sdym_wzw} 
\end{equation}
The equations of motion of $WZW_4$, which state that $\omega \wedge \d J^{0,1} = 0$, are equivalent to saying that this gauge field satisfies the self-dual Yang-Mills equation $F(A)_+ = 0$.

From the twistor space perspective, the fields of $WZW_4$ are, as we have discussed, $(0,1)$ forms which vanish at $z = 0$, $z = \infty$. The fields of self-dual Yang-Mills are simply $(0,1)$-forms.  We can turn the twistor representative of a field of $WZW_4$ into one of self-dual Yang-Mills by the natural inclusion from $(0,1)$ forms which vanish at $0,\infty$ into all $(0,1)$-forms. This map corresponds, in four dimensions, to the map \eqref{eqn_sdym_wzw}.  

This discussion tells us that we lose no generality by focusing on the $WZW_4$ case, and inserting a background field $\mc{A}(z)$ appropriate to the twistor representation of $WZW_4$. 

Let us now make and justify some simplifying assumptions which will allow us to write down an explicit field $\mc{A}(z)$.  For $WZW_4$,  we  are computing the OPE to $\phi = \log \sigma$: we will find that
\begin{equation} 
	\phi^a(0) \phi^b(x) \sim \hbar^2 \op{Tr} (\t^{(a} \dots t^{d)}) \partial_{u_i} \partial_{\ubar_j} \phi^c (0) \partial_{u_k} \partial_{\ubar_l} \phi^d(0) \eps^{ik} \eps^{jl} \log \norm{x} + \dots \label{eqn_ope_form} 
\end{equation}
where $\dots$ indicates terms of higher order in $\phi$, terms which are non-singular as $\norm{x} \to 0$, and terms with no branch cuts as $\norm{x} \to 0$. 

In this approximation, we need only compute $\d \mc{A}(z)$ to linear order in $\phi$.  

Similarly, for self-dual Yang-Mills, we will compute the output of the OPE to quadratic order in the gauge field $A$.  This is sufficient, because we know we must end up with a gauge invariant operator, and the cubic and higher order terms in $A$ are determined by gauge invariance. In this case, we can again write $A = \d \phi$ where $\phi^a$ is again harmonic. 

Since we are interested in an OPE of the form \eqref{eqn_ope_form}, we will be taking one $u$ and one $\ubar$ derivative of $\phi$. We can detect an OPE of this form by inserting a background field $\phi$ which is a quadratic function of $u,\ubar$ (linear in each variable).  We can assume that $\phi$ is of this form, and take for a twistor representative of $\phi$ the expression
\begin{equation} 
	\mc{A}(z) =  \frac{v_i v_j \abs{z}^2  \d \zbar }{(1 + \abs{z}^2)^4}\label{twistor_harmonic} 
\end{equation}

There are $3$ such expressions, which correspond to the three quadratic harmonic functions $u_1\ubar_2$, $u_2 \ubar_1$, $u_1 \ubar_1 - u_2 \ubar_2$.  We can detect the OPE we are interested in by taking the field at the vertex with coordinate $z$ to be the expression \eqref{twistor_harmonic} with $i,j = 1$, and that at coordinate $z'$ by taking the same expression with $i,j = 2$:
\begin{equation}
	\begin{split}
		\mc{A}(z) &=  \frac{v_1^2 \abs{z}^2  \d \zbar }{(1 + \abs{z}^2)^4}\\
	\mc{A}(z') &=  \frac{v_2^2 \abs{z'}^2  \d \zbar' }{(1 + \abs{z'}^2)^4}
	\end{split}
\end{equation}
This means that $\d \mc{A}(z)$ and $\d \mc{A}(z')$, which are the expressions that appear in our computation, are
\begin{equation} 
	 \begin{split}
		 \d \mc{A}(z) &= 2 (\d u_1 + z \d \ubar_2)  \frac{(u_1 + z \ubar_2) \abs{z}^2  \d \zbar }{(1 + \abs{z}^2)^4}\\
		 \mc{A}(z') &=  2 (\d u_2 - z' \d \ubar_1) \frac{(u_2 - z'\ubar_1) \abs{z'}^2  \d \zbar' }{(1 + \abs{z'}^2)^4}\\
	\end{split}
\end{equation}
Note that we have not included any terms involving $\d z$ and $\d z'$, as there is already a $\d z \d z'$ in the vertex.

With the values described above, we need to perform the integral
\begin{equation} 
	\int_{z,z',u^{(2)},\br{u}^{(2)}} \mc{A}^{(1)}(z) \partial \mc{A} (z) \frac{\d z \d z'}{(z - z')^2} \mc{A}^{(3)}(z') \partial \mc{A}(z')\label{eqn_local_vertex_integral} 
\end{equation}

Now let us turn to the computation. The non-local vertex is integrated over pairs of points in twistor space which both lie on the curve corresponding to $u^{(2)} \in \R^4$. We then integrate over the $u^{(2)}, \ubar^{(2)}$ variables.  Therefore we can set 
\begin{equation} 
	\begin{split} 
		v_1 &= u_1^{(2)} + z \ubar_2^{(2)} \\
		v_2 &= u_2^{(2)} - z \ubar_1^{(2)} \\
		v'_1 &= u_1^{(2)} + z' \ubar_2^{(2)} \\
		v'_2 &= u_2^{(2)} - z' \ubar_1^{(2)}  
	\end{split}
\end{equation}
where we have used $v,v'$ to indicate the coordinates at the two points in twistor space, where the third coordinate at each point is $z,z'$. 

This allows us to express the contribution from the fields $\mc{A}^{(1)}$,$\mc{A}{(3)}$ in terms of the $u$ coordinates. At the vertex with coordinates $z',v'$,  $\mc{A}^{(3)}$ contributes a $\delta$-function at 
\begin{equation} 
	\delta^{(2)}_{v_1' = u_1^{(3)} + z' \ubar_2^{(3)} } 
\end{equation}
This $\delta$-function is viewed as $(0,1)$ form, namely the function multiplied by $\d \vbar'_1$ (a priori we would also have a term involving $\d \zbar'$ but this drops out because $\mc{A}(z')$ contains a $\d \zbar'$). Further, it is a $\delta$-function enforcing two real, or one complex, constraint.   

When translating to the coordinates $u^{(2)}$, $z'$ it becomes
\begin{equation} 
	\frac{\d \ubar_1^{(2)} + \zbar' \d u_2^{(2)}}{\abs{u_2^{(23)}}^2}   \delta_{ z' = -\frac{ u_1^{(23)}  }{\ubar_2^{(23)}}  } 
\end{equation}
where $u_i^{(23)}= u_i^{(2)} - u_i^{(3)}$.  The factor of $\abs{\ubar_2^{(23)}}^{-2}$ comes from the Jacobian factor from the change of coordinates in the $\delta$-functions.

Similarly, associated to the vertex with coordinate $z$, we have the product of the delta functions $\delta_{v_2 = u_2^{(1)} - z \ubar_1^{(1)}}$ and $\delta_{v_2 = u_2^{(2)} - z \ubar_1^{(2)} }$.  This gives
\begin{equation} 
	\frac{\d \ubar^{(2)}_2 - \zbar \d u^{(2)}_1}{\abs{u_1^{(21)}}^2} \delta_{z = \frac{u_2^{(12)}}{\ubar_1^{(12)}}}  
 \end{equation}

Including all the terms gives in equation \eqref{eqn_local_vertex_integral} 
\begin{equation} 
	\begin{split}
		\int_{z,z',u^{(2)},\ubar^{(2)} }  \d z \d \zbar \d z' \d \zbar'
		&\delta_{z = \frac{u_2^{(12)}}{\ubar_1^{(12)}}}	\delta_{ z' = -\frac{ u_1^{(23)}  }{\ubar_2^{(23)}}  }
		\frac{zz' }{(z-z')^2}  \frac{1}{\abs{u_2^{(23)}}^2 \abs{u_1^{(21)}}^2} \\	
		&\frac{1}{ u_1^{(21)} + z \ubar_2^{(21)}}\frac{1}{ u_2^{(23)} - z' \ubar_1^{(23)}}\\
		& (u_1^{(2)} + z \ubar_2^{(2)})(u_2^{(2)} - z' \ubar_1^{(2)})\frac{\abs{z}^2 \abs{z'}^2}{(1 + \abs{z}^2)^4 (1 + \abs{z'}^2)^4} \\
		&(\d u^{(2)}_1 + z \d \ubar_2^{(2)} )(\d u^{(2)}_2 - z' \d \ubar_1^{(2)})	(\d \ubar^{(2)}_2 - \zbar \d u^{(2)}_1)(\d \ubar^{(2)}_1 + \zbar' \d u^{(2)}_2)\label{nonlocal_allterms} 
	\end{split} 
\end{equation} 
Let us simplify this expression a bit by noting that
\begin{equation}
	\begin{split}
	(\d u^{(2)}_1 + z \d \ubar_2^{(2)} )(\d u^{(2)}_2 - z' \d \ubar_1^{(2)})	(\d \ubar^{(2)}_2 - \zbar \d u^{(2)}_1)(\d \ubar^{(2)}_1 + \zbar' \d u^{(2)}_2) \\
	= -(1 + \abs{z}^2) (1 + \abs{z'}^2) \d u^{(2)}_1 \d \ubar_1^{(2)} \d u^{(2)}_2 \d \ubar_2^{(2)} 
	\end{split}
\end{equation}
Therefore we can replace the factors of $(1 + \abs{z}^2)^{-4}$ and $(1+\abs{z'}^2)^{-4}$ by $(1+\abs{z}^2)^{-3}$,$(1+\abs{z'}^2)^{-3}$ in \eqref{nonlocal_allterms}. 		

As $z$ is a $\delta$-function at $u_2^{(21)}/\ubar_1^{(21)}$ and $z'$ is a $\delta$-function at $-u_1^{23}/\ubar_2^{23}$, we simply have to insert these values for $z,z'$ in the expression above. Running through the factors one by one, we have
\begin{equation}
\begin{split}
	\frac{1}{\abs{u_2^{(23)}}^2 \abs{u_1^{(21)}}^2} &= \frac{1}{\abs{u_2^{(23)}}^2 \abs{u_1^{(21)}}^2} \\	
\frac{z z'}{(z-z')^2} &= \frac{1}{z - z'} \frac{1}{1/z'  - 1/z} \\
	&=\frac{ \ubar_2^{(23)} \ubar_1^{(12)} u_1^{(23)} u_2^{(12)} }{ (\ubar_2^{(23)} u_2^{(12)} + u_1^{(23)} \ubar_1^{(12)} )^2 } \\
	\frac{1}{ u_1^{(21)} + z \ubar_2^{(21)}}&=   \frac{\ubar_1^{(21)}}{\norm{u^{(12)}}^2} \\
	\frac{1}{ u_2^{(23)} - z' \ubar_1^{(23)}}	&=\frac{\ubar_2^{(23)}}{\norm{u^{(23)}}^2} \\
		\frac{\abs{z}^2}{(1 + \abs{z}^2)^3} &= \frac{ \abs{u_1^{(12)}}^4 \abs{u_2^{(12)}}^2} {\norm{u^{(12)}}^6}\\ 
	\frac{\abs{z'}^2}{(1+ \abs{z'}^2)^3} &= \frac{ \abs{u_2^{(23)}}^4\abs{u_1^{(23)} }^2 }{\norm{u^{(23)}}^6} \\ 
	u_1^{(2)} + z \ubar_2^{(2)} &=  \frac{ u_1^{(2)} \ubar_1^{(21)} + \ubar_2^{(2)} u_2^{(21)}  }{\ubar_1^{(21)}} \\
	u_2^{(2)} - z' \ubar_1^{(2)} &= \frac{u_2^{(2)} \ubar_2^{(23)} + \ubar_1^{(2)} u_1^{(23)} }{\ubar_2^{(23)}} 	
\end{split}
\end{equation}

Now let us further simplify the expression by choosing explicit values for $u^{(1)}$ and $u^{(3)}$, which are the positions of the two operators whose OPE we are studying. We set $u^{(1)} = 0$, $u_2^{(3)} = 0$, $u_1^{(3)} = c$ where $c$ is real.  We then drop the indices labeling which vertex, and let $v$ be the unit vector $(1,0,0,0)$, so $u^{(3)} = c v$.  Then, in the above expressions, we replace each occurrence of $u^{(2)}$ or $u^{(21)}$ by $u$, replace $u_2^{(23)}$ by $u_2$, $u_1^{(23)}$ by $u_1 - c$, and $\ubar_1^{(23)}$ by $\ubar_1 - c$.

The result is the product of 
\begin{equation}
\begin{split}
	\frac{1}{\abs{u_1}^2 \abs{u_2}^2}	\frac{ \ubar_2  \ubar_1  (u_1-c)u_2  }{ (\ubar_2  u_2  + (u_1 - c) \ubar_1  )^2 } 
   \frac{\ubar_1 }{\norm{u }^2} 
\frac{\ubar_2 }{\norm{u-c v}^2} 
 \frac{ \abs{u_1 }^4 \abs{u_2 }^2} {\norm{u }^6} 
 \frac{ \abs{u_2 }^4\abs{u_1 - c }^2 }{\norm{u - c v}^6}  
	\frac{ \norm{u}^2 }{\ubar_1 } 
	\frac{u_2  \ubar_2  + \ubar_1  (u_1- c) }{\ubar_2 } 
\end{split}
\end{equation}
which simplifies to
\begin{equation} 
	\begin{split}
\frac{ \ubar_1 (u_1-c)   \abs{u_1 - c }^2  \abs{u_1 }^2 \abs{u_2 }^6} {\norm{u }^6  \norm{u - c v}^8  (\ubar_2  u_2  + (u_1 - c) \ubar_1  )}  
	\end{split}
\end{equation}
Putting all these things together, we find we need to compute the integral 
\begin{equation} 	
	\int_{\norm{u} \le 1} \d^2 u_1 \d^2 u_2 \frac{ \ubar_1 (u_1-c)   \abs{u_1 - c }^2  \abs{u_1 }^2 \abs{u_2 }^6} {\norm{u }^6  \norm{u - c v}^8  (\ubar_2  u_2  + (u_1 - c) \ubar_1  )}  
\end{equation}
We have an IR cut-off where we integrate over the domain $\norm{u} \le 1$.  

The integral is tricky to compute exactly. However, we are only interested in the logarithmic singularities as $c \to 0$.  These we can compute by examining the logarithmic derivative of the integral:
\begin{equation} 	
		c \frac{\d}{\d c} \int_{\norm{u} \le 1} \d^2 u \d^2 \frac{ \ubar_1 (u_1-c)   \abs{u_1 - c }^2  \abs{u_1 }^2 \abs{u_2 }^6} {\norm{u }^6  \norm{u - c v}^8  (\ubar_2  u_2  + (u_1 - c) \ubar_1  )}  
\end{equation}
Note that the integrand is unchanged when we scale $c$ and $u,\ubar$ by the same factor. This means that
\begin{equation} 
	\begin{split}
		\int_{\norm{u} \le 1} c \frac{\d}{\d c} \d^2 u \d^2 \frac{ \ubar_1 (u_1-c)   \abs{u_1 - c }^2  \abs{u_1 }^2 \abs{u_2 }^6} {\norm{u }^6  \norm{u - c v}^8  (\ubar_2  u_2  + (u_1 - c) \ubar_1  )} \\ 
		= \int_{\norm{u} \le 1} \left( u_i \frac{\d}{\d u_i} +  \ubar_i \frac{\d}{\d \ubar_i}  \right)  \d^2 u_1 \d^2 u_2 \frac{ \ubar_1 (u_1-c)   \abs{u_1 - c }^2  \abs{u_1 }^2 \abs{u_2 }^6} {\norm{u }^6  \norm{u - c v}^8  (\ubar_2  u_2  + (u_1 - c) \ubar_1  )}   
	\end{split}
\end{equation}
Stokes' theorem then recasts the logarithmic derivative as
\begin{equation}
	\begin{split}
	\int_{\norm{u} = 1} (\eps^{ij} u_i \d u_j \d^2 \ubar + \eps^{ij} \d^2 u \ubar_i \d \ubar_j)     \frac{ \ubar_1 (u_1-c)   \abs{u_1 - c }^2  \abs{u_1 }^2 \abs{u_2 }^6} {\norm{u }^6  \norm{u - c v}^8  (\ubar_2  u_2  + (u_1 - c) \ubar_1  )}\\ 
		= \int_{\norm{u} = 1} \d Vol_{S^3}  \frac{ \ubar_1 (u_1-c)   \abs{u_1 - c }^2  \abs{u_1 }^2 \abs{u_2 }^6} {\norm{u }^6  \norm{u - c v}^8  (\ubar_2  u_2  + (u_1 - c) \ubar_1  )}  
	\end{split}
\end{equation}
The integral over the sphere $\norm{u} = 1$ can be expanded in series in $c$. Only the constant term will contribute a logarithm in the OPE. The other terms will give expressions that converge as $c \to 0$.  Sending $c \to 0$, we conclude that the logarithmic term in the OPE is
\begin{equation} 
	\begin{split}
		& \log c \lim_{c \to 0} \int_{\norm{u} =1} \d Vol_{S^3}  \frac{  \abs{u_1 }^6 \abs{u_2 }^6} {  \norm{u}^{16}  }   
	\\
		&	= \log c \int_{\norm{u} =1} \d Vol_{S^3} \abs{u_1}^6 \abs{u_2}^6 
	\end{split}
\end{equation}
which is clearly non-zero, as it is the integral of a non-negative function against the volume form on the sphere.

We conclude that there is a non-zero logarithmic divergence, as desired.

\section{Normalizing the two-point function}
We will work in coordinates $w_{\alpha \dot{\alpha}}$ so that
\begin{equation} 
	\begin{split} 
		w_{1\dot{1}} &= u_1\\
		w_{1 \dot{2}} &= u_2 \\
		w_{2 \dot{1}} &= \ubar_2\\
		w_{2 \dot{2}} &= -\ubar_1
	\end{split}
\end{equation}
We have
\begin{equation}
	\begin{split} 
	\d^4 w &= 4 \d^4 x \ \ \norm{w}^2 =   -\half{\eps}^{\alpha\beta} \eps{\dot{\alpha} \dot{\beta}} w_{\alpha \beta} w_{\dot{\alpha}} \dot{\beta}  \\   \Lap &= \partial_{x_i}^2 = -2   \eps_{\alpha \beta} \eps_{\dot{\alpha} \dot{\beta}} \partial_{w_{\alpha \dot{\alpha}}} \partial_{w_{\beta \dot{\beta}}} 
	\end{split}
\end{equation}
	Consider the field sourced by $B^{\alpha \beta}$.  This is a one-form $A^{\alpha \beta} $  which satisfies
\begin{equation} 
	(\d A^{\alpha \beta} )_+ =  \tfrac{1}{4} \eps_{\dot{\alpha} \dot{\beta}} \iota_{\partial_{\alpha \dot{\alpha}} } \iota_{\partial_{\beta \dot{\beta}}} \delta_{w = 0} 
\end{equation}
	where we view $\delta_{w = 0}$ as a $4$-form.

	Lowering the indices of $A^{\alpha \beta}$ using the $\eps$ tensor, and translating the $\delta$-function viewed as a four-form into that viewed as a scalar by using the volume form $ \d^4 x = \frac{1}{4} \d^4 w$, the equation becomes
	\begin{equation} 
		(	\d A_{\alpha \beta} )_+ =  \tfrac{1}{16} \d w_{\alpha \dot{\alpha}} \d w_{\beta \dot \beta} \eps^{\dot{\alpha} \dot{\beta}}  \delta_{w = 0} 	 
	\end{equation}

This can be solved by	
	\begin{equation} 
		A_{\alpha \beta} =  C \d^\ast  \d w_{\alpha \dot{\alpha}} \d w_{\beta \dot{\beta}} \eps^{\dot{\alpha} \dot{\beta} }  G(w) 
	\end{equation}
	where $G(w)$ is the Green's function, $C$ is a constant to be determined, and  $\d^\ast$ is normalized so that it takes the form
	\begin{equation} 
		\d^\ast =   \iota_{\partial_{w_{\alpha \dot{\alpha}} }} \partial_{w_{\beta \dot{\beta}}} \eps_{\alpha \beta} \eps_{\dot{\alpha} \dot{\beta}}  
	\end{equation}
	More precisely, the Green's function is 
	\begin{equation} 
		G(w) = \frac{1}{4 \pi^2 \norm{w}^2} =  \frac{1}{4 \pi^2 \norm{u}^2} = \frac{1}{4 \pi^2 \norm{x}^2}  
	\end{equation}
	It satisfies	
	\begin{equation} 
		\Lap G = \delta_{w = 0} / \d^4 x
	\end{equation}
		To determine the constant $C$, we take the indices $\alpha,\beta = 1$ so that
	\begin{equation}
		\begin{split} 
		A_{11} &= 2 C \d^\ast \d u_1 \d u_2 \frac{1}{4 \pi^2 \norm{u}^2} \\
			&= 2 C  (\d u_2 \partial_{\ubar_1} - \d u_1 \partial_{\ubar_2})   \frac{ 1  }{4 \pi^2 \norm{u}^2} .
		\end{split}	
	\end{equation}
	It is clear that the self-dual part of $\d A_{11}$ is  
	\begin{equation}
		\begin{split}
			2 C \d u_1 \d u_2 ( \partial_{u_1} \partial_{\ubar_1} + \partial_{u_2} \partial_{\ubar_2} ) \frac{1}{4 \pi^2 \norm{u}^2} 	&= \frac{1}{2} C \d u_1 \d u_2 \delta_{u,\ubar = 0} / \d^4 x\\
			& = \frac{1}{4} C \d w_{1 \dot{\alpha}} \d w_{1 \dot{\beta}} \eps^{\dot{\alpha} \dot{\beta}} \delta_{w = 0}/\d^4 x. 
		\end{split}
	\end{equation}
	Therefore, we should take $C = \frac{1}{4}$, and we find
	\begin{equation} 
		\begin{split} 
			A_{\alpha \beta} &=  -\frac{1}{16 \pi^2}  \d^\ast  \d w_{\alpha \dot{\alpha}} \d w_{\beta \dot{\beta}} \eps^{\dot{\alpha} \dot{\beta} } \frac{1}{\norm{w}^2 }\\ 
			&=	-\frac{1}{8 \pi^2}  \d^\ast  \d w_{\alpha \dot{\alpha}} \d w_{\beta \dot{\beta}} \eps^{\dot{\alpha} \dot{\beta} } \frac{1}{ \eps^{\alpha\beta} \eps{\dot{\alpha} \dot{\beta}} w_{\alpha \beta} w_{\dot{\alpha}} \dot{\beta} } \\
			&= \frac{1}{8 \pi^2  (\eps^{\alpha\beta} \eps{\dot{\alpha} \dot{\beta}} w_{\alpha \beta} w_{\dot{\alpha}} \dot{\beta})^2 }      \left( w_{\alpha \dot{\alpha}} \d  w_{\beta \dot{\beta}} \eps^{\dot{\alpha} \dot{\beta}} +    w_{\beta \dot{\alpha}} \d  w_{\alpha \dot{\beta}} \eps^{\dot{\alpha} \dot{\beta}}    \right)\\ 
			&= \frac{1}{32 \pi^2 \norm{w}^4}   \left( w_{\alpha \dot{\alpha}} \d  w_{\beta \dot{\beta}} \eps^{\dot{\alpha} \dot{\beta}} +    w_{\beta \dot{\alpha}} \d  w_{\alpha \dot{\beta}} \eps^{\dot{\alpha} \dot{\beta}}    \right) 
		\end{split}
	\end{equation}
	where $\norm{w}^2 = -\half \eps^{\alpha \beta} \eps^{\dot{\alpha}\dot{\beta}} w_{\alpha \dot{\alpha}} w_{\beta \dot{\beta}}$ as before.

From this we see that
	\begin{equation}
		\begin{split} 
			(\d A_{\alpha \beta} )_- = &-\frac{1}{16 \pi^2 \norm{w}^6} \eps^{\lambda \gamma} \eps^{\dot{\lambda} \dot{\gamma}}  w_{\lambda \dot{\lambda}} \d w_{\gamma \dot{\gamma}}   \left( w_{\alpha \dot{\alpha}} \d  w_{\beta \dot{\beta}} \eps^{\dot{\alpha} \dot{\beta}} +    w_{\beta \dot{\alpha}} \d  w_{\alpha \dot{\beta}} \eps^{\dot{\alpha} \dot{\beta}}    \right)   		
		\end{split}	
	\end{equation}
	Contracting this with $\frac{1}{4} \iota_{\partial_{\mu \dot{\mu}}} \iota_{\partial_{\nu \dot{\nu}}} \eps^{\mu \nu}$ gives us 
	\begin{equation} 
		\frac{1}{32 \pi^2 \norm{w}^6 } \left(  w_{\beta \dot{\lambda}} w_{\alpha \dot{\alpha}}  \eps^{\dot{\lambda} \dot{\mu}} \eps^{\dot{\alpha} \dot{\nu}}  + w_{\alpha \dot{\lambda}} w_{\beta \dot{\alpha}}  \eps^{\dot{\lambda} \dot{\mu}} \eps^{\dot{\alpha} \dot{\nu}}       \right)  
	\end{equation} 
We conclude that the two-point function is
	\begin{equation} 
		\ip{ F_{\dot{\alpha} \dot{\beta}} B_{\alpha \beta} } = 		 \frac{1}{32 \pi^2 \norm{w}^6 } \left(  w_{\beta \dot{\beta}} w_{\alpha \dot{\alpha}} + w_{\alpha \dot{\beta}} w_{\beta \dot{\alpha}} \right) 
	\end{equation}
For any two anti-self-dual two-forms $F_1$, $F_2$, we have
\begin{equation}
	\begin{split} 
		F_1 \wedge F_2 &= 4 \d^4 w \eps_{\dot{\alpha}_1 \dot{\beta}_1} \eps_{\dot{\alpha}_2 \dot{\gamma}_2} F_1^{\dot{\alpha}_1 \dot{\beta}_1} F_2^{\dot{\gamma}_1 \dot{\gamma}_2}     \\ 	
		&= 16 \d^4 x \eps_{\dot{\alpha}_1 \dot{\beta}_1} \eps_{\dot{\alpha}_2 \dot{\gamma}_2} F_1^{\dot{\alpha}_1 \dot{\beta}_1} F_2^{\dot{\gamma}_1 \dot{\gamma}_2}      	 
	\end{split}	
\end{equation}
Therefore,
\begin{equation}
	\begin{split} 
		F \wedge \d A_{\alpha \beta} &= 16 \d^4 x  F^{\dot{\alpha} \dot{\beta}} \ip{   \ip{ F_{\dot{\alpha} \dot{\beta}} B_{\alpha \beta} }}   \\
		&= \frac{1}{ \pi^2 \norm{w}^6 } \d^4 x  F^{\dot{\alpha} \dot{\beta}}     w_{\beta \dot{\beta}} w_{\alpha \dot{\alpha}} \\ 
		&=  2 \d^4 x  \eps_{\alpha \lambda} \eps_{\beta \mu} 	F_{\dot{\lambda} \dot{\mu}} \dpa{w_{\lambda \dot{\lambda}}}  \dpa{w_{\mu \dot{\mu}} }  G(w) 
\end{split}
\end{equation}

\bibliographystyle{JHEP}

\bibliography{final.bib}

\end{document}